\documentclass[11pt]{article}

\usepackage{amsfonts}
\usepackage{amssymb}
\usepackage{amstext}
\usepackage{amsmath}
\usepackage{mathtools}

\usepackage{amsthm} 

\usepackage{caption}
\usepackage{subcaption}

\usepackage{color}
\usepackage{nameref}
\definecolor{ForestGreen}{rgb}{0.1333,0.5451,0.1333}
\definecolor{DarkRed}{rgb}{0.8,0,0}
\definecolor{Red}{rgb}{0.9,0,0}
\usepackage[linktocpage=true,
	pagebackref=true,colorlinks,
	linkcolor=DarkRed,citecolor=ForestGreen,
	bookmarks,bookmarksopen,bookmarksnumbered]
	{hyperref}
\usepackage{cleveref}

\usepackage{thmtools,thm-restate} 

\usepackage[numbers,sort&compress]{natbib}

\usepackage{graphicx}
\usepackage{graphics}
\usepackage{colordvi}
\usepackage{xspace}
\usepackage{algorithm}
\usepackage{algorithmicx}
\usepackage{url}
\usepackage{enumitem}

        {\hspace*{\fill}$\Box$\par\vspace{4mm}}

\usepackage[left=1in,top=1in,right=1in,bottom=1in]{geometry} 
\usepackage{booktabs}
\usepackage{threeparttable}

\makeatletter
\renewcommand{\paragraph}{%
  \@startsection{paragraph}{4}%
  {\z@}{1ex \@plus 1ex \@minus .2ex}{-1em}%
  {\normalfont\normalsize\bfseries}%
}
\makeatother

\makeatletter
\def\thmt@refnamewithcomma #1#2#3,#4,#5\@nil{%
  \@xa\def\csname\thmt@envname #1utorefname\endcsname{#3}%
  \ifcsname #2refname\endcsname
    \csname #2refname\expandafter\endcsname\expandafter{\thmt@envname}{#3}{#4}%
  \fi
}
\makeatother

\declaretheorem[numberwithin=section,refname={Theorem,Theorems},Refname={Theorem,Theorems}]{theorem}
\declaretheorem[numberlike=theorem,refname={Lemma,Lemmas},Refname={Lemma,Lemmas}]{lemma}

\declaretheorem[numberlike=theorem,refname={Corollary,Corollaries},Refname={Corollary,Corollaries}]{corollary}

\declaretheorem[numberlike=theorem,refname={Claim, Claims},Refname={Claim, Claims}]{claim}

\declaretheorem[numberlike=theorem]{definition}

\renewcommand{\phi}{\varphi}

\newcommand{\poly}{\operatorname{poly}}

\newcommand{\NN}{\mathcal{N}}

\newcommand{\QValue}{{\sc QueryValue}\xspace}

\newcommand{\N}{\mathcal{N}}
\newcommand{\B}{\mathcal{B}}

\newcommand{\eps}{\epsilon}
\newcommand{\dl}{\delta}
\newcommand{\dd}{D}

\ifdefined\ShowComment

\def\danupon#1{\marginpar{$\leftarrow$\fbox{D}}\footnote{$\Rightarrow$~{\sf #1 --Danupon}}}
\def\babis#1{\marginpar{$\leftarrow$\fbox{B}}\footnote{$\Rightarrow$~{\sf #1 --Babis}}}
\def\sayan#1{\marginpar{$\leftarrow$\fbox{S}}\footnote{$\Rightarrow$~{\sf #1 --Sayan}}}
\def\monika#1{\marginpar{$\leftarrow$\fbox{M}}\footnote{$\Rightarrow$~{\sf #1 --Monika}}}

\else

\def\danupon#1{}
\def\babis#1{}
\def\sayan#1{}
\def\monika#1{}

\fi

\newboolean{short}
\setboolean{short}{true} 

\newcommand{\shortOnly}[1]{\ifthenelse{\boolean{short}}{#1}{}}
\newcommand{\longOnly}[1]{\ifthenelse{\boolean{short}}{}{#1}}

\title{Space- and Time-Efficient Algorithm for Maintaining\\ Dense Subgraphs on One-Pass Dynamic Streams\thanks{A preliminary version of this paper appeared in the 47th ACM Symposium on Theory of Computing (STOC 2015).}}

\date{}

\author{ 
	Sayan Bhattacharya\thanks{The Institute of Mathematical Sciences, Chennai, India. Part of this work was done while the author was in Faculty of Computer Science, University of Vienna, Austria.}
	\and Monika Henzinger\thanks{Faculty of Computer Science, University of Vienna, Austria. The research leading to these results has received funding from the European Unions Seventh Framework Programme (FP7/2007-2013) under grant agreement 317532 and from
	the European Research Council under the European Union's Seventh Framework Programme (FP7/2007-2013) / ERC Grant Agreement number 340506.}       
	\and Danupon Nanongkai\thanks{KTH Royal Institute of Technology, Sweden. Part of this work was done while the author was in Faculty of Computer Science, University of Vienna, Austria.}      
    \and Charalampos E. Tsourakakis\thanks{Harvard University,  School of Engineering and Applied Sciences.}    
}

\begin{document}


\maketitle
\pagenumbering{roman}
\begin{abstract}
While in many graph mining applications it is crucial to handle a stream of updates efficiently in terms of {\em both} time and space, not much\danupon{``not much -- if anything --'' is now ``not much''} was known about achieving such type of algorithm. In this paper we study this issue for a problem which lies at the core of many graph mining applications called {\em densest subgraph problem}. We develop an algorithm that achieves time- and space-efficiency for this problem simultaneously. It is one of\danupon{``one of''} the first of its kind for graph problems to the best of our knowledge.

Given an input graph $G = (V, E)$, the ``density'' of a  subgraph induced by a subset of nodes $S \subseteq V$ is  defined as $|E(S)|/|S|$, where $E(S)$ denotes the set of edges in $E$ with both endpoints in $S$. In the densest subgraph problem, the goal is to find a subset of nodes that maximizes the density of the corresponding induced subgraph.  

For any $\epsilon>0$, we present a dynamic algorithm that, with high probability, maintains a $(4+\epsilon)$-approximate solution for the densest subgraph problem under a sequence of edge insertions and deletions  in an input graph with $n$ nodes. The algorithm uses $\tilde O(n)$ space, and has an amortized update time of $\tilde O(1)$ and a query time of $\tilde O(1)$. Here,  $\tilde O$ hides a $O(\poly\log_{1+\epsilon} n)$ term. The approximation ratio can be improved to $(2+\epsilon)$ at the cost of increasing the query time to $\tilde O(n)$. It can be extended to a $(2+\epsilon)$-approximation sublinear-time algorithm and a distributed-streaming algorithm.  Our algorithm is the first streaming algorithm that can maintain the densest subgraph in {\em one pass}. Prior to this, no algorithm could do so even in the special case of an incremental stream and even when there is no time restriction. The previously best algorithm in this setting required $O(\log n)$ passes [Bahmani, Kumar and Vassilvitskii, VLDB'12]. The space required by our algorithm is tight up to a polylogarithmic factor.

\end{abstract}

\newpage
\setcounter{tocdepth}{3}
\tableofcontents

\newpage

\pagenumbering{arabic}


\newpage



\renewcommand{\O}{\tilde{O}}

\section{Introduction}\label{sec:intro}

In analyzing large-scale rapidly-changing graphs, it is crucial that algorithms must use small space and adapt to the change quickly.
This is the main subject of interest in at least two areas, namely {\em data streams} and {\em dynamic algorithms}. In the context of graph problems, both areas are interested in maintaining some graph property, such as connectivity or distances, for graphs undergoing a stream of edge insertions and deletions. This is known as the (one-pass) {\em  dynamic semi-streaming} model in the data streams community, and as the {\em fully-dynamic} model in the dynamic algorithm community. 
 
The two areas have been actively studied since at least the early 80s (e.g. \cite{EvenS81,MunroP80}) and have produced several sophisticated techniques for achieving time and space efficiency. In dynamic algorithms, where the primary concern is {\em time}, the heavy use of {\em amortized analysis} has led to several extremely fast algorithms that can process updates and answer queries in  poly-logarithmic amortized time.
In data streams, where the primary concern is {\em space}, the heavy use of {\em sampling} techniques to maintain small {\em sketches} has led to algorithms that require space significantly less than the input size; in particular, for dynamic graph streams the result by Ahn, Guha, and McGregor \cite{AhnGM12SODA} has demonstrated the power of  linear graph sketches in the dynamic model, 
and initiated an extensive study of dynamic graph streams (e.g. \cite{KapralovLMMS14,KapralovW14,AhnGM12SODA,AhnGM12PODS,AhnGM13}).
Despite  numerous successes in these two areas, we are not aware of many results that  combine techniques from {\em both} areas to achieve time- and space-efficiency  {\em simultaneously} in dynamic graph streams. A notable exception  we are aware of is the connectivity problem, where one can combine the space-efficient streaming algorithm of Ahn~et~al.~\cite{AhnGM12PODS} with the fully-dynamic algorithm of Kapron~et~al.~\cite{KapronKM13}\footnote{We thank Valerie King (private communication) for pointing out this fact.}.

\subsection{Problem definition}
\label{sec:intro:problem}

In this paper, we study  the {\em densest subgraph} problem in dynamic and streaming setting. Fix any unweighted undirected input graph $G = (V, E)$. The density of a subgraph induced by the set of nodes $H \subseteq V$ is defined as $\rho(H) = |E(H)|/|H|$, where $E(H) = \{ (u,v) \in E : u, v \in H \}$ is the set of edges in the induced subgraph. The densest subgraph of $G$ is the subgraph induced by a  node set $H \subseteq V$ that maximizes $\rho(H)$, and we denote the density of such a subgraph by $\rho^*(G)=\max\limits_{H\subseteq V} \rho(H)$. For any $\gamma \geq 1$ and $\eta$, we say that $\eta$ is an {\em $\gamma$-approximate value} of $\rho^*(G)$ if $\rho^*(G)/\gamma \leq \eta \leq \rho^*(G)$. The (static) densest subgraph problem is to compute or approximate $\rho^*(G)$ and the corresponding subgraph. Throughout the paper, we use $n = |V|$ and $m = |E|$ to denote the number of nodes and edges in the input graph, respectively.  
 
This problem and its variants have been intensively studied in practical areas as it is an important primitive in analyzing massive graphs. Its applications range from identifying dense communities in social networks (e.g. \cite{DourisboureGP07}), link spam detection (e.g. \cite{DGibsonKT05}) and finding stories and events (e.g. \cite{AngelKSS12}); 
for many more applications of this problem see, e.g., \cite{BahmaniKV12,LeeRJA10,Tsourakakis14,TangL10}. 
Goldberg \cite{Goldberg84} was one of the first to study this problem although the notion of graph density has been around much earlier (e.g. \cite[Chapter~4]{Lawler:combopt}). His algorithm can solve this problem in polynomial time by using $O(\log n)$ flow computations. Later Gallo, Grigoriadis and Tarjan slightly improved the running time using parametric maximum flow computation. These algorithms are, however, not very practical, and an algorithm that is more popular in practice is an $O(m)$-time $O(m)$-space $2$-approximation algorithm of Charikar \cite{Charikar00}.
However, as mentioned earlier, graphs arising in modern applications are huge and keep changing, and the earlier algorithms cannot handle edge insertions/deletions in the input graph. 
Consider, for example, an application of detecting a dense community in social networks. Since people can make new friends as well as ``unfriend'' their old friends, the algorithm must be able to process these updates efficiently.
With this motivation, it is natural to consider a dynamic version of this problem as defined below.

\paragraph{Our Model.} We start with an empty graph $G = (V, E)$ where $E = \emptyset$. Subsequently, at each time-step, an adversary  either  inserts an edge into the graph, or deletes an already existing edge  from the graph. The set of nodes, on the other hand, remain unchanged. The goal is to maintain a good approximation to the value of the densest subgraph while processing this sequence of edge insertions/deletions.  More formally, we want to design a data structure for the input graph $G = (V, E)$ that supports the following operations.
\begin{itemize}
\item {\sc Initialize}($V$): Initialize the data structure with an empty graph $G = (V, E)$ where $E = \emptyset$.
\item {\sc Insert}($u,v$): Insert the edge $(u,v)$, where $u, v \in V$, into the graph $G$.
\item {\sc Delete}($u,v$): Delete an existing edge $(u,v) \in E$ from the graph $G$.
\item {\sc QueryValue}: Return an estimate of the value of the maximum density $\rho^*(G) = \max_{S \subseteq V} \rho(S)$. If this estimate is always within a $\gamma$-factor of $\rho^*(G)$, for some $\gamma \geq 1$, then we say that the algorithm maintains a $\gamma$-approximation to the value of the densest subgraph. We want this approximation factor to be a small constant. 
\end{itemize}
\noindent The performance of a data structure is measured in term of four different metrics, as defined below.
\begin{itemize}
\item {\em Space-complexity:} This is given by the total space (in terms of bits) used by the data structure.
\item {\em Update-time:} This is the time taken to handle an {\sc Insert} or {\sc Delete} operation.
\item {\em Query-time:} This is the time taken to handle a {\sc QueryValue} operation. 
\item {\em Preprocessing-time:} This is the time taken to handle the {\sc Initialize} operation. Unless explicitly mentioned otherwise, in this paper the preprocessing time  will always be $\O(n)$. 
\end{itemize}

\paragraph{Comparison with the semi-streaming model.} In the streaming algorithms literature,  the ``{\em semi-streaming model}'' for graph problems is defined as follows. We start with an empty graph of $n$ nodes. Subsequently, we have to process a ``{\em stream}'' of updates in the graph. For ``{\em insert-only}'' streams, each update consist of inserting a new edge into the graph. For ``{\em dynamic}'' (or, ``turnstile'') streams, each update consists of either inserting a new edge into the graph or deleting an already existing edge from the graph. 

A ``{\em semi-streaming algorithm}'' can use only $\O(n)$ bits of space while processing a stream of updates. In particular, the algorithm cannot store all the edges in the graph (which might require $\Omega(n^2)$ space). At the end of the stream, the algorithm has to output an (approximate) solution to the problem concerned, which, in our case, happens to be the value of the densest subgraph. The algorithm is  allowed to make ``multiple passes'' over this stream. Typically, in the streaming algorithms literature the focus is on the space complexity,  and optimizing  the update-time and the query-time (which can be as large as $\Omega(n)$) are of secondary importance. 

\paragraph{Our goal.} We want to design  algorithms that maintain  constant factor approximations to the value  of the densest subgraph in a dynamic setting,  have very fast (polylogarithmic in $n$) update and query times, and use very little (near-linear in $n$) space. In other words, we want single-pass semi-streaming algorithms over dynamic streams with polylogarithmic update and query times. 


\paragraph{Remark on the query operation.}
The \QValue operation described above asks only for an estimate of the value $\rho^*(G)$. This raises a natural question: Can we answer a more general query that asks for a subset of nodes which constitute an approximate densest subgraph (in time proportional to the number of nodes returned in response to the query)? The answer is yes. We can easily extend all the algorithms presented in this paper so as  to enable them with this new feature  (see the discussion after the proof of Corollary~\ref{cor:test:1}).

\subsection{Our Results} Our main result is an efficient $(4+\epsilon)$-approximation algorithm for this problem (see Theorem~\ref{main:th:dynamic:stream:main}). To be more specific, we present a randomized algorithm that can process a stream of polynomially many edge insertions/deletions starting from an empty graph using only $\O(n)$ space, and with high probability, the algorithm maintains a $(4+\epsilon)$-approximation to the value of the densest subgraph. The algorithm has $\O(1)$ amortized update-time and $\O(1)$ query-time.

For every integer $t \geq 0$, let $G^{(t)} = (V, E^{(t)})$ be the state of the input graph $G = (V,E)$ just after we have processed the first $t$ updates (edge insertions/deletions) in the dynamic stream, and define  $m^{(t)} \leftarrow |E^{(t)}|$. Thus, we have $m^{(0)} = 0$ and $m^{(t)} \geq 0$ for all $t \geq 1$.
We let $\text{{\sc Opt}}^{(t)} = \rho^*(G^{(t)})$ denote the density of the densest subgraph in $G^{(t)}$. 

\paragraph{Notation.} Throughout  this paper, the notations $\O(.)$ and $\tilde \Theta(.)$ will hide $\text{poly}(\log n, 1/\epsilon)$ factors in the running times and space complexities  of our algorithms, where $\epsilon \in (0,1)$ is a small constant.

\begin{theorem}
\label{main:th:dynamic:stream:main}
Fix a  small constant $\epsilon \in (0,1)$, a constant $\lambda > 1$, and let $T = \lceil n^{\lambda} \rceil$. We can process the first $T$ updates (edge insertions/deletions) in a dynamic stream using $\tilde O(n)$  space, and maintain a value $\text{{\sc Output}}^{(t)}$ at each $t \in [T]$. The algorithm gives the following  guarantees with high probability: We have $\text{{\sc Opt}}^{(t)}/(4+O(\epsilon)) \leq \text{{\sc Output}}^{(t)} \leq \text{{\sc Opt}}^{(t)}$ for all $t \in [T]$. Further, the total amount of computation performed while processing the first $T$ updates in the  stream is $\O(T)$.
\end{theorem}

\paragraph{Oblivious Adversary.} We remark that  Theorem~\ref{main:th:dynamic:stream:main} holds only when the sequence of edge insertions/deletions in the input graph does not depend on the random bits used by our algorithm. In other words, the ``adversary'', who decides upon the sequence of edge insertions/deletions, is ``oblivious'' to  the random bits used in the algorithm. This is a standard assumption in the graph streaming literature. For example, the paper by Ahn, Guha and  McGregor~\cite{AhnGM12SODA}  also requires this assumption on the adversary. We prove Theorem~\ref{main:th:dynamic:stream:main} in Section~\ref{sec:combine}.   In addition, we  obtain the following results. 


\paragraph{$\bullet$ A $(2+\epsilon)$-approximation one-pass dynamic semi-streaming algorithm:} This follows from the fact that with the same space, preprocessing time, and update time, and an additional $\tilde O(n)$ query time, our main algorithm can output a $(2+\epsilon)$-approximate solution. See Section~\ref{sec:sketch}.

\paragraph{$\bullet$ A $(4+\epsilon)$-approximation deterministic dynamic algorithm with $\O(1)$ update time.} In Section~\ref{sec:dynamic}, we present a deterministic  algorithm that maintains a $(4+\epsilon)$-approximation to the value of the densest subgraph. This  requires $\O(m+n)$ space, and  $\O(1)$ update and query times. 

\paragraph{$\bullet$ Extensions to directed graphs.} In Section~\ref{sec:directed}, we extend our result from Section~\ref{sec:dynamic} to directed graphs. Specifically, we present a deterministic dynamic algorithm that maintains a $(8+\epsilon)$-approximation to the value of the densest subgraph of a directed graph. This requires $\O(m+n)$ space, and  $\O(1)$ update and query times. 

\paragraph{$\bullet$ Sublinear-time algorithm:} We show that Charikar's linear-time linear-space algorithm~\cite{Charikar00} can be improved further. In particular, if the graph is represented by an incident list (this is a standard representation \cite{ChazelleRT05,GoelKK13}), our algorithm needs to read only $\tilde O(n)$ edges in the graph (even if the graph is dense) and requires $\tilde O(n)$ time to output a $(2+\eps)$-approximate solution. We also provide a lower bound that matches this running time up to a poly-logarithmic factor. See Section~\ref{sec:sublinear}.

\paragraph{$\bullet$ Distributed streaming algorithm:} In the distributed streaming setting with $k$ sites as defined in \cite{CormodeMYZ10}, we can compute a $(2+\epsilon)$-approximate solution with  $\tilde O(k+n)$ communication by employing the algorithm of Cormode~et~al.~\cite{CormodeMYZ10}. See Section~\ref{sec:distributed}.

\subsection{Previous work}
To the best of our knowledge, our main algorithm is the first dynamic graph algorithm that requires $\tilde O(n)$ space  and at the same time can quickly process each update and answer each query for densest subgraph. 
Previously, there was no space-efficient algorithm known for this problem, even when time efficiency is not a concern, and even for insert-only streams. In this insert-only model, Bahmani, Kumar, and Vassilvitskii~\cite{BahmaniKV12} provided a deterministic $(2+\epsilon)$-approximation $O(n)$-space algorithm.  
Their algorithm   needs $O(\log_{1+\epsilon} n)$ passes; i.e., it has to read through the sequence of edge insertions $O(\log_{1+\epsilon} n)$ times. 
(Their algorithm was also extended to a MapReduce algorithm, which was later improved by \cite{BahmaniGM14}.)
In Section~\ref{sec:sketch}, we improve this result of Bahmani et al. in two respects: (a) We can process a dynamic stream of updates, and (b) we need only a single pass. 
%
Further, the space usage of our algorithm from Section~\ref{sec:sketch} matches the lower bound provided by \cite[Lemma~7]{BahmaniKV12} up to a polylogarithmic factor.
%

We note that while in some settings it is reasonable to compute the solution at the end of the stream or even make multiple passes (e.g. when the graph is kept on an external memory), and thus our and Bahmani~et~al's $(2+\epsilon)$-approximation algorithms are sufficient in these settings, there are many natural settings where the stream keeps changing, e.g.  social networks where users keep making new friends and disconnecting from old friends. In the latter case our main algorithm is necessary since it can quickly prepare to answer the densest subgraph query after every update. 

Another related result in the streaming setting is by Ahn et al. \cite{AhnGM12PODS} which approximates the fraction of some dense subgraphs such as a small clique in dynamic streams. This algorithm does not solve the densest subgraph problem but might be useful for similar applications. 

Not much was known about time-efficient algorithm for this problem even when space efficiency is not a concern. 
%
One possibility is to adapt dynamic algorithms for the related problem called {\em dynamic arboricity}. The arboricity of a graph $G$ is $\alpha(G)=\max_{U\subseteq V(G)} |E(U)|/(|U|-1)$ where $E(U)$ is the set  of edges of $G$ that belong to the subgraph induced by $U$. Observe that $\rho^*(G) \leq \alpha(G)\leq 2\rho^*(G)$. Thus, a $\gamma$-approximation algorithm for the arboricity problem will be a $(2\gamma)$-approximation algorithm for densest subgraph. In particular, we can use the $4$-approximation algorithm of Brodal and Fagerberg \cite{BrodalF99} to maintain an $8$-approximate solution to the densest subgraph problem in $\tilde O(1)$ amortized update time. (With a little more thought, one can in fact improve the approximation ratio to $6$.)

In a work that appeared at about the same time as the preliminary version of this paper, Epasto~et~al. \cite{EpastoLS15} presented a $(2+\epsilon)$-approximation algorithm for densest subgraph which can handle arbitrary edge insertions and random edge deletions. After the preliminary version of our paper appeared, Esfandiari~et~al.~\cite{recent} and McGregor~et~al.~\cite{McGregorTVV15} presented semi-streaming algorithms for densest subgraph that give $(1+\epsilon)$-approximation and require $\O(n)$ space. The same result was obtained independently by Mitzenmacher et al. \cite{mitzenmacher2015scalable}. These improve the approximation ratio of our $(2+\epsilon)$-approximation semi-streaming algorithm. Like our $(2+\epsilon)$-approximation algorithm, their algorithms have an update-time of $\O(1)$, but the query-time can be as large as $\tilde \Omega(n)$.

\subsection{Overview of our techniques}  An intuitive way to combine techniques from data streams and dynamic algorithms for any problem is to run the dynamic algorithm using the sketch produced by the streaming algorithm as an input. This idea does not work straightforwardly. The first obvious issue is that the streaming algorithm might take excessively long time to maintain its sketch and the dynamic algorithm might require an excessively large additional space. A more subtle issue is that  the sketch might need to be processed in a specific way to recover a solution, and the dynamic algorithm might not be able to facilitate this. As an extreme example, imagine that the sketch for our problem is not even a graph; in this case, we cannot even feed this sketch to a dynamic algorithm as an input. 

The key idea that allows us to get around this difficulty is to develop streaming and dynamic algorithms based on the same structure called {\em ($\alpha$, $d$, $L$)-decomposition}. This structure is an extension of a concept called {\em $d$-core}, which was studied in graph theory since at least the 60s (e.g., \cite{Erdos46,Matula68,Szekeres68}) and has played an important role in the studies of the densest subgraph problem (e.g., \cite{BahmaniKV12,SariyuceGJWC13}).
The $d$-core of a graph is its (unique) largest induced subgraph with every node having degree at least $d$. It can be computed by repeatedly removing nodes of degree less than $d$ from the graph, and can be used to $2$-approximate the densest subgraph.  
Our  ($\alpha,  d, L$)-decomposition with parameter $\alpha \geq 1$ is an approximate version of this process where we repeatedly remove nodes of degree ``approximately'' less than $d$: in this decomposition we must remove all nodes of degree less than $d$ and are allowed to remove {\em some} nodes of degree between $d$ and $\alpha d$. We will repeat this process for $L$ iterations. 
Note that  the ($\alpha,  d, L$)-decomposition of a graph is not unique. However, for $L=O(\log_{1+\epsilon} n)$, an ($\alpha,  d, L$)-decomposition can be use to $2\alpha(1+\epsilon)^2$-approximate the densest subgraph. We explain this concept in detail in Section \ref{sec:decomposition}. 

We show that this concept can be used to obtain an approximate solution to the densest subgraph problem and leads to both a streaming algorithm with a small sketch and a dynamic algorithm with small amortized update time. 
In particular, it is intuitive that to check if a node has degree approximately $d$, it suffices to sample every edge with probability roughly $1/d$. The value of $d$ that we are interested in is approximately $\rho^*(G)$, which can be shown to be roughly the same as the average degree of the graph. Using this fact, it follows almost immediately that we only have to sample $\tilde O(n)$ edges. 
Thus, to repeatedly remove nodes for $L$ iterations, we will need to sample $\tilde O(Ln)=\tilde O(n)$ edges (we need to sample a new set of edges in every iteration to avoid dependencies). 

We turn the ($\alpha,  d, L$)-decomposition concept into a dynamic algorithm by dynamically maintaining the sets of nodes removed in each of the $L$ iterations, called {\em levels}. Since the ($\alpha, d, L$)-decomposition gives us a choice  whether to keep or remove each node of degree between $d$ and $\alpha d$, we can save time needed to maintain this decomposition by moving nodes between levels {\em only when it is necessary}.
If we allow $\alpha$ to be large enough, nodes will not be moved often and we can obtain a small amortized update time; in particular, it can be shown that the amortized update time is $\tilde O(1)$ if $\alpha\geq 2+\epsilon$. In analyzing an amortized time, it is usually tricky to come up with the right {\em potential function} that can keep track of the cost of moving nodes between levels, which is not frequent but expensive. 
In case of our algorithm, we have to define  two potential functions for our amortized analysis, one on nodes and one on edges. 
(For intuition, we provide an analysis for the simpler case where we run this dynamic algorithm directly on the input graph in 
Section~\ref{sec:dynamic}.)

Our goal is to run the dynamic algorithm   on top of the sketch maintained by our streaming algorithm in order to maintain the ($\alpha, d, L$)-decomposition. To do this, there are a few issues we have to deal with that makes the analysis rather complicated: In the sketch we maintain $L$ sets of sampled edges, and for each of  the $L$ iterations we use different such sets to determine which nodes to remove. This causes the potential functions and its analysis to be even more complicated since whether a node should be moved from one level to another depends on its degree in one set, but the cost of moving such node depends on its degree in other sets as well. The analysis, however, goes through (intuitively because all sets are sampled from the same graph and so their degree distributions are close enough). See Section \ref{sec:combine} for further details.

 \subsection{Roadmap}
 
 The rest of the paper is organized as follows. 
 \begin{itemize}
 \item We define the preliminary concepts and notations in Section~\ref{sec:prelim}.
 \item In Section~\ref{sec:sketch}, we present an algorithm that returns a $(2+\epsilon)$-approximation to the value of the densest subgraph. The algorithm processes a stream of edge insertions/deletions using only $\O(n)$ bits of space, and at the end of the stream returns an estimate of $\rho^*(G)$ in $\O(n)$ time. The output of the algorithm is correct with high probability. 
 \item  In Section~\ref{sec:dynamic}, we present a deterministic algorithm that maintains a $(4+\epsilon)$-approximation to the value of the densest subgraph in $\O(m+n)$ space. It  has $\O(1)$ update and query times.
 \item We present our main result in Section~\ref{sec:combine}. Specifically, combining  the techniques from Sections~\ref{sec:sketch} and~\ref{sec:dynamic}, we design  an algorithm that maintains a $(4+\epsilon)$-approximation to the value of the densest subgraph with high probability, and requires only $\O(n)$ space and $\O(1)$  update  time. 
  \item In Section~\ref{sec:directed}, we extend the result from Section~\ref{sec:dynamic} to directed graphs. Specifically, in a directed graph, we present a deterministic algorithm that maintains an $(8+\epsilon)$-approximation  to the value of the densest subgraph using $\O(m+n)$ space. It has $\O(1)$ update and query times.
   \item In Sections~\ref{sec:sublinear} and \ref{sec:distributed} we present simple extensions of our result from Section~\ref{sec:sketch}, giving sublinear time and distributed-streaming algorithms for densest subgraph.
 \end{itemize}

\section{Notations and Preliminaries}
\label{sec:prelim}

We start by defining some notations that will be used throughout the rest of the paper. We denote the input graph by $G = (V, E)$. It has $n = |V|$ nodes and $m = |E|$ edges. Let $\N_v = \{u \in V : (u,v) \in E\}$ and $\dd_v = |\N_v|$ respectively denote the set of  neighbors and the degree of a node $v \in V$.  Consider any subset of nodes $S \subseteq V$. Let $E(S) = \{ (u,v) \in E : u, v \in S\}$ denote the set of edges with both endpoints in $S$, and let $G(S) = (V, E(S))$ denote the subgraph of $G$ induced by the  nodes in $S$. Further, given  any subset of edges  $E' \subseteq E$ and any node $u \in V$, define $\N_u(S, E') = \{ v \in \N_u \cap S : (u,v) \in E'\}$ and  $\dd_u(S, E') = |\N_u(S, E')|$. In other words,  $\N_u(S, E')$ is the  subset of nodes in $S$ that are neighbors of $u$ in the subgraph induced by the edges in $E'$, whereas $\dd_u(S, E')$ denotes the degree of $u$ among the nodes in $S$ in the same subgraph. For simplicity, we write $\N_u(S)$ and  $\dd_u(S)$ instead of $\N_u(S,E)$ and $\dd_u(S,E)$. 
If the set of nodes $S \subseteq V$ is nonempty, then its {\em density} and {\em average-degree}  are defined as $\rho(S) = |E(S)|/|S|$ and $\dl(S) = \sum_{v \in S} \dd_v(S)/|S|$ respectively.  Throughout the paper, the symbol $\O(.)$ will be used to hide  $\text{poly} (\log n, 1/\epsilon)$ factors in the running times of our algorithms, where $\epsilon > 0$ is some arbitrary small constant (the approximation guarantee will depend on $\epsilon$). Finally, for any positive integer $k$, we will use the symbol $[k]$ to denote the set $\{1, \ldots, k\}$.

This paper deals with the  ``Densest Subgraph Problem'', which is about finding a subset of nodes of maximum density. Specifically, we want to find a subset of nodes $S \subseteq V$ in the input graph $G = (V, E)$ that maximizes $\rho(S)$. Further, a subset  $S \subseteq V$ is called a {\em $\gamma$-approximate densest subgraph}, for $\gamma \geq 1$, iff $\gamma \cdot \rho(S) \geq \max_{S' \subseteq V} \rho(S')$. 
We will consider the densest subgraph problem in a dynamic setting. For a detailed description of our model, see Section~\ref{sec:intro:problem}. The main result of this paper is summarized below.

\begin{theorem}
\label{th:main:summary}
There is a dynamic data structure for the densest subgraph problem that requires $\O(n)$ bits of space, has an amortized update time of $\O(1)$, a query time of $O(1)$, and with high probability maintains a $(4+\epsilon)$-approximation to the value of the densest subgraph.
\end{theorem}

It follows that Theorem~\ref{th:main:summary} gives a single-pass semi-streaming algorithm over dynamic streams for the approximate densest subgraph problem. And, unlike most other semi-streaming algorithms, Theorem~\ref{th:main:summary} gives very fast update and query times.

\subsection{Three basic properties}

We now state three basic lemmas that will be used throughout the rest of the paper. The first lemma shows that the average degree of a set of nodes  is  twice its density.

\begin{lemma}
\label{lm:prelim:density}
For all $S \subseteq V$, we have $\dl(S) =  2 \cdot \rho(S)$. 
\end{lemma}

\begin{proof}
We have $\dl(S) = \sum_{v \in S} \dd_v(S)/|S| = 2 \cdot |E(S)|/|S| = 2 \cdot \rho(S)$. The second equality holds since every edge is incident upon two nodes. 
\end{proof}

The second lemma gives simple upper and lower bounds on the maximum density of a subgraph.
\begin{lemma}
\label{main:lm:stream:range}
Let $d^* = \max_{S \subseteq V} \rho(S)$ be the maximum density of any subgraph in $G$. Then $m/n \leq d^* < n$.
\end{lemma}

\begin{proof}
Clearly, we have $d^* \geq \rho(V) = |E|/|V| = m/n$. On the other hand, consider any subset of nodes $S' \subseteq V$. We have $\rho(S') = |E(S')|/|S'| < n |S'| /|S'| = n$. The inequality holds since the maximum degree of a node is $(n-1)$, and hence the subgraph induced by the nodes in $S'$ can have at most $n |S'|$ edges. Thus, we get: $d^* = \max_{S' \subseteq V} \rho(S') < n$.  
\end{proof}

The final lemma will also be very helpful in analyzing our algorithm in later sections.

\begin{lemma}
\label{lm:prelim:structure}
\label{main:lm:prelim:structure}
Let $S^* \subseteq V$ be a subset of nodes with maximum density, i.e., $\rho(S^*) \geq \rho(S)$ for all $S \subseteq V$. Then  $\dd_v({S^*}) \geq \rho(S^*)$ for all $v \in S^*$. Thus, the  degree of each node in $G(S^*)$ is at least  the density of $S^*$.
\end{lemma}

\begin{proof}
 Suppose that there is a node $v \in S^*$ with $\dd_{S^*}(v) <  \rho(S^*)$. Define the set $S' \leftarrow S^* \setminus \{v\}$. We derive the following bound on the average degree in $S'$.
\begin{eqnarray*}
\delta(S') & = & \frac{\sum_{u \in S'} \dd^{\o}(u)}{|S'|} \\
& = & \frac{\sum_{u \in S^*} \dd_{S^*}(u) -  2 \cdot \dd_{S^*}(v)}{|S^*| -1} \\
& = & \frac{\delta(S^*) \cdot |S^*| - 2 \cdot \dd_{S^*}(v)}{|S^*| - 1} \\
& > &  \frac{\delta(S^*) \cdot |S^*| - \delta(S^*)}{|S^*| -1} \qquad \qquad \qquad (\text{since by assumption } \dd_{S^*}(v) <  \rho(S^*) = \delta(S^*)/2)\\
& = & \delta(S^*) 
\end{eqnarray*}
Since $\dl(S') > \dl(S^*)$, we infer that $\rho(S') > \rho(S^*)$. But this contradicts the assumption that the subset of nodes $S^*$ has maximum density. Thus, we conclude that $\dd_{S^*}(v) \geq \rho(S^*)$ for every node $v \in S^*$.
\end{proof}

\subsection{$(\alpha,  d, L)$-decomposition}
\label{sec:decomposition}

Our algorithms will use the concept of an  ``$(\alpha,  d, L)$-decomposition'', as defined below. To give some intuitions behind Definition~\ref{main:def:partition}, suppose that we start by setting $Z_1 \leftarrow V$. Next, suppose that we have already constructed the subsets $Z_1 \supseteq \cdots \supseteq Z_i$ for some positive integer $i < L$. While  constructing the next subset $Z_{i+1}$, we ensure that the following two conditions are satisfied.
\begin{itemize}
\item All the nodes $v \in Z_i$ with $\dd_v(Z_i) > \alpha d$ must be included in $Z_{i+1}$.
\item All the nodes $v \in Z_i$ with $\dd_v(Z_i) < d$ must be excluded from $Z_{i+1}$. 
\end{itemize}
Using this iterative procedure, we can build an $(\alpha, d, L)$-decomposition.

\begin{definition}
\label{main:def:partition}
\label{def:partition}
Fix any $\alpha \geq 1$,  $d \geq 0$, and any positive integer $L$. Consider a family of subsets 
  $Z_1  \supseteq \cdots \supseteq Z_L$. The tuple $(Z_1, \ldots, Z_L)$ is an $(\alpha,  d, L)$-decomposition 
of the input graph $G = (V,E)$ iff $Z_1 = V$ and, for every $i \in [L-1]$, we have 
$Z_{i+1} \supseteq \left\{v \in Z_i : \dd_v(Z_i) > \alpha d \right\}$ and   $Z_{i+1} \cap \left\{ v \in Z_i : \dd_v(Z_i) < d \right\} = \emptyset$.

Given an $(\alpha, d, L)$-decomposition $(Z_1, \ldots, Z_L)$, we define $V_i = Z_i \setminus Z_{i+1}$ for all $i \in [L-1]$, and $V_i = Z_i$ for $i = L$. We say that the nodes in $V_i$ constitute the $i^{th}$ level of this decomposition. We also denote the level of a node $v \in V$ by $\ell(v)$. Thus, we have $\ell(v) = i$ whenever $v \in V_i$. 
\end{definition}

The following theorem and its immediate corollary will be of crucial importance. Roughly speaking, they state that we can use the $(\alpha, d, L)$-decomposition to $2\alpha (1+\epsilon)^2$-approximate the densest subgraph by trying different values of $d$ in powers of $(1+\epsilon)$.

\begin{theorem}
\label{main:thm:test}
\label{thm:test}
Fix any  $\alpha \geq 1$, $d \geq 0$, $\epsilon \in (0,1)$, $L = 2 + \lceil \log_{(1+\epsilon)} n \rceil$. Let $d^*= \max_{S \subseteq V} \rho(S)$ be the maximum density of any subgraph in $G = (V,E)$, and let $(Z_1, \ldots, Z_L)$ be an $(\alpha,  d, L)$-decomposition of  $G = (V,E)$. Then we have:
\begin{enumerate}
\item  If $d >  2  (1+\eps)d^*$, then $Z_L = \emptyset$.
\item  Else if $d < d^*/\alpha$, then  $Z_L \ne \emptyset$ and there is an index $j \in \{1, \ldots, L-1\}$ such that $\rho(Z_j) \geq d/(2(1+\epsilon))$.
\end{enumerate}
\end{theorem}





\begin{proof} \ 
\begin{enumerate}
\item Suppose that $d > 2  (1+\eps)d^*$. Consider any level $i \in [L-1]$, and note that $\delta(Z_i) = 2 \cdot \rho(Z_i) \leq 2 \cdot \max_{S \subseteq V} \rho(S) = 2 d^* < d/(1+\eps)$.
It follows  that the number of nodes $v$ in $G(Z_i)$ with degree $\dd_v(Z_i) \geq d$ is less than $|Z_i|/(1+\epsilon)$,
as otherwise $\delta(Z_i) \ge d/(1+\eps)$. 
Let us define the set $C_i = \{ v \in Z_i : \dd_v(Z_i) < d\}$. We have $|Z_i \setminus C_i| \leq |Z_i|/(1+\eps)$. Now, from Definition~\ref{main:def:partition} we have $Z_{i+1} \cap C_i = \emptyset$, which, in turn, implies that $|Z_{i+1}| \leq |Z_i \setminus C_i| \leq |Z_i|/(1+\eps)$. Thus, for all $i \in [L-1]$, we have $|Z_{i+1}| \leq |Z_i|/(1+\eps)$. Multiplying all these inequalities, for $i = 1$ to $L-1$, we conclude that $|Z_L| \leq |Z_1|/(1+\eps)^{L-1}$. Since $|Z_1| = |V| = n$ and $L = 2 + \lceil \log_{(1+\eps)} n \rceil$, we get $|Z_L| \leq n/(1+\eps)^{(1+ \log_{(1+\eps)} n)} < 1$. This can happen only if $Z_L = \emptyset$.

\item Suppose that $d < d^*/\alpha$, and let $S^* \subseteq V$ be a subset of nodes with highest density, i.e., $\rho(S^*) = d^*$. We will  show that $S^* \subseteq Z_i$ for all $i \in \{1, \ldots, L\}$. This will imply that $Z_L \ne \emptyset$.
Clearly, we have $S^* \subseteq V = Z_1$. By induction hypothesis, assume that $S^* \subseteq Z_i$ for some $i \in [L-1]$. We show that $S^* \subseteq Z_{i+1}$.
By Lemma~\ref{main:lm:prelim:structure}, for every node $v \in S^*$, we have $\dd_v(Z_i) \geq \dd_v(S^*) \geq \rho(S^*) = d^* > \alpha d$. Hence, from Definition~\ref{main:def:partition}, we get $v \in Z_{i+1}$ for all $v \in S^*$. This implies that $S^* \subseteq Z_{i+1}$. 

Next, we will show that if $d < d^*/\alpha$, then there is an index $j \in \{1, \ldots, L-1\}$ such that $\rho(Z_j) \geq d/(2(1+\epsilon))$. For the sake of contradiction, suppose that this is not the case. Then we have $d < d^*/\alpha$ and $\delta(Z_i) = 2 \cdot \rho(Z_i) < d/(1+\epsilon)$ for every $i \in \{1, \ldots, L-1\}$. Then, applying an argument similar to case (1), we conclude that $|Z_{i+1}| \leq |Z_i|/(1+\epsilon)$ for every $i \in \{1, \ldots, L-1\}$, which implies that $Z_L = \emptyset$. Thus, we arrive at a contradiction.
\end{enumerate}
\end{proof}

\begin{corollary}
\label{main:cor:test:1}
\label{cor:test:1}
Fix  $\alpha,  \epsilon, L, d^*$ as in Theorem~\ref{main:thm:test}. Let $\pi, \sigma > 0$ be any two numbers satisfying $\alpha \cdot \pi < d^* < \sigma/(2(1+\epsilon))$. Fix any integer $K \geq 2 + \lceil \log_{(1+\epsilon)} \left(\sigma / \pi\right) \rceil$. Discretize  the range $[\pi, \sigma]$ into powers of $(1+\epsilon)$, by defining $d_{k} = (1+\epsilon)^{k-1} \cdot \pi$ for every  $k \in [K]$. 
Next, for every $k \in [K]$, construct an $(\alpha,  d_k, L)$-decomposition $(Z_1(k), \ldots, Z_L(k))$ of $G = (V,E)$.  Let  $k'  =  \max\{k \in [K] : Z_L(k) \neq \emptyset\}$. Then we have the following guarantees:
\begin{enumerate}
\item $d^*/(\alpha (1+\epsilon)) \leq d_{k'} \leq 2(1+\epsilon)^2 \cdot d^*$.
\item There exists an index $j' \in \{1, \ldots, L-1\}$ such that $\rho(Z_{j'}(k')) \geq d_{k'}/(2(1+\epsilon))$.
\end{enumerate}
\end{corollary}

\begin{proof} \
\begin{enumerate}
\item Note that $\pi < d^*/\alpha$ and $\sigma > 2 d^*(1+\epsilon)$. Furthermore, we have $d_1 < \pi$ and $d_K \geq (1+\epsilon) \sigma$. This implies that $d_1 < d^*/\alpha$ and $d_K > 2d^*(1+\epsilon)^2$. Next, note that successive $d_k$ values differ from each other by a factor of $(1+\epsilon)$. 
Accordingly,  there exists some index $k \in [K]$ for which $d^*/(\alpha (1+\epsilon)) \leq d_k \leq 2(1+\epsilon)^2 \cdot d^*$. In other words, the set $Q = \{ k \in [K] : d^*/(\alpha (1+\epsilon)) \leq d_k \leq 2(1+\epsilon)^2 \cdot d^*\}$ is nonempty. Let $k_1 = \min_{k \in Q} \{ k \}$ and $k_2 = \max_{k \in Q} \{k\}$ respectively denote the minimum and maximum indices in the set $Q$. Observe that the set $Q$ is ``contiguous'', i.e., $Q = \{k_1, k_1+1, \ldots, k_2\}$.  Since the $d_k$ values are discretized in powers of $(1+\epsilon)$, we have $k_1 < d^*/\alpha$ and $k_2 > 2d^*(1+\epsilon)$. Hence, by Theorem~\ref{main:thm:test}, we have $Z_L(k_1) \neq \emptyset$, and $Z_L(k) = \emptyset$ for all $k \geq k_2$. It follows that the index $k'$ must satisfy the inequality $k_1 \leq k' \leq k_2$, which means that $k' \in Q$. Thus, we have $d^*/(\alpha (1+\epsilon) \leq k' \leq 2(1+\epsilon)^2 d^*$. 
\item Suppose that the claim is false. Then we have $Z_L(k') \neq \emptyset$ and $\delta(Z_i(k')) = 2 \cdot \rho(Z_i(k')) < d_{k'}/(1+\epsilon)$ for every $i \in \{1, \ldots, L-1\}$. Then, applying an argument similar to the proof of case (1) in Theorem~\ref{main:thm:test}, we conclude that $|Z_{i+1}(k')| \leq |Z_i(k')|/(1+\epsilon)$ for every $i \in \{1, \ldots, L-1\}$, which implies that $Z_L(k') = \emptyset$. Thus, we arrive at a contradiction.
\end{enumerate}
\end{proof}

\noindent  We will use the above corollary as follows. Lemma~\ref{main:lm:stream:range} states that $m/n \leq d^* < n$. Thus, in Corollary~\ref{cor:test:1}, we can choose the values of $\pi, \sigma$ and $K$ in such a way which ensures that $K = \tilde  \Theta(1)$. Hence, to maintain a $2\alpha(1+\epsilon)^3 = (2\alpha +\Theta(\epsilon))$-approximation of the maximum density, it suffices to maintain $K = \tilde \Theta(1)$ many $(\alpha, d, L)$-decompositions and to keep track of the maximum $d$ for which the topmost level (i.e., the node set $V_{L}$) of the decomposition is nonempty. This gives a query time of $O(1)$. In addition, if we want to answer a more general query which asks us to output a subgraph of approximate maximum density, then we  simply keep track of the densities of all the node-induced subgraphs of the form $Z_i(k)$, where $i \in [L-1]$, $k \in [K]$, and output the one among them with maximum density. Since there are only $K = \tilde \Theta(1)$ many decompositions to consider, and since each such decomposition has $L  = \tilde \Theta(1)$ levels, this can be done by incurring an additional cost of no more than $\Theta (K L) = \tilde \Theta(1)$  in the update time. It turns our that this simple extension applies to all the dynamic and streaming algorithms presented in the paper. Accordingly,  for simplicity of exposition, from now on we only focus on the simpler query which asks for an estimate of  the {\em value} of the densest subgraph (and not the subgraph itself).

\subsection{Two results on $\ell_0$-sampling and uniform hashing}
\label{sec:prelim:sample:hash}

We now state a well known result on $\ell_0$-sampling in the streaming setting. All the streaming algorithms in this paper will use this result.

\begin{lemma}[$\ell_0$-sampler~\cite{JowhariST11}]
\label{main:th:l0:sample}
We can process a dynamic stream of $O(\text{poly } n)$ updates in the graph $G = (V,E)$ in $\O(1)$ space, and with high probability, at each step we can maintain a simple random sample from the set $E$. The algorithm takes $\O(1)$ time to handle each update in the stream.
\end{lemma}

The next lemma deals with uniform hashing in constant time and optimal space. We will use this lemma in Section~\ref{sec:new:new:maintain}. 

\begin{lemma}~\cite{Pagh}.
\label{lm:pagh:hashing}
Let $E^*={[n] \choose 2}$ be the set of all possible unordered pairs of nodes in $V$. 
Consider any two integers $w, q \geq 1$. We can construct a $w$-wise independent uniform hash function $h : E^* \rightarrow [q]$ using $O(w \text{ poly} (\log w, \log q, \log n))$ bits of space. Given any $e \in E^*$, the hash value $h(e)$ can be evaluated in $O(1)$ time.
\end{lemma}

\subsection{Concentration bounds}
\label{sec:prelim:concentration}

We will use the following concentration bounds throughout the rest of this paper.

\begin{theorem}(Chernoff bound-I)
\label{th:chernoff:up}
Consider a collection of mutually independent random variables $\{X_1, \ldots, X_t\}$ such that $X_i \in [0,1]$ for all $i \in \{1, \ldots, t\}$. Let $X = \sum_{i=1}^t X_i$ be the sum of these random variables. Then we have $\Pr[X > (1+\epsilon) \mu] \leq e^{-\epsilon^2 \mu/3}$ whenever $E[X] \leq \mu$. \end{theorem}

\begin{theorem}(Chernoff bound-II)
\label{th:chernoff:down}
Consider a set of mutually independent random variables $\{X_1, \ldots, X_t\}$ such that $X_i \in [0,1]$ for all $i \in \{1, \ldots, t\}$. Let $X = \sum_{i=1}^t X_i$ be the sum of these random variables. Then we have $\Pr[X < (1-\epsilon) \mu] \leq e^{-\epsilon^2 \mu/2}$ whenever $E[X] \geq \mu$.
\end{theorem}

\begin{definition}(Negative association)
\label{def:negative}
A set of random variables $\{X_1, \ldots, X_t\}$ are negatively associated iff for all disjoint subsets $I, J \subseteq \{1, \ldots, t\}$ and all non-decreasing functions $f$ and $g$, we have $E[f(X_i, i \in I) \cdot g(X_j, j \in J)] \leq E[f(X_i, i \in I)] \cdot E[g(X_j, j \in J)]$.
\end{definition}

\begin{theorem}(Chernoff bound with negative dependence)
\label{th:chernoff:negative}
The Chernoff bounds, as stated in Theorems~\ref{th:chernoff:up} and~\ref{th:chernoff:down}, hold even if the random variables $\{X_1, \ldots, X_t\}$ are negatively associated.
\end{theorem}

\section{A  Semi-Streaming Algorithm}
\label{sec:sketch}\label{main:sec:stream}

In this section, we present a single-pass semi-streaming  algorithm for the densest subgraph problem. The algorithm requires only $\O(n)$ bits of space, and at the end of the stream outputs a  $(2+\epsilon)$-approximation to the value of the densest subgraph with high probability. On the negative side,  its update time can be as large as $\Omega(n)$. Our result in this section is stated in the theorem below.

\begin{theorem}
\label{main:th:stream:main}
In a single pass, we can process a dynamic stream of  updates in the graph $G$ in $\tilde{O}(n)$ space. With high probability, we can  return a $(2+O(\epsilon))$-approximation of the maximum density $d^* = \max_{S \subseteq V} \rho(S)$  at the end of the stream.
\end{theorem}

We devote the rest of this section to the proof of Theorem~\ref{main:th:stream:main}. Throughout this section, we fix a small constant $\epsilon \in (0,1/2)$ and a sufficiently large constant $c > 1$. Moreover, we set $\alpha \leftarrow (1+\epsilon)/(1-\epsilon)$, $L \leftarrow 2 + \lceil \log_{(1+\epsilon)} n\rceil$.

First, we show that we can construct a $(\alpha, d, L)$-decomposition by sampling $\tilde O(n)$ edges.

\begin{lemma}
\label{main:lm:stream:1}
Fix an integer $d > 0$, and  let $S$ be a collection of $c m (L-1) \log n/d$ mutually independent  random samples (each consisting of one edge) from the edge set $E$ of the input graph $G = (V,E)$. With high probability we can construct from $S$ an $(\alpha, d, L)$-decomposition $(Z_1, \ldots, Z_L)$ of $G$, using only $\tilde{O}((n + m/d))$ bits of space.
\end{lemma}

\begin{proof}
We partition the samples in $S$ evenly among $(L-1)$ groups $\left\{S_i\right\}, i \in [L-1]$. Thus, each $S_i$ is a collection of $c m \log n/d$ mutually independent random samples from the edge set $E$, and, furthermore, the collections $\left\{S_i\right\}, i \in [L-1]$, themselves are mutually independent. 

Consider any index $i \in \{1, \ldots, L-1\}$. Note that an edge $(u,v) \in E$ can appear multiple times in the collection of samples $S_i$. We will slightly abuse the notation introduced in the beginning of Section~\ref{sec:prelim}, and let $D_v(V', S_i)$ denote the degree of a node $v \in V' \subseteq V$ in the {\em multigraph} induced by the node set $V'$ and the samples in $S_i$. With this notation in hand,  our algorithm works as follows.
\begin{itemize}
\item Set $Z_1 \leftarrow V$. 
\item {\sc For}  $i = 1$ to $(L-1)$: Set $Z_{i+1} \leftarrow \{v \in Z_i : \dd_v(Z_i, S_i) \geq (1-\epsilon) \alpha c \log n \}$.
\end{itemize}
To analyze the correctness of the algorithm, define the (random) sets $A_i = \{ v \in Z_i : \dd_v(Z_i, E) > \alpha d\}$ and $B_i = \{ v \in Z_i : \dd_v(Z_i, E) < d\}$ for all $i \in [L-1]$. Note that for all $i \in [L-1]$, the random sets $Z_i, A_i, B_i$ are completely determined by the outcomes of the samples in $\left\{S_j\right\}, j < i$. In particular, the samples in $S_i$ are chosen independently of the sets $Z_i, A_i, B_i$. Let $\mathcal{E}_i$ be the event that (a) $Z_{i+1} \supseteq A_i$ and (b) $Z_{i+1} \cap B_i = \emptyset$. By Definition~\ref{main:def:partition}, the output $(Z_1, \ldots, Z_L)$ is a valid $(\alpha, d, L)$-decomposition of $G$ iff the event $\bigcap_{i = 1}^{L-1} \mathcal{E}_i$ occurs. Consider any $i \in [L-1]$. Below, we show that the event $\mathcal{E}_i$ occurs with high probability. The lemma follows by taking a union bound over all $i \in [L-1]$.

Fix any instantiation of the random set $Z_i$. Condition on this event, and note that this event completely determines the sets $A_i, B_i$. Consider any node $v \in A_i$. Let $X_{v,i}(j) \in \{0,1\}$ be an indicator random variable for the event that the $j^{th}$ sample in $S_i$ is  of the form $(u,v)$, with $u \in \N_v(Z_i)$. Note that the random variables $\{X_{v,i}(j)\}, j$, are mutually independent. Furthermore, we have $E[X_{v,i}(j)|Z_i] = \dd_v(Z_i)/m > \alpha d /m$ for all $j$. Since there are $cm \log n/d$ such samples in $S_i$, by linearity of expectation we get: $E[\dd_v(Z_i, S_i) | Z_i] = \sum_{j} E[X_{v,i}(j)|Z_i] > (cm \log n/d) \cdot (\alpha d/m) = \alpha c \log n$.  The node $v$ is included in $Z_{i+1}$ iff $\dd_v(Z_i, S_i) \geq (1-\epsilon) \alpha c \log n$, and this event, in turn, occurs with high probability (by Chernoff bound). Taking a union bound over all nodes $v \in A_i$, we conclude that  $\Pr[Z_{i+1} \supseteq A_i \, | \, Z_i] \geq 1 - 1/(\text{poly } n)$.  Using a similar line of reasoning, we get that  $\Pr[Z_{i+1} \cap B_i =  \emptyset \, | \, Z_i] \geq 1 - 1/(\text{poly } n)$.  Invoking a union bound over these  two events, we get $\Pr[\mathcal{E}_i \, | \, Z_i] \geq 1 - 1/(\text{poly } n)$. Since this holds for all possible instantiations of $Z_i$,  the event $\mathcal{E}_i$ itself occurs with high probability.

The space requirement of the algorithm, ignoring poly log factors, is proportional to the number of samples in $S$ (which is $cm(L-1) \log n/d$) plus the number of nodes in $V$ (which is $n$). Since $c$ is a constant and since $L = \O(1)$, we derive that the total space requirement is $O((n+m/d) \text{ poly} \log n)$.
\end{proof}

Now, to turn Lemma~\ref{main:lm:stream:1} into a streaming algorithm, we simply have to invoke Lemma~\ref{main:th:l0:sample} which follows from a well-known result about $\ell_0$-sampling in the streaming model~\cite{JowhariST11}, and a simple (but important) observation in Lemma~\ref{main:lm:stream:range}.

\begin{proof}[Proof of Theorem~\ref{main:th:stream:main}]
Define $\lambda^* = 2 \alpha \cdot (c n (L-1) \log n)$ and $K^* = 2 + \lceil \log_{(1+\epsilon)} (8\alpha n^2) \rceil$. While processing the stream of edge insertions/deletions, we simultaneously run $\lambda^* K^*$ mutually independent copies of the $\ell_0$-sampler as per Lemma~\ref{main:th:l0:sample}. Furthermore, we maintain a counter to keep track of the number of edges in the graph. Initially, the counter is set to zero. After each edge insertion (resp. deletion), the counter is incremented (resp. decremented) by one.  Thus, at the end of the stream, we get $\lambda^* K^*$ mutually independent uniform random samples from the edge set $E$, and a correct estimate of the number of edges in the graph (which is given by $m = |E|$). All these steps can be implemented in $\tilde \Theta(\lambda^* K^*) = \tilde \Theta(n)$ space. Note that if $m = 0$, then clearly the maximum density of a subgraph is also zero. Thus, for the rest of the proof, we assume that $1 \leq m \leq {n \choose 2}$. 

Next, at the end of the stream, we define $\pi = m/(2 \alpha n)$ and $\sigma = 2(1+\epsilon) n$. By Lemma~\ref{main:lm:stream:range}, we have $\alpha \cdot \pi < d^* < \sigma/(2(1+\epsilon))$. Thus, the values of $\pi$ and $\sigma$ satisfies the condition required by Corollary~\ref{cor:test:1}. Hence, following Corollary~\ref{cor:test:1}, we  set $K = 2 + \lceil \log_{(1+\epsilon)} (\sigma/\pi) \rceil$, and discretize the range $[\pi, \sigma]$ in powers of $(1+\epsilon)$ by defining $d_k = (1+\epsilon)^{k-1} \cdot \pi$ for every integer $k \in [K]$. Furthermore, we define $\lambda_k = c m (L-1) \log n/d_k$ for all $k \in [K]$. Our goal is to construct an $(\alpha, d_k, L)$-decomposition of $G$ for every $k \in [K]$. By Lemma~\ref{main:lm:stream:1},  for this we need $\sum_{k = 1}^K \lambda_k$ many mutually independent random samples from $E$. But note that: 
$$\lambda_k \geq \lambda_1 = cm (L-1) \log n/d_1= c m (L-1) \log n/\pi  = 2 \alpha \cdot (c (L-1) n \log n) = \lambda^* \ \ \text{ for all } k \in [K].$$
Next, since $\pi = m/(2\alpha n)$, $\sigma = 2 (1+\epsilon) n$ and $m \geq 1$, we also have $K^* = 2 + \lceil \log_{(1+\epsilon)} (8 \alpha n^2) \rceil \geq 2 + \lceil \log_{(1+\epsilon)} (\sigma/\pi) \rceil = K$. To summarize, we infer the following guarantee.
$$ \sum_{k = 1}^K \lambda_K \leq \lambda^* K^*.$$
In other words, while processing the stream of edge insertions/deletions, we have collected sufficiently many mutually independent random samples from $E$ so as to construct an $(\alpha, d_k, L)$-decomposition for every integer $k \in [K]$. We can now get a $(2 \alpha + \Theta(\epsilon)) = (2 + \Theta(\epsilon))$-approximation to the value of the densest subgraph by invoking Corollary~\ref{cor:test:1}.
\end{proof}

\paragraph{Remark on the update and query times.} The semi-streaming algorithm presented in this section runs  $\tilde{\Theta}(n)$ mutually independent $\ell_0$-samplers as per Lemma~\ref{main:th:l0:sample}. Thus, when an edge is inserted into (resp. deleted from) the graph, the algorithm has to update each of these $\ell_0$-samplers. This requires an update time of $\tilde{\Theta}(n)$. Finally, to answer a query about the maximum density, we have to first construct  $\tilde{\Theta}(1)$ many $(\alpha, d, L)$ decompositions (for different values of $d$) from the sampled edges maintained by the $\ell_0$ samplers, and then invoke Corollary~\ref{cor:test:1}. Thus, the query time is also $\tilde{\Theta}(n)$.

\section{A Dynamic Algorithm with Fast Update and Query Times}
\label{sec:dynamic}

The algorithm in Section~\ref{sec:sketch} maintains a $(2+\epsilon)$-approximation to the value of the densest subgraph and requires only $\O(n)$ space, but its update and query times can be as large as $\Omega(n)$. In this section, we present an algorithm with $\O(1)$ update and query times. This algorithm, however, has to store all the edges in the graph and hence has a space requirement of $\tilde{\Theta}(m+n)$. Furthermore, this algorithm maintains an approximation guarantee of $(4+\epsilon)$.

Throughout this section, we set  $\pi = 1/(4n)$, $\sigma = 4n$, $L = 2 + \lceil \log_{(1+\epsilon)} n \rceil$, and $\alpha = 2+3\epsilon$, where $\epsilon \in (0,1)$ is some small constant. Note that $\alpha \cdot \pi < d^* < \sigma/(2(1+\epsilon))$, where $d^*$ is the optimal density in the input graph.  As in Corollary~\ref{cor:test:1}, we discretize the range $[\pi, \sigma]$ in powers of $(1+\epsilon)$ by defining the values $\{d_k\}, k \in [K]$, by setting $K = 2 + \lceil \log_{(1+\epsilon)} (\sigma/\pi) \rceil$. 

We show in Section~\ref{sec:dynpart} how  to maintain an $(\alpha, d_k, L)$-decomposition of $G$  for each $k \in [K]$ in  $O(L/\eps) = O(\log n/\epsilon^2)$ amortized update time and $\Theta(m+n)$ space (see Theorem~\ref{th:sec:dynpart:main}). Since  $K = O(\log n/\epsilon)$, the total update time for all the $K$ decompositions is $O(K \log n /\eps^2) = \O(1)$ and the total space requirement is also $O(K \cdot (m+n)) = \O(m+n)$.  By Corollary~\ref{cor:test:1}, this gives  a $2\alpha (1+\epsilon)^2 = 4 + O(\eps)$-approximation to the optimal density in $G$, for sufficiently small $\epsilon$. To answer a query about the value of the densest subgraph, the algorithm needs to keep track of the index $k'$ as defined in Corollary~\ref{cor:test:1}. Since there are $O(K) = \tilde{O}(1)$ decompositions to deal with, the time taken for this operation can be subsumed within the $\O(1)$ update time. This gives us a query time of $O(1)$. We thus get the main result of this section, which is summarized below.

\begin{theorem}
\label{th:sec:dynamic:main}
There is a deterministic dynamic algorithm that maintains a $(4 + O(\eps))$-approximation to the value of the densest subgraph. The algorithm requires $\O(m+n)$ space,  has an amortized update time of $\O(1)$ and a query time of $O(1)$. 
\end{theorem}

\subsection{Dynamically maintaining an $(\alpha, d, L)$-decomposition}
\label{sec:dynpart}

We present a deterministic data structure that
is initially given $\alpha$, $d$, $L$,  and a graph $G = (V,E)$ with $|V| = n$, $E = \emptyset$. The data structure maintains an $(\alpha, d, L)$-decomposition of the  graph $G = (V,E)$ at each time-step,  and
 supports the following operations:
\begin{itemize}
\item {\em Insert($u,v$):} Insert the edge $(u,v)$ into the graph.
\item {\em Delete ($u,v$):} Delete the edge $(u,v)$ from the graph.
\end{itemize}
\noindent  Theorem~\ref{th:sec:dynpart:main} summarizes our result. We devote the rest of Section~\ref{sec:dynpart} to its proof.

\begin{theorem} 
\label{th:sec:dynpart:main}
For every polynomially bounded $\alpha \ge 2 + 3\eps$, we can deterministically  maintain an $(\alpha, d, L)$-decomposition  of  $G = (V,E)$. Starting from an empty graph, our data structure  handles a sequence of $t$ update operations (edge insertions/deletions) in total time $O(t L/\eps)$. Thus, we get an amortized update time of $O(L/\eps)$. The space complexity of the data structure at a given time-step is $O(n +m)$, where $m = |E|$ denotes the number of edges in the input graph at that time-step. 
\end{theorem}

\paragraph{Data Structures.} 
We use the following data structures.

\begin{enumerate}
\item Every node $v \in V$ maintains $L$  lists $\text{{\sc Friends}}_i[v]$, for $i \in \{1, \ldots, L\}$. For $i < \ell(v)$, the list $\text{{\sc Friends}}_i[v]$ consists of the neighbors of $v$ that are at level $i$ (given by the set $\NN_v(V_i)$). For $i = \ell(v)$, the set $\text{{\sc Friends}}_i[v]$ consists of the neighbors of $v$ that are at level $i$ or above (given by the set $\NN_v(Z_i)$). For $i > \ell(v)$, the list $\text{{\sc Friends}}_i[v]$ is empty.
Each list is stored in a doubly linked list together with its size, $\text{{\sc Count}}_i[v]$. Using appropriate pointers, we can insert or delete a given node to or from a concerned list in constant time.
\item The counter $\text{{\sc Level}}[v]$ keeps track of the level of the node $v$.
\end{enumerate}

\paragraph{Algorithm.}
If a node   violates one of the conditions
of an $(\alpha, d, L)$-decomposition (see Definition~\ref{def:partition}), then we call the node ``dirty'', else the node is called ``clean''. Specifically a node $y$ at level $\ell(y) = i$ is dirty iff either $\dd_y(Z_i) > \alpha d$ or $\dd_y(Z_{i-1}) < d$. Initially, the input graph $G = (V,E)$ is empty,  every node $v \in V$  is at level $1$, and  every node is clean. 

When an edge $(u,v)$ is inserted/deleted,  we first update the $\text{{\sc Friends}}$ lists of $u$ and
$v$ by adding or removing neighbors in constant time. Next we check whether $u$ or $v$ are dirty.
If so, we run the RECOVER() procedure described in Figure~\ref{fig:dirty:main}. Note that a single iteration of the {\sc While} loop (Steps 01-05) may change the status of some more nodes from clean to dirty (or vice versa). If and when the procedure terminates, however, every node is clean by definition.

\begin{figure}[htbp]
\centerline{\framebox{
\begin{minipage}{5.5in}
\begin{tabbing}
01. \ \ \ \    \=  {\sc While} there exists a dirty node  $y$ \\
02.  \>  \ \ \ \ \ \ \ \= {\sc If} $\dd_y(Z_{\ell(y)}) > \alpha d$ and $\ell(y) < L$, {\sc Then} \\
03.  \> \> \ \ \ \ \ \ \ \ \ \= Increment the level of $y$ by setting $\ell(y) \leftarrow \ell(y)+1$. \\
04.  \> \> {\sc Else if} $\dd_y(Z_{\ell(y)-1}) <  d$ and $\ell(y) > 1$, {\sc Then} \\
05.  \>  \> \> Decrement the level of $y$ by setting $\ell(y) \leftarrow \ell(y)-1$. 
\end{tabbing}
\end{minipage}
}}
\caption{\label{fig:dirty:main} RECOVER().}
\end{figure}

\paragraph{Analyzing the space complexity.} Since each edge in $G$  appears in two linked lists (corresponding to each of its endpoints), the space complexity of the data structure is $O(n+ m)$, where $m = |E|$.

\paragraph{Analysis of the Update Time.}
Each update operation takes constant time plus the time for the RECOVER() procedure. We show below that the total time spent in procedure 
RECOVER() during $t$ update operations is $O(t L/\eps)$.

\paragraph{Potential Function.} To determine the amortized update time we use a potential function $\B$. Let $f(u,v) = 1$ if $l(u) = l(v)$ and let it be 0 otherwise. We define $\B$ and the node and edge potentials $\Phi(v)$ and $\Psi(u,v)$ as follows. 

\begin{eqnarray}
\label{eq:potential:main}
\B & = & \sum_{v \in V} \Phi(v) + \sum_{e \in E} \Psi(e) \\
\Phi(v) & = & {1 \over \eps} \sum_{i = 1}^{\ell(v)-1} \max(0, \alpha d - \dd_v(Z_i)) \ \ \text{ for all nodes } v \in V \label{eq:potential:node} \\
\Psi(u,v) & = &2  (L -  \min(\ell(u), \ell(v))) + f(u,v)\ \ \text{ for all edges } (u,v) \in E \label{eq:potential:edge}
\end{eqnarray}
It is easy to check that all these potentials are nonnegative, and that they are uniquely defined by the partition $V_1, \ldots, V_L$ of the set of nodes $V$. Initially, the input graph $G$ is empty and the total potential $\B$ is zero. We show that (a) insertion/deletion of an edge (excluding subroutine RECOVER())  increases the total potential by at most $3L/\eps$, and (b) for each unit of computation performed by procedure RECOVER() in Figure~\ref{fig:dirty:main}, the total potential decreases by at least $\Omega(1)$. Since the total potential remains always nonnegative, these two conditions together imply an amortized update time of $O(L/\eps)$.

\paragraph{Insertion.} The insertion of  edge $(u,v)$ creates a new potential $\Psi(u,v)$ with value  at most $3L$. Further, the potentials $\Phi(u)$ and $\Phi(v)$ do not increase, and the potentials associated with all other nodes and edges remain unchanged. {Thus, the net increase in the potential $\B$ is at most $3L$.}

\paragraph{Deletion.} The deletion of edge $(u,v)$  destroys the (nonnegative) potential $\Psi(u,v)$. Further, each of the potentials $\Phi(u)$ and $\Phi(v)$ increases by at most $L/\eps$, and the potentials of all other nodes and edges remain unchanged. {Thus, the net increase in the potential $\B$ is at most $2L/\eps$.}

\medskip
\noindent It remains to relate the change in the potential $\B$  with the amount of computation performed.  See Section~\ref{sec:potential:recover}. For ease of exposition, we first describe a high level overview of the  analysis. 

\subsubsection{A high level overview of the potential function based analysis}
The intuition behind this potential function is as follows. We maintain a data structure so that the change of the level of a node $y$ from $i$ to $i+1$ or from $i$ to $i-1$ takes time $O(1 + \dd_y(Z_i))$.
Ignoring the constant factor (as we can multiply the potential function by this constant), we assume in the following that the cost is simply $1 + \dd_y(Z_i)$.
The basic idea is that the insertion or deletion of an edge should increase the potential function in order to pay for all future level changes. To implement this idea (1) each node gets a potential that
increases when an adjacent  edge is deleted and that will pay for future level {\em decreases} of the node, and (2) each edge in $G$ gets a potential that pays for future level {\em increases} of its end points.
We explain how we implement this in more detail next:
We know that when a node moves up to level $i+1$ it has degree at least
$\alpha d$ to nodes in $Z_{i}$, while when it moves back down it has degree at most $d$ to nodes in $Z_i$.  Assuming that the drop in the nodes degree was caused by the deletion of adjacent edges,
the difference of the two, i.e. $(\alpha -1)d$ has to be used to pay for the cost of  a level
decreases of a node, which is $1 + \dd_y(Z_i) \le d$. This is possible if we set $\alpha \ge 2$. The value of $\alpha$ can even be reduced by multiplying the potential of each node by $1/\eps$. Then the drop in 
potential is $(\alpha - 1)d/\eps$ while the cost is only of $d$.

There is, however, an additional complication in this scheme, which forces us to set $\alpha = 2 + \Theta(\eps)$: A node on level $i$ might not only decrease its level because of edge deletions (of edges to nodes on  level $i$ or higher), but also if a node on level $i$ moves down to level $i-1$.
Said differently, the drop of $(\alpha -1)d/\eps$ of the degree of a node $y$  on level $i$ might not only be caused by edge deletions, but also by the level drop of incident nodes.  Thus, when the level of a node $y$ decreases, the potential of all its neighbors
on a larger level has to increase by $1/\eps$ to pay for their future level decrease. Thus the drop of the potential of $y$ by $(\alpha -1)d/\eps$ has to ``pay'' for the increase of the potential of its neighbors, which is in total at most $d/\eps$, {\em and} the cost of the operation, which is $d$. This is possible if we set $\alpha = 2 + \eps$.

\renewcommand{\O}{\tilde{O}}

\subsubsection{Analyzing the subroutine RECOVER().}
\label{sec:potential:recover}
We will analyze any single iteration of the {\sc While} loop in Figure~\ref{fig:dirty:main}. During this iteration, a dirty node $y$ either increments its level by one unit, or decrements its level by one unit. Accordingly, we consider two possible events.

\bigskip
\noindent {\bf Event 1: A dirty node $y$ changes its level from  $i$ to $(i+1)$.} 

\medskip
 First, we upper bound the amount of computation  performed during this event. Our algorithm scans through the list $\text{{\sc Friends}}_i[y]$ and identifies the neighbors of $y$ that are at level $(i+1)$ or above. For every such node $x \in \NN_y \cap Z_{i+1}$, we need to (a) remove $y$ from the list $\text{{\sc Friends}}_i[x]$ and add $y$ to the list $\text{{\sc Friends}}_{i+1}[x]$, (b) increment the counter $\text{{\sc Count}}_{i+1}[x]$ by one unit, (c) add $x$ to the list $\text{{\sc Friends}}_{i+1}[y]$ and remove $x$ from the list $\text{{\sc Friends}}_{i}[y]$, (d) decrement the counter $\text{{\sc Count}}_i[y]$ by one unit and increment the counter $\text{{\sc Count}}_{i+1}[y]$ by one unit. Finally, we set  $\text{{\sc Level}}[y] \leftarrow i+1$. Overall,   $O(1 + \dd_y(Z_i))$ units of computation are performed during this event.

\medskip
\noindent Next, we lower bound the net decrease in the $\B$ due to this event. We first discuss the node potentials.
\begin{itemize}
\item (a)  Since the node $y$ gets promoted to level $i+1$, we must have $\dd_y(Z_i) > \alpha d$, which implies that $\max(0, \alpha d - \dd_y(Z_i)) = 0$, so that the potential $\Phi(y)$ does not change.
\item (b) The potential of a node $x \in \N_y$ can only decrease.
\item (b) The potential of a node $x \notin \{y \cup \N_y\}$ does not change.
\end{itemize}
\noindent Accordingly,  the sum $\sum_{v \in V} \Phi(v)$ does not increase.

\medskip
\noindent Next we consider the edge potentials.  Towards this end, we first consider the edges incident upon $y$. Specifically, consider a node $x \in \NN_y$.
\begin{itemize}
\item If $\ell(x) < i$, the potential of the edge $(x,y)$ does not change.
\item If $\ell(x) = i $, the potential of $(x,y)$ is $2(L - i) + 1$ before the level change and $2(L - i)$ afterwards, i.e., the potential drops by one.
\item If $\ell(x) = i + 1$, the potential of $(x,y)$ is $2(L - i)$ before the level change and $2(L - (i+1)) + 1 = 2(L - i) - 1$ afterwards, i.e., it drops by one.
\item If $\ell(x) > i+1$, the potential of $(x,y)$ is $2(L - i)$ before the level change and $2(L - ( i+1))$ afterwards, i.e., it drops by two. 
\end{itemize}
\noindent The potentials associated with all other edges remain unchanged. Thus, the sum $\sum_{e \in E} \Psi(e)$ drops by at least $\dd_y(Z_i)$. 

\medskip
\noindent We infer that the net decrease in the overall potential $\B$ is at least $\dd_y(Z_i)$. Note that $\dd_y(Z_i) > 0$ (for otherwise the node $y$ would not have been promoted to level $i+1$). It follows that the net decrease in $\B$ is sufficient to pay for the cost of the computation performed, which, as shown above, is $O(1+\dd_y(Z_i))$.

\bigskip
\noindent {\bf Event 2: A dirty node $y$ changes its level from level $i$ to $(i-1)$.}

\medskip
 First, we upper bound the amount of computation  performed during this event. Our algorithm scans through the nodes in the list $\text{{\sc Friends}}_i[y]$. 
For each such node $x \in \NN_y \cap Z_i$, we need to (a) remove $y$ from the list $\text{{\sc Friends}}_i[x]$ and add $y$ to the list $\text{{\sc Friends}}_{i-1}[x]$ and (b) decrement the counter $\text{{\sc Count}}_i[x]$.
 Finally, we need to add all the nodes in $\text{{\sc Friends}}_i[y]$ to the list $\text{{\sc Friends}}_{i-1}[y]$, make $\text{{\sc Friends}}_i[y]$ into an empty list, and set $\text{{\sc Counter}}_i[y]$ to zero. Finally, we set $\text{{\sc Level}}[y] \leftarrow i-1$. Overall, $O(1 + \dd_y(Z_i))$ units of computation are performed during this event.

\medskip
\noindent Next, we lower bound the net decrease in the  overall potential $\B$ due to this event. 
We first consider the changes in the node potentials.
\begin{itemize}
\item (a) Since the node $y$ was demoted to level $i-1$, we must have $\dd_v(Z_{i-1}) < d$. Accordingly, the potential $\Phi(y)$ drops by at least 
$(\alpha  - 1) \cdot (d/\eps)$ units due to the decrease in $\ell(y)$.
\item (b) For every neighbor $x$ of $y$, $\dd_x(Z_i)$ decreases by one while $\dd_x(Z_j)$ for $j \neq i$ is unchanged.  The potential function of a node
$x$ considers only the $\dd_x(Z_j)$ values if $j < \ell(x)$. Thus, only for neighbors $x$ with $\ell(x) > i$ does the potential function change, specifically it
increases by at most $1/\eps$. Thus the sum $\sum_{x \in \N_y} \Phi(x)$ increases by at most $\dd_y(Z_{i+1})/\eps$. Further, note that $\dd_y(Z_{i+1})/\epsilon \leq \dd_y(Z_{i-1})/\epsilon < d/\epsilon$.  The last inequality holds since the node $y$ was demoted from level $i$ to level $(i-1)$.
\item The potential $\Phi(x)$ remains unchanged for every node $x \notin \{y \} \cup \N_y$.
\end{itemize}
\noindent Thus, the sum $\sum_{v \in V} \Phi(v)$ drops by at least $ (\alpha  - 1) \cdot (d/\epsilon)  - (d/ \eps) = (\alpha -2) \cdot (d/\epsilon)$.

\medskip
\noindent We next consider edge potentials.  Towards this end, we first consider the edges incident upon $y$.  Specifically, consider any node $x \in \NN_y$.
\begin{itemize}
\item If $\ell(x) < i -1$, the potential of the edge $(x,y)$ does not change.
\item If $\ell(x) = i-1 $, the potential of $(x,y)$ is $2(L - (i-1))$ before the level change and $2(L - (i-1)) + 1$ afterwards, i.e., the potential increases by one.
\item If $\ell(x) = i$, the potential of $(x,y)$ is $2(L - i) + 1$ before the level change and $2(L - (i-1))  = 2(L - i) +2$ afterwards, i.e., it increases by one.
\item If $\ell(x) \ge i+1$, the potential of $(x,y)$ is $2(L - i)$ before the level change and $2(L - ( i-1))$ afterwards, i.e., it increases by two. 
\end{itemize}
\noindent The potentials associated with all other  edges remain unchanged.  Thus, the sum $\sum_{e \in E} \Psi(e)$  increases by at most $2 \dd_y(Z_{i-1}) < 2 d$.

\medskip
\noindent We infer that the overall potential $\B$ drops by at least $ (\alpha - 2) \cdot (d/\epsilon) - 2d = (\alpha - 2 - 2\epsilon) \cdot (d/\epsilon)$.
Accordingly, for $\alpha \ge 2 + 3 \eps$ this potential drop is at least $d \ge \dd_y(Z_i) + 1$. We conclude that the net drop in the overall potential $\B$ is again sufficient to pay for the cost of the computation performed. This concludes the proof of Theorem~\ref{th:sec:dynpart:main}.


\section{A Semi-Streaming Algorithm with Fast Update and Query times}
\label{main:sec:dynamic-stream}\label{sec:combine}

In Section~\ref{sec:sketch}, we presented a semi-streaming algorithm that maintains a $(2+\epsilon)$-approximation of $\rho^*(G)$. Specifically, the algorithm can process a dynamic  stream of updates (edge insertions/deletions) using only $\tilde O(n)$ bits of space. Unfortunately, however, it has a large update time of $\tilde \Theta(n)$, and it answers a query only at the end of the stream (also in time  $\tilde \Theta(n)$). 

On the other hand, in Section~\ref{sec:dynamic} we presented an algorithm that maintains a $(4+\epsilon)$-approximation of $\rho^*(G)$. This algorithm has the advantage of having very fast (i.e., $\O(1)$) update and query times. Furthermore, it can answer a query at any given time-instant (even in the middle of the stream). But, unlike the algorithm from Section~\ref{sec:sketch} whose space complexity is $\O(n)$, it has to store all the edges  and requires $\tilde \Theta(m+n)$ bits of space.

In this section, we combine the techniques from Sections~\ref{sec:sketch} and~\ref{sec:dynamic} to get a result that captures the best of both worlds. Specifically, we present a new algorithm that maintains a $(4+\epsilon)$-approximation of $\rho^*(G)$ while processing a stream of updates (edge insertions/deletions). We highlight that:
\begin{itemize}
\item The algorithm has very fast (i.e., $\O(1)$) update and query times.
\item It requires very little (i.e., $\O(n)$) space.
\item It  can answer a query at any time-instant (i.e., even in the middle of the stream). 
\end{itemize}

\subsection{An Overview of Our Result} 
\label{sec:new:setting}

We denote the input graph by $G = (V, E)$. It has $|V| = n$ nodes, and in the beginning of our algorithm the graph is empty (i.e.,  $E = \emptyset$).  Subsequently, our algorithm processes a ``stream of updates'' in the graph. Each update consists of an edge insertion/deletion. Specifically, at each ``time-step'', either an edge is inserted into the graph or an already existing edge is deleted from the graph. The  node set of the graph, however, remains unchanged over time.  For any integer $t \geq 0$, we let $G^{(t)} = (V, E^{(t)})$ denote the status of the input graph at time-step $t$ (i.e., after the $t^{th}$ edge insertion/deletion). Thus,  $G^{(0)} = (V, E^{(0)})$ denotes the status of $G$ in the beginning, which implies that  $E^{(0)} = \emptyset$. Further, we let $m^{(t)} = |E^{(t)}|$ denote the number of edges in the graph $G^{(t)}$. Finally, $\text{{\sc Opt}}^{(t)} = \rho(G^{(t)})$ gives the value of the densest subgraph in $G^{(t)}$. We also use the notations and concepts introduced in Section~\ref{sec:prelim}. Our algorithm will maintain a value $\text{{\sc Output}}^{(t)}$ at each time-step $t$. We want  $\text{{\sc Output}}^{(t)}$ to be a $(4+\epsilon)$-approximation to $\text{{\sc Opt}}^{(t)}$.

Throughout Section~\ref{sec:combine}, we  fix the symbols $\epsilon, \alpha, L, c, \lambda$ and $T$ as defined below.
\begin{eqnarray}
\label{eq:new:epsilon}
\epsilon \in (0,1) \text{ has a sufficiently small positive value.} \\
\alpha = 2 + \Theta(\epsilon) = 2 + c^* \cdot \epsilon, \text{ where } c^* \text{ is some constant independent of } \epsilon \text{ (to be decided later).} \label{eq:new:epsilon} \\
L = 2 + \lceil \log_{(1+\epsilon)} n \rceil \label{eq:new:L} \\
\lambda \geq 1 \text{ is any positive constant, and } c \text{ is a constant such that } c >> \lambda. \label{eq:new:lambda} \\
T = \lceil n^{\lambda} \rceil \label{eq:new:T}
\end{eqnarray}
We use the symbols $\O(.)$ and $\tilde \Theta(.)$ to hide $\text{poly} (\log n, 1/\epsilon)$ factors.  We  now  state the main result.

\begin{theorem}
\label{th:new:main}
Define $T$ as in equation~\ref{eq:new:T}. There is an algorithm that processes a stream of  $T$ updates   (starting from an empty graph), and
satisfies the following properties with high probability:
\begin{itemize}
\item It uses only $\O(n)$ bits of space.
\item The total time taken  to process the  $T$ edge insertions/deletions is $\O(T)$. Thus, it has an amortized update time of $\O(1)$. 
\item It maintains a value $\text{{\sc Output}}^{(t)}$ such that for all time-steps $t \in [1, T]$ we have $\text{{\sc Output}}^{(t)} \leq \text{{\sc Opt}}^{(t)} \leq (4+\epsilon) \cdot \text{{\sc Opt}}^{(t)}$. Thus, the algorithm maintains a $(4+\epsilon)$-approximation to the value of the densest subgraph while processing the stream of updates, and the query time is $O(1)$.
\end{itemize}
\end{theorem}

Note that the algorithm in Theorem~\ref{th:new:main} works only for polynomially many time-steps (since $T = \Theta(n^{\lambda})$ and $\lambda$ is a constant). In contrast, we did not impose this restriction while presenting our semi-streaming algorithm in Section~\ref{sec:sketch}. To see why this is the case, recall that a semi-streaming algorithm maintains a ``sketch'' of the input while processing the stream of edge insertions/deletions. For our algorithm in Section~\ref{sec:sketch}, the sketch is simply the collection of  random samples from the edge set of the input graph. Let $\text{{\sc Sketch}}^{(t)}$ denote the status of the sketch at time-step $t$ (i.e., it corresponds to the graph $G^{(t)}$). Now, the following condition holds in Section~\ref{sec:sketch}: 
\begin{itemize}
\item (P1) Fix any time-step $t$. With high probability, if we run the procedure in Section~\ref{sec:sketch} on $\text{{\sc Sketch}}^{(t)}$, then this gives us a good approximation to the value of $\text{{\sc Opt}}^{(t)}$.
\end{itemize}
A semi-streaming algorithm typically needs to invoke property (P1) {\em only at the end of the stream}, since it answers a query after processing all the edge insertions/deletions. In this section, however, we want an algorithm that can answer a query at any given time-instant (i.e., even in the middle of the stream). Thus, we want the stronger property stated below.
\begin{itemize}
\item (P2) The following event holds with high probability. For {\em every time-step} $t \in [1, T]$, we can get a good approximation to $\text{{\sc Opt}}^{(t)}$ using $\text{{\sc Sketch}}^{(t)}$. 
\end{itemize}
Intuitively, (P2) follows if we take a union bound over the complement of (P1) for all time-steps $t \in [1, T]$. But this can be done only if the length of the interval $[1, T]$ is bounded by some polynomial in $n$. Thus, we require that $T = \Theta(n^{\lambda})$ for some constant $\lambda$.

\subsubsection{Main technical challenges} 
\label{main:sec:challenge}
At a very high level,  the following approach seems natural  for proving Theorem~\ref{th:new:main}. First, using the techniques from Section~\ref{sec:sketch}, maintain $\O(n)$ uniformly random samples from the edge set of the input graph. Next, using the techniques from Section~\ref{sec:dynamic}, maintain $(\alpha, d, L)$-decompositions on these randomly sampled edges. The first step should ensure that the algorithm requires only $\O(n)$ bits of space, while the second step should ensure that the algorithm requires $\O(1)$ update and query times. Unfortunately, however, to implement this simple idea we need to overcome several intricate technical challenges. They are described below. 
\begin{enumerate}
\item As stated at the end of Section~\ref{sec:sketch}, maintaining the random samples from the edge set $E$ requires $\tilde \Theta(n)$ update time, since to process the insertion/deletion of an edge we have to update $\tilde \Theta(n)$ many $\ell_0$-samplers. So the first challenge is to speed up the update time of the subroutine that maintains the randomly sampled edges. This is done in Section~\ref{sec:new:new:maintain}. 
\item The algorithm in Section~\ref{sec:sketch} uses one crucial observation that is captured in Lemma~\ref{main:lm:stream:range}. Specifically, the maximum density of a subgraph of $G = (V, E)$ lies in the range $[m/n, n]$, where $m = |E|$ and $n = |V|$. Thus,  we set $\pi = m/(2\alpha n)$ and $\sigma = 2 (1+\epsilon) n$, so that we have $\alpha \cdot \pi < d^* < \sigma/(2(1+\epsilon))$, where $d^* = \max_{S \subseteq V} \rho(S)$ is the value of the maximum density of a subgraph of $G$ (see the discussion immediately after the proof of Lemma~\ref{main:lm:stream:1}). Then we discretize the range $[\pi, \sigma]$ in powers of $(1+\epsilon)$, by defining the values $d_k, k \in [K],$ as per Corollary~\ref{cor:test:1}. For each $k \in [K]$, we construct an $(\alpha, d_k, L)$-decomposition. Finally, we approximate the value of $d^*$ by looking at the topmost levels (i.e., the node set $V_L$) of each of these decompositions. To implement this approach in Section~\ref{sec:sketch}, we  wait till the end of the stream to get the  value of $m$  after all insertions/deletions. This is of crucial importance since the degree-threshold $d_k$ (as per Corollary~\ref{cor:test:1}) for the $k^{th}$ $(\alpha, d, L)$-decomposition depends on the value of $\pi$, which, in turn, depends on $m$. In this section, however, we have to maintain a solution at every time-step in the interval $[1, T]$. Consequently,  we have to maintain an $(\alpha, d_k, L)$-decomposition for each $k \in [K]$  throughout the interval $[1, T]$. Hence, the degree-threshold $d_k$ of the $k^{th}$ decomposition changes over time as edges are inserted/deleted into the input graph. Thus, we need to extend the dynamic algorithm from Section~\ref{sec:dynpart} (which maintains an $(\alpha, d, L)$-decomposition for a fixed $d$) so that it can handle the  scenario where $d$  changes over time. See Section~\ref{sec:dense} for further details. 
\item Suppose that we want to construct an $(\alpha, d, L)$-decomposition using $\O(n)$ bits of space. In Section~\ref{sec:sketch}, to achieve this goal we used a collection of sets $\{S_i\}, i \in \{1, \ldots, L-1\}$. Recall  the proof of Lemma~\ref{main:lm:stream:1} for details. Specifically, each $S_i$ consisted of $\Theta(m \log n/d)$ many uniformly random samples from the edge set $E$. Given the subset of nodes $Z_i$, we used the samples  $S_i$ to construct the next subset  $Z_{i+1} \subseteq Z_i$. Thus, we used different collections of sampled edges for different levels of the $(\alpha, d, L)$-decomposition. The reason behind this was as follows: For the proof of Lemma~\ref{main:lm:stream:1} to be valid, it was crucial that the samples used for defining the set $Z_{i+1}$ be chosen independently of the samples used for the sets $Z_1, \ldots, Z_i$. This is in sharp contrast to our dynamic algorithm  in Section~\ref{sec:dynpart}  for maintaining an $(\alpha, d, L)$-decomposition:  that algorithm uses the same edge set $E$ for different levels of the decomposition. Thus, we need to find a way to extend the potential function based analysis from Section~\ref{sec:dynpart} to a setting where different levels of the $(\alpha, d, L)$-decomposition are concerned with different sets of edges (chosen uniformly at random). See Section~\ref{sec:dense} for further details. 
\end{enumerate}

\paragraph{Roadmap for the rest of Section~\ref{sec:combine}.} The rest of this section is organized as follows. 
\begin{itemize}
\item In Section~\ref{sec:new:new:maintain}, we show how to maintain random samples from the edge set of the input graph in $\O(n)$ space and $\O(1)$ update time. For technical reasons, we have to run two separate algorithms for this purpose. Roughly speaking, the first algorithm (as stated in Theorem~\ref{main:th:dynamic:sample:sparse}) maintains the entire edge set of the graph whenever $m^{(t)} = \O(n)$, whereas the second algorithm (as stated in Theorem~\ref{main:th:dynamic:sample:dense}) maintains $\tilde \Theta(n)$ random samples from the edge set of the graph whenever $m^{(t)} = \tilde \Omega(n)$. 
\item In Section~\ref{sec:overview}, we present a high level overview of our main algorithm. The main idea is to classify each time-step as either ``dense'' or ``sparse'', depending on the number of edges in the input graph. This classification is done in such a way that Theorem~\ref{main:th:dynamic:sample:sparse} applies to all sparse time-steps, whereas Theorem~\ref{main:th:dynamic:sample:dense} applies to all dense time-steps. 
\item In Section~\ref{sec:sparse}, we present our algorithm for maintaining a $(4+\epsilon)$-approximation to $\text{{\sc Opt}}^{(t)}$ during all the sparse time-steps (see Theorem~\ref{th:sample:sparse}). This algorithm takes as input the set of edges maintained by the subroutine from Theorem~\ref{main:th:dynamic:sample:sparse}.
\item In Section~\ref{sec:dense}, we present our algorithm for maintaining a $(4+\epsilon)$-approximation to $\text{{\sc Opt}}^{(t)}$ during all the dense time-steps (see Theorem~\ref{th:sample:dense}). This algorithm takes as input the set of edges maintained by the subroutine from Theorem~\ref{main:th:dynamic:sample:dense}.
\item Theorem~\ref{th:new:main} follows from Theorem~\ref{main:th:dynamic:sample:sparse} and Theorem~\ref{main:th:dynamic:sample:dense}.
\end{itemize}

\subsection{Maintaining the randomly sampled edges in $\O(1)$ update time}
\label{sec:new:new:maintain}

Intuitively, we want to maintain  $s$ uniformly random samples from the edge-set $E$ of the input graph $G = (V, E)$, for some $s = \tilde \Theta(n)$. In Section~\ref{sec:sketch}, we achieved this by running $\tilde \Theta(n)$ mutually independent copies of the $\ell_0$-sampler. This ensured a space complexity of $\tilde \Theta(n)$. However, after each edge insertion/deletion in the input graph, we had to update each of the $\ell_0$-samplers. So the update time of our algorithm became $\tilde \Theta(n)$. In this section, we show how to bring down this update time (for maintaining the randomly sampled edges) to $\tilde \Theta(1)$ without compromising on the space complexity. 

To see the high level idea behind our approach, suppose that the input graph has a large number of edges, i.e.,  $s \ll m = |E|$. We maintain $s$ ``buckets'' $B_1, \ldots, B_s$. Whenever an edge $e$ is inserted into the input graph, we insert the edge into a bucket chosen uniformly at random. When the edge gets deleted from the input graph, we also delete it from the bucket it was assigned to. Thus, the buckets $B_1, \ldots, B_s$ give a random partition of the edge set of the input graph $G = (V, E)$. We can maintain this partition using an appropriate hash function. Now, we run $s$ mutually independent copies of the $\ell_0$ sampler, {\em one for each bucket $B_i, i \in \{1, \ldots, s\}$}. Let $S$ be the collection of edges returned by these $\ell_0$-samplers. Thus, we have $S \subseteq E$, $|S| = s$, and any given edge $e \in E$ belongs to $S$ with a probability that is very close to  $s/m$. In other words, the set $S$ is a set of $s$ edges chosen uniformly at random from the edge set $E$ (without replacement). This solves our problem. The space requirement is $\tilde \Theta(s) = \tilde \Theta(n)$ since we run $s$ copies of the $\ell_0$-sampler and each of this samplers needs $\tilde \Theta(1)$ space. Furthermore, unlike our algorithm in Section~\ref{sec:sketch}, here if an edge  insertion/deletion takes place in the input graph $G$, then we only need to update a single $\ell_0$ sampler (the one running on the bucket that edge was assigned to). This improves the update time to $\tilde \Theta(1)$. 

To be more specific, we present two results in this section. In Theorem~\ref{main:th:dynamic:sample:sparse}, we show how to maintain the edge set of the input graph $G = (V, E)$ at all time-steps where it is ``sparse'' (i.e., $m = |E|$ is small). Next, in Theorem~\ref{main:th:dynamic:sample:dense}, we show how to maintain $\tilde \Theta(n)$ uniformly random samples (without replacement) from the edge set $E$ at all time-steps where the input graph is ``dense'' (i.e., $m = |E|$ is large).  The proofs of Theorems~\ref{main:th:dynamic:sample:sparse} and~\ref{main:th:dynamic:sample:dense} appear in Sections~\ref{sub:sec:dynamic:sample:sparse} and~\ref{sub:sec:dynamic:sample:dense} respectively. Both the proofs crucially use well known results on $\ell_0$-sampling in a streaming setting (see Lemma~\ref{main:th:l0:sample}) and $w$-wise independent  hash functions (see Lemma~\ref{lm:pagh:hashing}).

\begin{theorem}
\label{main:th:dynamic:sample:sparse}
Starting from an empty graph, we can process the first $T$ updates in a dynamic stream so as to maintain a random  subset of edges $F^{(t)} \subseteq E^{(t)}$ at each time-step $t \leq T$. This requires $\O(n)$ space and $\O(1)$ worst case update time. Furthermore, the following conditions hold with high probability.
\begin{equation}
F^{(t)} = E^{(t)} \text{ at each time-step } t \leq T \text{ with } m^{(t)} \leq 8 \alpha c^2 n \log^2 n.  \label{eq:th:sparse}
\end{equation}
\end{theorem}

Note that the proof of Theorem~\ref{main:th:dynamic:sample:sparse} is by no means obvious, for the following reason. It might happen that we have not included an edge $e$ from the input graph in our sample (since the graph currently contains many edges). However, with the passage of time, a lot of edges get deleted from the graph so as to make it sparse, and at that time we might need to recover the edge $e$.

\begin{theorem}
\label{main:th:dynamic:sample:dense}
Fix any positive integer $s \leq 2 \alpha c n \log n$. Starting from an empty graph, we can process the first $T$ updates in a dynamic stream so as to maintain a random subset of edges $S^{(t)} \subseteq E^{(t)}$ at each time-step $t \in [T]$. This requires $\O(n)$ space and $\O(1)$ worst case update time. 
For every edge $e \in E^{(t)}$, let $X_e^{(t)} \in \{0, 1\}$ be an indicator random variable that is set to one iff $e \in S^{(t)}$.  Then we have:
\begin{enumerate}
\item The following condition holds with high probability. 
\begin{equation}
E\left[X_e^{(t)}\right] \in \left[ (1\pm \epsilon) \cdot \frac{s}{m^{(t)}}\right] \text{ for all edges } e \in E^{(t)}, \text{ at each } t \leq T \text{ with } m^{(t)} \geq 4 \alpha c^2 n \log^2 n. \label{eq:th:dense} 
\end{equation}
\item At each time-step $t \leq T$, the random variables $\{ X_e^{(t)} \}, e \in E^{(t)},$ are negatively associated. 
\item  Insertion/deletion of an edge in $G$ leads to at most two insertion/deletions in the set $S$. 
\end{enumerate}
\end{theorem}

\subsubsection{Proof of Theorem~\ref{main:th:dynamic:sample:sparse}}
\label{sub:sec:dynamic:sample:sparse}

\paragraph{The algorithm.} We define $w = q = 8\alpha c^2 n \log^2 n$,   and build a $w$-wise independent hash function $h : E^* \rightarrow [w]$ as per Lemma~\ref{lm:pagh:hashing}. This  requires $\O(n)$ space, and the hash value $h(e)$ for any given $e \in E^*$ can be evaluated in $O(1)$ time. For all $t \in [T]$ and $i \in [w]$, let $B_i^{(t)}$ denote the  set of edges $e \in E^{(t)}$ with $h(e) = i$. So the edge set $E^{(t)}$ is partitioned into $w$ random subsets $B_1^{(t)}, \ldots, B_w^{(t)}$. 

As per Lemma~\ref{main:th:l0:sample}, for each $i \in [w]$ we run $r = c^2 \log^2 n$ copies of a subroutine called {\sc Streaming-Sampler}. Specifically, for every $i \in [w]$ and $j \in [r]$, the subroutine {\sc Streaming-Sampler}$(i,j)$ maintains a uniformly random sample from the set $B_i^{(t)}$ in $\O(1)$  space and $\O(1)$ worst case update time. Furthermore, the subroutines \{{\sc Streaming-Samplers}$(i,j)$\}, $i \in [w], j \in [r]$, use mutually independent random bits. Let $Y^{(t)}$ denote the collection of the random samples maintained by all these {\sc Streaming-Sampler}s. Since we have multiple {\sc Streaming-Samplers} running on the same set $B_i^{(t)}, i \in [w]$, a given edge can occur multiple times in $Y^{(t)}$. We define $F^{(t)} \subseteq E^{(t)}$  to be the collection of those edges in $E^{(t)}$ that appear at least once in $Y^{(t)}$. Our algorithm maintains the subset $F^{(t)}$ at each time-step $t \in T$. 

\paragraph{Update time.} Suppose that an edge $e$ is inserted into (resp. deleted from) the graph $G = (V,E)$ at  time-step $t \in [T]$. To handle this edge insertion (resp. deletion), we first compute the value $i = h(e)$ in constant time. Then we insert (resp. delete) the edge to the set $B_i^{(t)}$, and call the subroutines {\sc Streaming-Sampler}$(i,j)$ for all $j \in [r]$ so that they can accordingly update the random samples maintained by them. Each {\sc Streaming-Sampler} takes $\O(1)$ time in the worst case to handle an update. Since $r = O(\log^2 n)$, the total time taken by our algorithm to handle an edge insertion/deletion is  $O(r \text{ poly} \log n) = \O(1)$.

\paragraph{Space complexity.} We need $\O(1)$ space to implement each {\sc Streaming-Sampler}$(i,j)$, $i \in [w], j \in [r]$. Since $w = O(n \log^2 n)$ and $r = O(\log^2 n)$, the total space required by all the streaming samplers is $O(w r \text{ poly} \log n) = \O(n)$. Next, note that we can construct the hash function $h$ using $\O(n)$ space. These observations imply that the total space requirement of our scheme is  $\O(n)$. 

\paragraph{Correctness.} It remains to show that with high probability, at each  time-step $t \in [T]$ with $m^{(t)} \leq 8 \alpha c^2  n \log^2 n$, we have $F^{(t)} = E^{(t)}$. 

\medskip
Fix any  time-step $t \in [T]$ with $m^{(t)} \leq 8 \alpha c^2  n \log^2 n$. Consider any $i \in [w]$. The probability that any given edge $e \in E^{(t)}$ has $h(e) = i$  is equal to $1/w$. Since $w = 8 \alpha c^2 n \log^2 n$, the linearity of expectation implies that  $E\left[|B_i^{(t)}|\right] = m^{(t)}/w \leq 1$. Since the hash function $h$ is $w$-wise independent, and since  $m^{(t)} \leq  w$, we can apply the Chernoff bound and infer that    $|B_i^{(t)}| \leq c \log n$ with high probability. Now, a union bound over all $i \in [w]$ shows that with high probability, we have $|B_i^{(t)}| \leq c \log n$ for all $i \in [w]$. Let us call this event $\mathcal{E}^{(t)}$.

Condition on the event $\mathcal{E}^{(t)}$. Fix any edge $e \in E^{(t)}$. Let  $h(e) = i$, for some $i \in [w]$. We know that  $e \in B_i^{(t)}$, that there are at most $c \log n$ edges in $B_i^{(t)}$, and that  our algorithm runs $r = c^2 \log^2 n$  many {\sc Streaming-Sampler}s on $B_i^{(t)}$. Each such {\sc Streaming-Sampler} maintains (independently of others) a uniformly random sample from $B_i^{(t)}$. Consider the event where the edge $e$ is not picked in any of these random samples. This event occurs with probability at most $(1-1/(c\log n))^{c^2 \log^2 n} \leq 1/n^c$. 

In other words, conditioned on the event $\mathcal{E}^{(t)}$,  an edge $e \in E^{(t)}$ appears in $F^{(t)}$ with high probability. Taking a union bound over all $e \in E^{(t)}$, we infer that   $F^{(t)} = E^{(t)}$ with high probability, conditioned on the event $\mathcal{E}^{(t)}$. Next, we recall that the event $\mathcal{E}^{(t)}$ itself occurs with high probability. Thus, we get that the event $F^{(t)} = E^{(t)}$ also occurs with high probability. To conclude the proof, we take a union bound over all time-steps $t \in [T]$ with $m^{(t)} \leq 8 \alpha c^2 n \log^2 n$.

\subsubsection{Proof of Theorem~\ref{main:th:dynamic:sample:dense}}
\label{sub:sec:dynamic:sample:dense}


We define $w = 2 cs\log n$,  $q = s$, and build a $w$-wise independent hash function $h : E^* \rightarrow [s]$ as per Lemma~\ref{lm:pagh:hashing}. This  requires $\O(n)$ space, and the hash value $h(e)$ for any given $e \in E^*$ can be evaluated in $O(1)$ time. 

This hash function  partitions  the edge set $E^{(t)}$ into $s$ mutually disjoint buckets $\{B_j^{(t)}\}, j \in [s]$, where the bucket $B_j^{(t)}$ consists of those edges $e \in E^{(t)}$ with $h(e) =j$. For each $j \in [s]$, we  run an independent copy of $\ell_0$-{\sc Sampler}, as per Lemma~\ref{main:th:l0:sample}, that  maintains a uniformly random sample from $B_j^{(t)}$.  The set $S^{(t)}$ consists of the collection of outputs of all these $\ell_0$-{\sc Samplers}. Note that (a) for each $e \in E^*$, the hash value $h(e)$ can be evaluated in constant time~\cite{Pagh}, (b) an edge insertion/deletion affects  exactly one  of the buckets, and (c) the $\ell_0$-{\sc Sampler}  of the affected bucket can be updated in $\O(1)$ time. Thus, we infer that this procedure  handles an edge insertion/deletion in the input graph in $\O(1)$ time, and furthermore, since $s = \O(n)$, the procedure can be implemented in $\O(n)$ space. We now show that this algorithm satisfies the three properties stated in Theorem~\ref{main:th:dynamic:sample:dense}.

\begin{enumerate}
\item Fix any  time-step $t \in [1, T]$ where $m^{(t)} \geq 4 \alpha c^2 n \log^2 n$. Since $s \leq 2 \alpha c n \log n$, we infer that  $m^{(t)} = |E^{(t)}| \geq 2 c s \log n$. Hence, we can  partition (purely as a thought experiment) the edges in $E^{(t)}$ into at most polynomially many groups $\left\{H_{j'}^{(t)}\right\}$, in such a way that the size of each group lies  between $c s \log n$ and $2c s \log n$. Thus, for any $j \in [s]$ and any $j'$, we have $|H_{j'}^{(t)} \cap B_j^{(t)}| \in [c\log n, 2c\log n]$ in expectation. Since the hash function $h$ is $(2cs \log n)$-wise independent, by applying a Chernoff  bound  we infer that with high probability, the value $|H_{j'}^{(t)} \cap B_j^{(t)}|$ is within a  $(1\pm \epsilon)$ factor of its expectation. Applying the union bound over all $j, j'$, we infer that with high probability,  the sizes of all the sets $\left\{H_{j'}^{(t)} \cap B_j^{(t)}\right\}$ are within a $(1\pm \epsilon)$ factor of  their expected values -- let us call this event $\mathcal{R}^{(t)}$. Note that $E[|B_j^{(t)}|] = m^{(t)}/s$ and $|B_j^{(t)}| = \sum_{j'} |B_j^{(t)} \cap H_{j'}^{(t)}|$. Hence, under the event $\mathcal{R}^{(t)}$, for all $j \in [s]$ the quantity  $|B_j^{(t)}|$ is within a  $(1\pm \epsilon)$ factor of $m^{(t)}/s$. Under the same event $\mathcal{R}^{(t)}$, due to the $\ell_0$-{\sc Sampler}s, the probability that a given edge $e \in E^{(t)}$ with $h(e) = j$ (say) becomes part of $S^{(t)}$ is  within a $(1\pm \epsilon)$ factor of  $1/|B_j^{(t)}|$, which, in turn, is within a $(1\pm \epsilon)$ factor of $s/m^{(t)}$. This implies that for any given edge $e \in E^{(t)}$, we have $E\left[X_e^{(t)}\right] = \Pr\left[e \in S^{(t)}\right] \in \left[(1\pm \epsilon) \cdot s / m^{(t)}\right]$ with high probability. 
\item The property of negative association follows from the facts that (a) if two edges are hashed to different buckets, then they are included in $S_i^{(t)}$ in a mutually independent manner, and (b) if they are hashed to the same bucket, then they are never simultaneously included in $S_i^{(t)}$. 
\item Finally, when an edge $e$ is inserted into (resp. deleted from) the input graph, only the $\ell_0$-sampler running on the bucket $B_j$, for $j = h(e)$ gets affected. This implies that a single edge insertion/deletion in the input graph leads to at most two edge insertions/deletions in the random subset of edges $S \subseteq E$. 
\end{enumerate}

\subsection{A high level overview of our algorithm: Sparse and Dense Intervals}
\label{sec:overview}

Our algorithm for Theorem~\ref{th:new:main} will consist of four different components. They are described below.
\begin{itemize}
\item (P1) This subroutine implements an algorithm as per Theorem~\ref{main:th:dynamic:sample:sparse}. 
\item (P2) This subroutine implements $\tilde \Theta(1)$  independent copies of the algorithm  in Theorem~\ref{main:th:dynamic:sample:dense}. 
\item (P3) For each $t \in [1, T]$, this subroutine classifies time-step $t$ as either ``dense'' or ``sparse''. This is done on the fly, i.e., immediately after receiving the $t^{th}$ edge insertion/deletion in $G$. Consequently,  the range $[1,T]$ is partitioned into ``dense'' and ``sparse'' ``intervals'', where a dense (resp. sparse) interval is  a maximal and contiguous block of dense (resp. sparse) time-steps. For example, we say that $[t_0, t_1] \subseteq [1, T]$ is a dense interval iff (a) time-step $t$ is dense for all $t \in [t_0,t_1]$, (b) either $t_0 = 1$ or time-step $(t_0-1)$ is sparse, and (c) either $t_1 = T$ or time-step $(t_1+1)$ is sparse. The sparse time-intervals are defined analogously.  The subroutine ensures the following properties.
\begin{enumerate}
\item We have $m^{(t)} \leq 8 \alpha c^2 n \log^2 n$ for every sparse time-step $t \in [1, T]$. In other words, the input graph has a small number of edges in a sparse time-step. Note that  $8 \alpha c^2 n \log^2 n$ is also the threshold used in Theorem~\ref{main:th:dynamic:sample:sparse}. Thus, with high probability, equation~\ref{eq:th:sparse} holds at all sparse time-steps.
\item We have $m^{(t)} \geq 4 \alpha c^2 n \log^2 n$ for every dense time-step $t \in [1, T]$.   In other words, the input graph has a large number of edges in a dense time-step. Note that $4 \alpha c^2 n \log^2 n$ is also the threshold used in Theorem~\ref{main:th:dynamic:sample:dense} (part 1). Thus, with high probability, equation~\ref{eq:th:dense} holds at all dense time-steps.
\item If a dense interval begins at a time-step $t$, then we have $m^{(t)} = 1+8 \alpha c^2 n \log^2 n$.
\item Every dense (resp. sparse) interval  spans at least $4 \alpha c^2 n \log^2 n$  time-steps, unless it is the interval ending at  $T$.\end{enumerate}
The classification of a time-step as dense or sparse is done according to the procedure outlined in Figure~\ref{fig:sample:classify}. This procedure can be easily implemented in $\tilde \Theta(1)$ space and $\Theta(1)$ update time, since all we need is a counter that keeps track of the number of edges in the input graph while processing the stream of updates. Furthermore, it is easy to check that the procedure in Figure~\ref{fig:sample:classify} ensures all the four properties described above. 
\item (P4) This subroutine takes as input the set of edges maintained by (P1). Furthermore, it also has access to the output of the subroutine (P3). At all sparse time-steps, it maintains a value $\text{{\sc Output}}^{(t)}$ that gives a $(4+\epsilon)$-approximation of $\text{{\sc Opt}}^{(t)}$. See Section~\ref{sec:sparse} for further details.
\item (P5) This subroutine takes as input the  subsets of edges maintained by (P2). Furthermore, it also has access to the output of the subroutine (P3). At all dense time-steps, it maintains a value $\text{{\sc Output}}^{(t)}$ that gives a $(4+\epsilon)$-approximation of $\text{{\sc Opt}}^{(t)}$. See Section~\ref{sec:dense} for further details.
\item Theorem~\ref{th:new:main} follows from Theorem~\ref{th:sample:sparse} and Theorem~\ref{th:sample:dense}. 
\end{itemize}

\begin{figure}[htbp]
\centerline{\framebox{
\begin{minipage}{5.5in}
\begin{tabbing}
01. \ \ \ \ \= The time-step $1$ is classified as sparse. \\
02. \> {\sc For } $t = 2$ to $T$\\
03. \> \ \ \ \ \ \ \ \ \ \= {\sc If} time-step $(t-1)$ was sparse, {\sc Then} \\
04. \> \> \ \ \ \ \ \ \ \ \= {\sc If} $m^{(t)} \leq 8 \alpha c^2 n \log^2 n$, {\sc Then} \\
05. \> \> \> \ \ \ \ \ \ \ \ \= Classify  time-step $t$ as sparse. \\
06. \> \> \> {\sc Else if} $m^{(t)} > 8 \alpha c^2 n \log^2 n$, {\sc Then} \\
07. \> \> \> \> Classify  time-step $t$ as dense. \\
08. \> \> {\sc Else if} time-step $(t-1)$ was dense, {\sc Then} \\
09. \> \> \> {\sc If} $m^{(t)} \geq 4 \alpha c^2 n \log^2 n$, {\sc Then} \\
10. \> \> \> \> Classify time-step $t$ as dense. \\
11. \> \> \> {\sc Else if} $m^{(t)} < 4 \alpha c^2 n \log^2 n$, {\sc Then} \\
12. \> \> \> \> Classify time-step $t$ as sparse. 
\end{tabbing}
\end{minipage}
}}
\caption{\label{fig:sample:classify} CLASSIFY-TIME-STEPS$(.)$.}
\end{figure}

\subsection{Algorithm for sparse intervals}
\label{sec:sparse}

In this section, we show how to maintain a $(4+\epsilon)$-approximation to the value of the densest subgraph during the sparse time-intervals. Specifically, we prove the following theorem.

\begin{theorem}
\label{th:sample:sparse}
There is an algorithm that uses $\O(n)$  space and maintains a value $\text{{\sc Output}}^{(t)}$ at every sparse time-step $t \leq T$. The algorithm  gives the following two guarantees with high probability.
\begin{itemize}
\item   $\text{{\sc Opt}}^{(t)}/(4+\epsilon) \leq \text{{\sc Output}}^{(t)} \leq \text{{\sc Opt}}^{(t)}$ at every sparse time-step $t \leq T$.
\item  The algorithm takes $\O(T)$ time to process the stream of $T$ updates in $G$. In other words, the amortized update time is $\O(1)$. 
\end{itemize}
\end{theorem}

\noindent The algorithm for Theorem~\ref{th:sample:sparse} consists of two major ingredients. 
\begin{itemize}
\item First, we run a subroutine  as per Theorem~\ref{main:th:dynamic:sample:sparse} while processing the stream of $T$ updates.\footnote{Note that this is the same subroutine (P1) from Section~\ref{sec:overview}.} This ensures that with high probability, we  maintain a subset $F^{(t)} \subseteq E^{(t)}$ such that $F^{(t)} = E^{(t)}$ at every sparse time-step $t \leq T$. The worst case update time and space complexity of this subroutine are $\O(1)$ and $\O(n)$ respectively. 
\item Second, we run our dynamic algorithm -- which we refer to as {\sc Dynamic-algo} -- from Section~\ref{sec:dynamic} on the graph $(V, F^{(t)})$ during every sparse time-interval.\footnote{We can identify each sparse time-interval using the subroutine (P3) from Section~\ref{sec:overview}.} Since $F^{(t)} = E^{(t)}$ throughout the duration of such an interval (with high probability),  this allows us to maintain an $\text{{\sc Output}}^{(t)} \in \left[\text{{\sc Opt}}^{(t)}/ (4+\epsilon), \text{{\sc Opt}}^{(t)}\right]$ at every sparse time-step $t \leq T$. As the input graph has $\O(n)$  edges at every sparse  time-step,  the space complexity of {\sc Dynamic-algo} is   $\O(n)$. 
\end{itemize}

\medskip
\noindent It remains to analyze the amortized  update time of {\sc Dynamic-algo}. Towards this end, fix any sparse time-interval $[t_0, t_1]$, and let $C(t_0,t_1)$ denote the  amount of computation performed during this interval by the subroutine {\sc Dynamic-algo}. Consider two possible cases.
\begin{itemize} 
\item  {\em Case 1.} ($t_1 < T$)

In this case,  our analysis from  Section~\ref{sec:dynamic} implies that $C(t_0,t_1) = \O(n  + (t_1-t_0))$. Since $t_1 < T$, the subroutine (P3) from Section~\ref{sec:overview} ensures that  $(t_1-t_0) = \Omega(n)$.\footnote{See item 4 in the description of the subroutine (P3).} This gives us the guarantee that  $C(t_0,t_1) = \O(t_1-t_0)$.
\item {\em Case 2.} ($t_1 = T$)

In this case,  the sparse time-interval under consideration  ends at $T$. Thus, if  $(t_1 - t_0) =  o(n)$, then we would have  $C(t_0,t_1) = \O(n)$. Else if $(t_1 - t_0) = \Omega(n)$, then applying a similar argument as in Case 1, we get $C(t_0, t_1) = \tilde O(t_1 - t_0)$.
\end{itemize}

\noindent Let $\left[t_0^i, t_1^i\right]$ denote the $i^{th}$ sparse time-interval, and let $C$ denote the amount of computation performed by {\sc Dynamic-algo} during the entire time-period $[1,T]$. Since the sparse time-intervals are mutually disjoint, and since there can be at most one sparse time-interval that ends at $T$, we get the following guarantee.
\begin{eqnarray}
\label{eq:2:thought}
C & = & \left(\sum_{i} \O((t_1^i - t_0^i))\right) + \O(n)  \leq  \O(T) + \O(n)  =  \O(T)
\end{eqnarray}
The last equality holds as $T = \Theta(n^{\lambda})$ and  $\lambda \geq 1$ (equation~\ref{eq:new:lambda}). 
This shows that the amortized update time of the algorithm is $\O(1)$, thereby concluding the proof of Theorem~\ref{th:sample:sparse}.

\subsection{Algorithm for dense intervals}
\label{sec:dense}

In this section, we show how to maintain a $(4+\epsilon)$-approximation to the value of the densest subgraph during the dense time-intervals. Specifically, we prove the following theorem.

\begin{theorem}
\label{th:sample:dense}
There is an algorithm that uses $\O(n)$  space and maintains a value $\text{{\sc Output}}^{(t)}$ at every dense time-step $t \leq T$. The algorithm  gives the following two guarantees with high probability.
\begin{itemize}
\item   $\text{{\sc Opt}}^{(t)}/(4+\epsilon) \leq \text{{\sc Output}}^{(t)} \leq \text{{\sc Opt}}^{(t)}$ at every dense time-step $t \leq T$.
\item  The algorithm takes $\O(T)$ time to process the stream of $T$ updates in $G$. In other words, the amortized update time is $\O(1)$. 
\end{itemize}
\end{theorem}

\subsubsection{Basic building block for proving Theorem~\ref{th:sample:dense}}
\label{sec:new:building:block}

For every time-step $t \in [1, T]$, we define:
\begin{equation}
\label{eq:building:block:pi}
\pi^{(t)} = m^{(t)}/(2\alpha n), \text{ and } \sigma = 2(1+\epsilon)n
\end{equation}
Accordingly,  Lemma~\ref{main:lm:stream:range} implies that:  
\begin{equation}
\label{eq:building:block:range}
\alpha \cdot \pi^{(t)} < \text{{\sc Opt}}^{(t)} <  \sigma/(2(1+\epsilon))
\end{equation} 
We now discretize the range $[\pi^{(t)}, \sigma]$ in powers of $(1+\epsilon)$ as in Corollary~\ref{main:cor:test:1}. Specifically, we define:
\begin{eqnarray}
\label{eq:building:block:K}
K = 2 + \lceil \log_{(1+\epsilon)} (\sigma \cdot (2\alpha n)) \rceil \\
d_{k}^{(t)} = (1+\epsilon)^{k-1} \cdot \pi^{(t)} \text{ for all } k \in [K] \label{eq:building:block:dk}
\end{eqnarray}
Note that $m^{(t)} \geq 1$ at every dense time-step $t \leq T$. This ensures that $\pi^{(t)} \geq 1/(2\alpha n)$, and hence we have $K \geq 2 + \lceil \log_{(1+\epsilon)} (\sigma/\pi^{(t)}) \rceil$ during those time-steps. In other words, the values of $\pi^{(t)}, \sigma$ and $K$ satisfy the condition stated in Corollary~\ref{cor:test:1} during the dense time-intervals. Also note that the value of $K$ {\em does not} depend on the time-step $t$ under consideration.\footnote{This is especially important since we want to maintain an  $(\alpha, d_k^{(t)}, L)$-decomposition for each $k \in [K]$ during the dense time-intervals. If  $K$ were a function of $t$, then  the number of such decompositions would  have varied over time. } Furthermore, we have $K = \O(1)$.

We want to maintain a $2\alpha (1+\epsilon)^3 = (4+\Theta(\epsilon))$-approximation of the maximum density during the dense time-intervals.  For this purpose it suffices to maintain an $(\alpha, d_{k}^{(t)}, L)$-decomposition of $G^{(t)}$ for each $k \in [K]$ (see Corollary~\ref{cor:test:1}). 
 
 \medskip
 Thus, we conclude that Theorem~\ref{th:sample:dense}  follows from Theorem~\ref{main:th:dynamic:stream:special}. Accordingly, for the rest of Section~\ref{sec:dense}, we fix an index $k \in [K]$, and focus on proving Theorem~\ref{main:th:dynamic:stream:special}.

\begin{theorem}
\label{main:th:dynamic:stream:special}
Fix any integer $k \in [K]$. There is a dynamic algorithm that uses $\O(n)$ bits of space and maintains $L$ subsets of nodes $V = Z_1^{(t)} \supseteq \cdots \supseteq Z_L^{(t)}$  at every dense time-step $t \in [T]$. The algorithm is randomized and gives the following two guarantees with high probability.
\begin{itemize}
\item  The tuple $(Z_1^{(t)} \ldots Z_L^{(t)})$ is an $(\alpha,  d_k^{(t)}, L)$-decomposition of $G^{(t)}$ at every dense time-step $t \in [T]$.
\item  The algorithm takes $\O(T)$ time to process the stream of $T$ updates. In other words, the algorithm has an amortized update time of $\O(1)$. 
\end{itemize}
\end{theorem}

\subsubsection{Overview of our algorithm for Theorem~\ref{main:th:dynamic:stream:special}} 
\label{sec:dynamic:stream:overview}

Our algorithm for proving Theorem~\ref{main:th:dynamic:stream:special} consists of two major ingredients.
\begin{itemize}
\item First,  at each dense time-step $t \leq T$ we maintain a collection of  $(L-1)$ random subsets of edges $S_1^{(t)}, \ldots, S_{L-1}^{(t)} \subseteq E^{(t)}$. To maintain these random subsets, we need to run $(L-1)$ mutually independent copies of the algorithm in Theorem~\ref{main:th:dynamic:sample:dense} (for an appropriate value of $s$).
\item Second, using the random subsets $S_1^{(t)}, \ldots, S_{L-1}^{(t)} \subseteq E^{(t)}$, we maintain an $(\alpha, d_k^{(t)}, L)$ decomposition $(Z_1^{(t)}, \ldots, Z_L^{(t)})$ during the dense time-steps. Specifically, we set $Z_1^{(t)} = V$. Next,  for $i = 1$ to $(L-1)$, we construct the node set $Z_{i+1}^{(t)} \subseteq Z_i^{(t)}$ by looking at the degrees $D_v(Z_i^{(t)}, S_i^{(t)})$ of the nodes  $v \in Z_i^{(t)}$ among the edges $e \in S_i^{(t)}$. This follows the spirit of the construction in Section~\ref{sec:sketch} (but note that there we did not concern ourselves with  the update time). 

From the proof of Lemma~\ref{main:lm:stream:1}, we infer that  each  $S_i^{(t)}$ should contain $s^{(t)} = c m^{(t)} \log n/d_k^{(t)}$ random samples from the edge set $E^{(t)}$. Equations~\ref{eq:building:block:pi} and~\ref{eq:building:block:dk} ensure that $s^{(t)} = 2 \alpha c n \log n/(1+\epsilon)^{k-1}$, which means that the value of $s^{(t)}$ is independent of the time-step $t$ under consideration. Accordingly, for the rest of Section~\ref{sec:dense}, we omit the superscript $t$ from $s^{(t)}$ and define:
\begin{equation}
\label{eq:overview:dense:st}
s = 2 \alpha c n \log n/(1+\epsilon)^{k-1} = c m^{(t)} \log n/d_k^{(t)} \text{ for each dense time-step } t \leq T.
\end{equation}
Furthermore, since $k \geq 1$, we observe that:
\begin{equation}
\label{eq:overview:dense:s}
s \leq 2 \alpha c n \log n
\end{equation}
Thus, the value assigned to $s$ satisfies the condition dictated by Theorem~\ref{main:th:dynamic:sample:dense}.
\end{itemize}
To summarize, our first subroutine maintains $(L-1)$ mutually independent copies of the algorithm in Theorem~\ref{main:th:dynamic:sample:dense} (for the value of $s$ as defined by equations~\ref{eq:overview:dense:st},~\ref{eq:overview:dense:s}). The random subsets of edges maintained by them are denoted as $S_1^{(t)}, \ldots, S_{L-1}^{(t)}$. Since $L = \tilde \Theta(1)$ (see equation~\ref{eq:new:L}),  Theorem~\ref{main:th:dynamic:sample:dense} ensures that this subroutine has $\O(1)$ worst case update time and $\O(n)$ space complexity.

Next, we present another subroutine {\sc Dynamic-stream} (see Section~\ref{sec:dynamic:stream:describe}) that is invoked only during the dense time-intervals.\footnote{The dense time-intervals can be easily identified using the subroutine (P3) from Section~\ref{sec:overview}.} This subroutine can access  the random subsets $\left\{S_i^{(t)}\right\}, i \in [L-1]$, and it maintains $L$ subsets of nodes $V = Z_1^{(t)} \supseteq \cdots \supseteq Z_L^{(t)}$ at each dense time-step $t \in [T]$. 

\paragraph{Roadmap.} The rest of Section~\ref{sec:dense} is organized as follows. 
\begin{itemize}
\item In Section~\ref{sec:dynamic:stream:describe}, we fix a dense time-interval $[t_0, t_1]$, and describe how the subroutine {\sc Dynamic-stream} processes the edge insertions/deletions in the input graph during this interval. We show that {\sc Dynamic-stream} maintains an $(\alpha, d_k^{(t)}, L)$-decomposition of the input graph $G^{(t)}$ throughout the duration of this dense time-interval (see Lemma~\ref{main:lm:performance}).
\item Section~\ref{main:sec:dynamic:implement} presents the data structures for implementing the subroutine {\sc Dynamic-stream}. 
\item In Section~\ref{sec:main:old}, we make some preliminary observations about the running time of the subroutine {\sc Dynamic-stream}, and analyze its space complexity (see Lemma~\ref{main:lm:sample:space}). 
\item In Section~\ref{main:sec:dynamic:update-time}, we analyze the amortized update time of the subroutine {\sc Dynamic-stream} (see Corollary~\ref{main:cor:dynamic:stream:runtime}). Theorem~\ref{main:th:dynamic:stream:special} follows from Lemma~\ref{main:lm:performance}, Lemma~\ref{main:lm:sample:space} and Corollary~\ref{main:cor:dynamic:stream:runtime}. 
\item In Sections~\ref{sec:extreme:new} and~\ref{sec:extreme:new:1} are devoted to the proof of two lemmas stated in Section~\ref{main:sec:dynamic:update-time}.
\end{itemize}

\subsubsection{The subroutine {\sc Dynamic-stream} for a dense time-interval $[t_0, t_1]$}
\label{sec:dynamic:stream:describe}
\label{main:sec:dynamic:algo}

Fix any dense time-interval $[t_0,t_1] \subseteq [1, T]$. We will  maintain $L$ subsets of nodes $V = Z_1^{(t)} \supseteq \cdots \supseteq Z_L^{(t)}$ at each time-step $t \in [t_0, t_1]$. With high probability, we will  prove that throughout the interval the tuple $(Z_1^{(t)}, \ldots, Z_L^{(t)})$ remains a valid $(\alpha,  d_k^{(t)}, L)$-decomposition of  $G^{(t)} = (V, E^{(t)})$.

We run $(L-1)$ mutually independent copies of the algorithm stated in Theorem~\ref{main:th:dynamic:sample:dense} (for the value of $s$ as defined by equations~\ref{eq:overview:dense:st},~\ref{eq:overview:dense:s}). Hence, at each time-step $t \in [t_0, t_1]$,   we can access the mutually independent random subsets of edges $S_1^{(t)}, \ldots, S_{L-1}^{(t)} \subseteq E^{(t)}$ as defined in Section~\ref{sec:dynamic:stream:overview}. 

\paragraph{Initialization in the beginning of the dense time-interval $[t_0, t_1]$.} \

\medskip
\noindent
Just before time-step $t_0$, we perform the initialization step outlined in Figure~\ref{fig:sample:init}. It ensures that $Z_1^{(t_0-1)} = V$ and $Z_i^{(t_0-1)} = \emptyset$ for all $i \in \{2, \ldots, L\}$.

\begin{figure}[htbp]
\centerline{\framebox{
\begin{minipage}{5.5in}
\begin{tabbing}
01. \ \ \ \ \= Set $Z_1^{(t_0-1)} \leftarrow V$. \\
02. \> {\sc For } $i = 2$ to $L$\\
03. \> \ \ \ \ \ \ \ \ \ \= Set  $Z_i^{(t_0-1)} \leftarrow \emptyset$. 
\end{tabbing}
\end{minipage}
}}
\caption{\label{fig:sample:init} INITIALIZE$(.)$.}
\end{figure}

\paragraph{Handling an update in $G$ during the dense time-interval $[t_0, t_1]$.}  \

\medskip
\noindent Consider  an edge insertion/deletion  in the input graph during time-step $t \in [t_0, t_1]$. The edge set $E^{(t)}$ is different from the edge set $E^{(t-1)}$. Accordingly, for all $i \in [L-1]$, the subset of edges $S_i^{(t)}$ may  differ from the subset of edges $S_i^{(t-1)}$. Therefore, at this time we call the subroutine RECOVER-SAMPLE$(t)$ in Figure~\ref{main:fig:sample:recover}. Its input is the old decomposition $(Z_1^{(t-1)}, \ldots, Z_L^{(t-1)})$. Based on this old decomposition, and the new samples $\left\{S_i^{(t)}\right\}, i \in [L-1]$, the   subroutine constructs a new decomposition $(Z_1^{(t)}, \ldots, Z_L^{(t)})$.


As in Section~\ref{sec:sketch}, we want to ensure that the node set $Z_i$ is completely determined by the outcomes of the random samples  in $\{S_j\}, j < i$ (see the proof of Lemma~\ref{main:lm:stream:1}).  Towards this end, we observe the following lemma.

\begin{figure}[htbp]
\centerline{\framebox{
\small
\begin{minipage}{5.5in}
\begin{tabbing}
01. \ \ \ \ \= Set $Z_1^{(t)} \leftarrow V$. \\
02. \> {\sc For } $i = 1$ to $L$\\
03. \> \ \ \ \ \ \ \ \ \ \= Set  $Y_i \leftarrow Z_i^{(t-1)}$. \\
04. \> {\sc For } $i = 1$ to $(L-1)$ \\
05.\>\ \ \ \ \ \ \ \ \= Let $A_i^{(t)}$ be the set of nodes  $y \in Z_i^{(t)}$ having $\dd_{y}(Z_i^{(t)}, S_i^{(t)}) > (1-\epsilon)^2 \alpha c \log n$. \\
06.\>\> Let $B_i^{(t)}$ be the set of nodes $y \in Z^{(t)}_{i}$ having $\dd_{y}(Z_i^{(t)}, S_i^{(t)}) < (1+\epsilon)^2 c \log n$. \\
07. \> \> Set $Y_{i+1} \leftarrow Y_{i+1} \cup A_i^{(t)}$. \\
08.\> \> {\sc For all} $j = (i+1)$ to $(L-1)$ \\
09.\> \> \ \ \ \ \  \ \ \ \= Set $Y_{j} \leftarrow Y_{j}  \setminus B_i^{(t)}$. \\
10. \> \> Set $Z_{i+1}^{(t)} \leftarrow Y_{i+1}$.
\end{tabbing}
\end{minipage}
}}
\caption{\label{main:fig:sample:recover} RECOVER-SAMPLE$(t)$. }
\end{figure} 

 \begin{lemma}
\label{main:lm:independence}
Fix any time-step  $t \in [t_0,t_1]$,  any index $i \in [L-1]$, and consider the set $Z_i^{(t)}$ as defined by the procedure in  Figure~\ref{main:fig:sample:recover}.
\begin{enumerate}
\item  The node set $Z_{i}^{(t)}$ is completely determined by the contents of  the sets $\left\{S_{j}^{(t)}\right\}, j < i$.
\item  The contents of the random sets  $\left\{S_j^{(t)}\right\}, j \geq i$, are independent  of the contents of the  set $Z_i^{(t)}$.
\end{enumerate}
\end{lemma}

\begin{proof}
Follows from the description of the procedure in Figure~\ref{main:fig:sample:recover}. 
\end{proof}

We now prove the correctness of our algorithm. Specifically, we show that with high probability the tuple $(Z_1^{(t)}, \ldots, Z_L^{(t)})$ remains a valid $(\alpha, d_k^{(t)}, L)$-decomposition of the input graph throughout the dense time-interval under consideration.

\begin{lemma}
\label{main:lm:performance}
With high probability, at each time-step $t \in [t_0,t_1]$ the tuple $(Z_1^{(t)} \ldots Z_L^{(t)})$ maintained by the subroutine {\sc Dynamic-stream} gives an $(\alpha,  d_k^{(t)}, L)$-decomposition of the input graph $G^{(t)}$.
\end{lemma}

\begin{proof}
For each time-step $t \in [t_0, t_1]$ and index  $i \in [L-1]$, we define an event  $\mathcal{E}_i^{(t)}$ as follows.
\begin{itemize}
\item The event $\mathcal{E}_i^{(t)}$ occurs iff the following two conditions hold simultaneously. 
\begin{itemize}
\item  $Z_{i+1}^{(t)} \supseteq \{ v \in Z_i^{(t)} : \dd_v(Z_i^{(t)}, E^{(t)}) > \alpha d_k^{(t)} \}$, and 
\item  $Z_{i+1}^{(t)} \cap \{ v \in Z_i^{(t)} : \dd_v(Z_i^{(t)}, E^{(t)}) <  d_k^{(t)} \} = \emptyset$.
\end{itemize}
\end{itemize} 
We now define another event $\mathcal{E}^{(t)}$ for each time-step $t \in [t_0, t_1]$. 
\begin{itemize}
\item $\mathcal{E}^{(t)} = \bigcap_{i=1}^{L-1} \mathcal{E}_i^{(t)}$.
\end{itemize}
\noindent By Definition~\ref{main:def:partition}, the tuple $(Z_1^{(t)} \ldots Z_L^{(t)})$ is an $(\alpha, d_k^{(t)}, L)$-decomposition of $G^{(t)}$ iff the event $\mathcal{E}^{(t)}$ occurs. Below, we show that $\Pr\left[\mathcal{E}_i^{(t)}\right] \geq 1-1/(\text{poly }n)$ for any given $i \in \{1, \ldots, L-1\}$ and $t \in [t_0, t_1]$.  Taking a union bound over all $i \in \{1, \ldots, L-1\}$, we get $\Pr\left[\mathcal{E}^{(t)}\right] \geq 1-1/(\text{poly }n)$ at each time-step $t \in [t_0, t_1]$.  Hence, the lemma follows by taking a union bound over all $t \in [t_0, t_1]$.

\medskip
\noindent For the rest of the proof, fix any time-step $t \in [t_0, t_1]$ and index $i \in \{1, \ldots, L-1\}$.
\begin{itemize}
\item Fix any instance of the random set $Z_i^{(t)}$ and condition on this event. 
\begin{itemize}
\item  By Theorem~\ref{main:th:dynamic:sample:dense} and equation~\ref{eq:overview:dense:st}, each edge $e \in E^{(t)}$  appears in $S_i^{(t)}$ with probability $(1\pm \epsilon) s/m^{(t)} = (1 \pm \epsilon) c \log n/d_k^{(t)}$. Furthermore, these events are negatively associated (see Section~\ref{sec:prelim:concentration}). \\

Consider any node $v \in Z_i^{(t)}$ with $\dd_v(Z_i^{(t)}, E^{(t)}) > \alpha d_k^{(t)}$. By linearity of expectation: 
$$E\left[\dd_v(Z_i^{(t)}, S_i^{(t)})\right] \geq \alpha d_k^{(t)} \cdot (1-\epsilon) c \log n/d_k^{(t)} =  (1-\epsilon) \alpha c \log n.$$ 
Since  the contents of the random  set $S_i^{(t)}$ are independent of the contents of $Z_i^{(t)}$ (see Lemma~\ref{main:lm:independence}), we can apply a Chernoff bound on this expectation, and  derive that: 
\begin{equation}
\Pr\left[\dd_v(Z_i^{(t)}, S_i^{(t)}) > (1-\epsilon)^2 \alpha c \log n \, | \, Z_i^{(t)}\right] \geq 1- 1/(\text{poly }n) \label{eq:last:one:1}
\end{equation}
Now,  Figure~\ref{main:fig:sample:recover} implies that if $\dd_v(Z_i^{(t)}, S_i^{(t)}) > (1-\epsilon)^2 \alpha c \log n$, then the node $v$ becomes part of $Z_{i+1}^{(t)}$. Thus, applying equation~\ref{eq:last:one:1} we get:
\begin{eqnarray}
\Pr\left[v \in Z_{i+1}^{(t)} \, | \, Z_i^{(t)}\right]  \geq  \Pr\left[\dd_v(Z_i^{(t)}, S_i^{(t)}) > (1-\epsilon)^2 \alpha c \log n \, | \, Z_i^{(t)}\right]  \geq 1 - 1/(\text{poly } n) \label{eq:last:one:2}
\end{eqnarray}
Note that equation~\ref{eq:last:one:2} would have been true even if $v$ did not belong to $Z_i^{(t)}$. Furthermore, equation~\ref{eq:last:one:2} holds regardless of the event $\mathcal{E}^{(t-1)}$. 

\medskip
Next, consider any node $u \in Z_i^{(t)}$ with $\dd_u(Z_i^{(t)}, E^{(t)}) < d_k^{(t)}$. 
A similar argument shows: 
\begin{equation}
\label{eq:last:one:3}\Pr\left[u \notin Z_{i+1}^{(t)} \, | \, Z_i^{(t)} \right] \geq \Pr\left[\dd_u(Z_i^{(t)}, E^{(t)}) < (1+\epsilon)^2 c \log n \, | \, Z_i^{(t)}\right] \geq 1- 1/(\text{poly }n)
\end{equation} 
Note that equation~\ref{eq:last:one:3} would have been true even if $u$ did not belong to $Z_i^{(t)}$. Furthermore, equation~\ref{eq:last:one:3} holds regardless of the event $\mathcal{E}^{(t-1)}$. 

\medskip
 Thus, applying a union bound on equations~\ref{eq:last:one:2},~\ref{eq:last:one:3} over all  nodes in $Z_i^{(t)}$, we infer that: 
$$\Pr\left[\mathcal{E}_i^{(t)} \, | \, Z_i^{(t)}\right] \geq 1-1/(\text{poly }n).$$ 
\end{itemize}
\end{itemize}
Since the  guarantee $\Pr\left[\mathcal{E}_i^{(t)} \, | \, Z_i^{(t)}\right] \geq 1-1/(\text{poly }n)$ holds for every possible instantiation  of $Z_i^{(t)}$, we get $\Pr\left[\mathcal{E}_i^{(t)}\right] \geq 1-1/(\text{poly }n)$. Taking a union bound over all indices $i \in \{1, \ldots, L-1\}$, we infer that $\Pr\left[\mathcal{E}^{(t)}\right] = \bigcap_{i=1}^{L-1} \mathcal{E}_i^{(t)} \geq 1-1/(\text{poly }n)$.

 In other words,  at every time-step $t \in [t_0, t_1]$ the event $\mathcal{E}^{(t)}$ occurs with high probability, and this holds {\em regardless of the  past events $\mathcal{E}^{(t')}, t' < t$}. Hence, taking a union bound over all time-steps in the interval $[t_0, t_1]$, we get: With high probability, for all $t \in [t_0, t_1]$ the tuple $(Z_1^{(t)} \ldots Z_L^{(t)})$ maintained by the subroutine {\sc Dynamic-stream} gives an $(\alpha,  d_k^{(t)}, L)$-decomposition of the input graph $G^{(t)}$. This concludes the proof of the lemma. 
 \end{proof}

\subsubsection{Data structures for implementing the procedure in Figure~\ref{main:fig:sample:recover}} 
\label{main:sec:dynamic:implement}

Recall the notations introduced in Section~\ref{sec:prelim}.  
\begin{itemize}
\item Consider any node $v \in V$ and any $i \in \{1, \ldots, L-1\}$.  We maintain the doubly linked lists $\left\{\text{{\sc Friends}}_i[v,j]\right\}, 1 \leq  j \leq L-1$ as defined below. Each of these lists is defined by the neighborhood of $v$ induced by the sampled edges in $S_i$. Recall  Definition~\ref{main:def:partition}.
\begin{itemize}
\item If $i \le \ell(v)$, then we have:
\begin{itemize}
\item $\text{{\sc Friends}}_i[v, j]$ is empty for all $j > i$.
\item $\text{{\sc Friends}}_{i}[v,j] = \NN_v(Z_j, S_{i})$ for $j = i$. 
\item $\text{{\sc Friends}}_{i}[v,j] = \NN_v(V_j, S_{i})$ for all $j < i$.
\end{itemize}
\item Else if $i > \ell(v)$, then we have:
\begin{itemize}
\item $\text{{\sc Friends}}_i[v,j]$ is empty for all $j > \ell(v)$.
\item $\text{{\sc Friends}}_{i}[v,j] = \NN_v(Z_{j}, S_{i})$ for $j = \ell(v)$.
\item  $\text{{\sc Friends}}_i[v,j] = \NN_v(V_j, S_{i})$ for all $j < \ell(v)$.
\end{itemize}
\end{itemize}
\end{itemize}

\noindent For every node $v \in V$, we maintain a counter $\text{{\sc Degree}}_i[v]$, which keeps track of  the number of nodes in $\text{{\sc Friends}}_i[v,i]$. Note that if $\ell(v) < i$, then this counter equals zero. Further, we maintain a doubly linked list $\text{{\sc Dirty-Nodes}}[i]$. This   consists of all the nodes $v \in V$ having either  \\  $\big\{$$\text{{\sc Degree}}_i[v] > (1-\epsilon)^2 \alpha c \log n$ and $\ell(v) = i$$\big\}$ or $\big\{$$\text{{\sc Degree}}_i[v]$ $< (1+\epsilon)^2 c \log n$ and $\ell(v) > i$$\big\}$.

\paragraph{Remark.} Note that for any given index $i \in \{1, \ldots, L-1\}$ and any time-step $t \leq T$, an edge $e \in E^{(t)}$ of the input graph appears at most once among the samples in $S_i^{(t)}$ (see Theorem~\ref{main:th:dynamic:sample:dense}). Thus, the number of occurrences of an edge among the samples $S_1^{(t)}, \ldots, S_{L-1}^{(t)}$ is at most $(L-1)$. 

\subsubsection{Implementing the procedure in Figure~\ref{main:fig:sample:recover} during a dense time-interval $[t_0, t_1]$}
\label{sec:main:old} Consider any dense time-interval $[t_0, t_1] \subseteq [1, T]$, and fix any time-step $t \in [t_0,t_1]$. The $t^{th}$ edge insertion/deletion in the input graph might lead to some changes in the random subsets of edges $S_1, \ldots, S_{L-1} \subseteq E$. However, Theorem~\ref{main:th:dynamic:sample:dense} implies that an edge insertion/deletion in $G$ can lead to at most two edge insertions/deletions in $S_i$, for all $i \in \{1, \ldots, L-1\}$. Thus, due to the $t^{th}$ update in the stream, there can be at most $O(L) = \O(1)$ insertions/deletions in the random sets $S_1, \ldots, S_{L-1}$ (see equation~\ref{eq:new:L}). After each such edge insertion/deletion in an $S_i, i \in \{1, \ldots, L-1\}$, we update the relevant data structures described in Section~\ref{main:sec:dynamic:implement}. Since an edge  $(u,v) \in S_i$ can potentially appear only in the lists $\text{{\sc Friends}}_i[x, j]$ with $x \in \{u,v\}$ and $j \in \{1, \ldots, L-1\}$ (see Section~\ref{main:sec:dynamic:implement}),  we reach the following conclusion:
\begin{itemize}

\item When an edge insertion/deletion in $G$ leads to changes in the random edge sets $S_i, \ldots, S_{L-1} \subseteq E$, we can update the {\sc Friends} and {\sc Dirty-Nodes} lists from Section~\ref{main:sec:dynamic:implement} in $\O(1)$ time.  
\end{itemize}
After updating the edge sets $S_1, \ldots, S_{L-1}$ and the {\sc Friends} and {\sc Dirty-Nodes} lists, we run the procedure described in Figure~\ref{main:fig:sample:recover}. Now, consider the $i^{th}$ iteration of the main {\sc For} loop (Steps~05-10)  in Figure~\ref{main:fig:sample:recover}, for some index $i \in \{1, \ldots, L-1\}$. The purpose of this iteration is to construct the set $Z_{i+1}^{(t)}$, based on the sets $Z_i^{(t)}$ and $S_i^{(t)}$. Below, we state an alternate way of visualizing this iteration.

We scan through the list of nodes $u$ with $\ell(u) = i$ and $\dd_u(Z_i^{(t)}, S_i^{(t)}) > (1-\epsilon)^2 \alpha c \log n$.   While considering each such node $u$, we increment its level from $i$ to $(i+1)$. This takes care of the Steps~(05) and~(07).  Next, we scan through the list of nodes $v$ with $\ell(v) > i$ and $\dd_v(Z_i^{(t)}, S_i^{(t)}) < (1+\epsilon)^2 c \log n$. While considering any such node $v$ at level $\ell(v) = j_v > i$ (say), we decrement its level from $j_v$ to $i$. This takes care of the Steps~(06),~(08) and~(09).

Note that the nodes undergoing a level-change in the preceding paragraph are precisely the ones that appear  in the list $\text{{\sc Dirty-Nodes}}[i]$ just before the $i^{th}$ iteration of the main {\sc For} loop. Thus, we can implement Steps~(05-10) as follows: Scan through the nodes $y$ in $\text{{\sc Dirty-Nodes}}[i]$ one after another. While considering any such  node $y$, change its level as per Figure~\ref{main:fig:sample:recover}, and then update the relevant data structures to reflect this change.

The next lemma states the space complexity of this procedure. 

\begin{lemma}
\label{main:lm:sample:space}
Our algorithm  in Figure~\ref{main:fig:sample:recover} can be implemented in $\O(n)$ space.
\end{lemma}

\begin{proof}
The amount of space needed is dominated by the number of edges in $\left\{S_i^{(t)}\right\}, i \in [L-1]$. Since $|S_i^{(t)}| \leq s$ for each $i \in [L-1]$, the space complexity is $(L-1) \cdot s = \O(n)$ (see equations~\ref{eq:new:L},~\ref{eq:overview:dense:s}).
\end{proof}

The claim below  bounds  the time taken by a single iteration of the main {\sc For} loop in Figure~\ref{main:fig:sample:recover}. This  will be crucially used to analyze the overall update time of our algorithm in Section~\ref{main:sec:dynamic:update-time}. 

\begin{claim}
\label{main:cl:sample:runtime:up}
Fix any time-step $t \in [t_0, t_1]$, and consider the $i^{th}$ iteration of the main {\sc For} loop in Figure~\ref{main:fig:sample:recover} for some $i \in \{1, \ldots, L-1\}$. Consider  two nodes $u, v \in Z_i^{(t)}$ such that:
\begin{itemize}
\item (a) the level of $u$ is increased from $i$ to $(i+1)$ in Step~(07), and 
\item (b) the level of $v$ is decreased to $i$ in Steps~(08--09). 
\end{itemize}
Updating the relevant data structures for this step require $\sum_{i' > i} O(1+\dd_y(Z_i^{(t)}, S_{i'}^{(t)}))$ time, where $y = u$ (resp. $v$) in the former (resp. latter) case.
\end{claim}

\begin{proof}
Follows from the fact that we only need to update the lists $\text{{\sc Friends}}_{i'}[x, j]$ where $i' > i$,  $x \in \{y \} \cup \N_y(Z_i^{(t)}, S_{i'}^{(t)})$ and $j \in \{i,i+1\}$.
\end{proof}

\subsubsection{The amortized update time of \text{{\sc Dynamic-Stream}} during a dense time-interval}
\label{main:sec:dynamic:update-time}

In this section, we bound the amortized update time of the subroutine {\sc Dynamic-stream} during a dense time-interval $[t_0, t_1]$.  Recall that the subroutine {\sc Dynamic-stream}  maintains an $(\alpha, d_k^{(t)}, L)$-decomposition $(V_1^{(t)}, \ldots, V_L^{(t)})$ of the input graph $G^{(t)}$ throughout the duration of such an interval.
To bound the amortized update time, we use a potential function $\B$ as defined in equation~\ref{main:eq:potential:main}. Note that the potential $\B$ is uniquely determined by the assignment of the nodes $v \in V$ to the levels $\{1, \ldots, L\}$ and by the contents of the random sets $S_1, \ldots, S_{(L-1)}$.  For all nodes $v \in V$, we define:
\begin{eqnarray}
\Gamma_i(v) & = & \max(0, (1-\epsilon)^2 \alpha c \log n - \dd_v(Z_i, S_i)) \text{ for all } i \in \{1, \ldots, \ell(v) -1\} \label{main:eq:potential:new} \\
\Phi(v)  & = &  (L / \eps) \cdot \sum_{i = 1}^{\ell(v)-1} \Gamma_i(v)  \label{main:eq:potential:node} 
\end{eqnarray}
For all $u, v \in V$, let $f(u,v) = 1$ if $\ell(u) = \ell(v)$ and $0$ otherwise.
Also, let $r_{uv}=\min(\ell(u), \ell(v))$.  For all $i \in \{1, \ldots, L-1\}$, $(u,v) \in S_i$, we define:
\begin{eqnarray}
\Psi_i(u,v)  &=&  \begin{cases} 0 & \text{if }  r_{uv}  \geq i;  \\
2 \cdot \left(i - r_{uv} \right) +   f(u,v) & \text{ otherwise.}
\end{cases}, \label{main:eq:potential:edge}
\end{eqnarray}
The potential $\B$ is defined as the sum of the potentials associated with all the nodes and edges. 
\begin{eqnarray}
\B  & =&  \sum_{v \in V} \Phi(v) + \sum_{i=1}^{(L-1)} \sum_{e \in S_i} \Psi_i(e)   \label{main:eq:potential:main} 
\end{eqnarray}
It might be instructive to contrast this potential function with the one used to analyze the dynamic algorithm in Section~\ref{sec:dynamic}.

\paragraph{Roadmap.} Our analysis works in three steps. 
\begin{enumerate}
\item In Definition~\ref{main:def:sample:main:2}, we  describe an event $\mathcal{F}$.  To understand the intuition behind this definition, recall the discussion on the third technical challenge in Section~\ref{main:sec:challenge}: We have to overcome the apparent obstacle that different levels of the $(\alpha, d_{k}^{(t)}, L)$-decomposition are constructed using different subsets of randomly sampled edges. Intuitively, the event $F$ guarantees that the degrees of a node among these different subsets of edges are approximately  the same with high probability. This helps in extending the ideas from the potential function based analysis in Section~\ref{sec:dynamic} to the current setting.  
\item In Lemma~\ref{main:cl:sample:main:2}, we show that the event $\mathcal{F}$ holds with high probability. The proof of Lemma~\ref{main:cl:sample:main:2} appears in Section~\ref{sec:extreme:new}.
\item Conditioned on the event $\mathcal{F}$, we show that our algorithm has $\O(1)$ amortized update time (see Lemma~\ref{main:lm:dynamic:stream:runtime} and Corollary~\ref{main:cor:dynamic:stream:runtime}). The proof of Lemma~\ref{main:lm:dynamic:stream:runtime} appears in Section~\ref{sec:extreme:new:1}.
\end{enumerate}

\begin{definition}
\label{main:def:sample:main:2}
Recall the procedure in Figure~\ref{main:fig:sample:recover}. For all levels $i, i' \in \{1, \ldots, L-1\}$ with $i < i'$, and time-steps $t \in [t_0, t_1]$, we define an  $\mathcal{F}_{i,i'}^{(t)}$ as follows.
\begin{itemize}
\item The event $\mathcal{F}_{i,i'}^{(t)}$ occurs iff the following conditions are satisfied.
\begin{itemize}
\item $\big\{$$\dd_v(Z_i^{(t)}, S_{i'}^{(t)}) \geq \frac{(1-\epsilon)^4}{(1+\epsilon)^2} \cdot (\alpha c \log n)$ for all $v \in A_i^{(t)}$$\big\}$, and 
\item $\big\{$$\dd_v(Z_i^{(t)}, S_{i'}^{(t)}) \leq \frac{(1+\epsilon)^4}{(1-\epsilon)^2} \cdot c \log n$ for all $v \in B_i^{(t)}$$\big\}$.
\end{itemize} 
\end{itemize}
Now, define the event $\mathcal{F}^{(t)} = \bigcap_{i, i'} \mathcal{F}_{i,i'}^{(t)}$ and the event $\mathcal{F} = \bigcap_{t=T'}^T \mathcal{F}^{(t)}$.\end{definition}

\begin{lemma}
\label{main:cl:sample:main:2}
The event $\mathcal{F}$ holds with high probability. 
\end{lemma}

\begin{lemma}
\label{main:lm:dynamic:stream:runtime}
Conditioned on the event $\mathcal{F}$, we have:
\begin{itemize}
\item (a) $0 \leq \B = \O(n)$ at each time-step $t \in [t_0, t_1]$.
\item (b) Insertion/deletion of an edge in $G$ (ignoring the call to the procedure in Figure~\ref{main:fig:sample:recover}) changes the potential $\B$ by $\O(1)$. 
\item (c) For every constant amount of computation performed while implementing the procedure in  Figure~\ref{main:fig:sample:recover}, the potential $\B$ drops by  $\Omega(1)$.  
\end{itemize}
\end{lemma}

\begin{corollary}
\label{main:cor:dynamic:stream:runtime}
 With high probability, subroutine {\sc Dynamic-stream} spends $\tilde O(n+ (t_1-t_0))$ time during the dense time-interval $[t_0, t_1]$. So its amortized update time is $\tilde O(1)$ with high probability.
\end{corollary}

\begin{proof}
Condition on the event $\mathcal{F}$ (which occurs with high probability by Lemma~\ref{main:cl:sample:main:2}). At  each time-step $t \in [1, T]$, we maintain the random sets of edges $\{S_i^{(t)}\}$ as per Theorem~\ref{main:th:dynamic:sample:dense}. This takes $\O(1)$ worst case update time.  Further, a single edge insertion/deletion in the input graph leads to at most two edge insertions/deletions in each of these sets $\{S_i^{(t)}\}, i \in \{1, \ldots, L-1\}$. 

Now, using the random sets $\{S_i^{(t)}\}$, at each time-step $t \in [t_0, t_1]$ we maintain an $(\alpha, d_k^{(t)}, L)$-decomposition of the input  graph (see Lemma~\ref{main:lm:performance}).  Lemma~\ref{main:lm:dynamic:stream:runtime} implies that with high probability, this requires a total update time of $\O(n + (t_1 - t_0))$ for the entire duration of the interval $[t_0, t_1]$. 

Finally, recall that either the dense time-interval spans $\tilde \Omega(n)$ time-steps, or it ends at time-step $T$ (see the discussion on the subroutine (P3) in Section~\ref{sec:overview}). Hence, applying an argument similar to the one used in Section~\ref{sec:sparse} (see the discussions preceding equation~\ref{eq:2:thought}), we conclude that with high probability the subroutine {\sc Dynamic-stream} spends $\O(T)$ total time during the first $T$ updates in the dynamic stream. In other words, with high probability the subroutine {\sc Dynamic-stream} has an amortized update time of $\O(1)$. 
\end{proof}

\subsubsection{Proof of Lemma~\ref{main:cl:sample:main:2}}
\label{sec:extreme:new}
Fix any  $1 \leq i < i' \leq L-1$, and any time-step $t \in [t_0,t_1]$. 
\begin{itemize}
\item Condition on an instantiation  of the random set $Z_i^{(t)}$. Note that this also determines the sets $A_i^{(t)}$ and $B_i^{(t)}$ (see Figure~\ref{main:fig:sample:recover}). 
\begin{itemize}
\item 
 Let $W_i^{(t)} \subseteq Z_i^{(t)}$ be the subset of nodes $v$ with small degrees $\dd_v(Z_i^{(t)}, E^{(t)})$. Specifically,
 \begin{eqnarray}
 \label{eq:new:old:new:3}
 W_i^{(t)} = \left\{ v \in Z_i^{(t)} : \dd_v(Z_i^{(t)}, E^{(t)})  < \frac{(1-\epsilon)^2}{(1+\epsilon)^2} \cdot \alpha d_k^{(t)}\right\}
 \end{eqnarray}
The node set $W_i^{(t)}$ is uniquely determined by the node set $Z_i^{(t)}$ (which we are conditioning upon) and the edge set $E^{(t)}$ (which is given by the stream of updates in the input graph). \\

Now, consider any node $v \in W_i^{(t)}$.  By Lemma~\ref{main:lm:independence}, the contents of the random set $S_{i}^{(t)}$ are  independent of $Z_i^{(t)}$. By Theorem~\ref{main:th:dynamic:sample:dense}, each edge $(u,v) \in E^{(t)}$ is included in $S_i^{(t)}$ with probability $(1\pm \epsilon) \cdot (s/m^{(t)})$.  Applying Linearity of expectation and equation~\ref{eq:overview:dense:st}, we get:
 \begin{eqnarray}
 \label{eq:new:old:new:4}
 E\left[\dd_v(Z_i^{(t)}, S_{i}^{(t)})\right]  & \leq &   (1+ \epsilon) \cdot (c \log n/d_{k}^{(t)}) \cdot \dd_v(Z_i^{(t)}, E^{(t)}) \\
 & \leq & \frac{(1-\epsilon)^2}{(1+\epsilon)} \cdot (\alpha c \log n) \label{eq:new:old:new:5}
 \end{eqnarray}
 Equation~\ref{eq:new:old:new:5} follows from equations~\ref{eq:new:old:new:3} and~\ref{eq:new:old:new:4}. Next, for each edge $(u, v) \in E^{(t)}$ incident upon the node $v$, consider the random event that $(u, v) \in S_i^{(t)}$. By Theorem~\ref{main:th:dynamic:sample:dense}, these random events are negatively associated (see Section~\ref{sec:prelim:concentration}).  Thus, applying Chernoff bound on equation~\ref{eq:new:old:new:5} and recalling the definition of $A_i^{(t)}$ from Figure~\ref{main:fig:sample:recover}, we get:
 \begin{eqnarray}
 \label{eq:new:old:new:6}
 \Pr\left[v \in A_i^{(t)}\right] = \Pr\left[\dd_v(Z_i^{(t)}, S_{i}^{(t)}) > (1-\epsilon)^2 \cdot (\alpha c \log n)\right] \leq 1/(\text{poly } n)
  \end{eqnarray}
  Applying a union bound on equation~\ref{eq:new:old:new:6} over all the nodes in $W_i^{(t)}$, we get:
  \begin{eqnarray}
 \label{eq:new:old:new:7}
 \Pr\left[W_i^{(t)} \bigcap A_i^{(t)} \neq \emptyset \right] \leq 1/(\text{poly } n)
 \end{eqnarray}  
 In other words, with high probability no node in $W_i^{(t)}$ belongs to the set $A_i^{(t)}$.  \\

We will now bound the degrees of the nodes in $Z_i^{(t)} \setminus W_i^{(t)}$ with respect to the random edge set $S_{i'}^{(t)}$. Towards this end, consider any node $x \in Z_i^{(t)} \setminus W_i^{(t)}$.  By Lemma~\ref{main:lm:independence}, the contents of the random set $S_{i'}^{(t)}$ are independent of $Z_i^{(t)}$. By Theorem~\ref{main:th:dynamic:sample:dense}, each edge $(u,v) \in E^{(t)}$ is included in $S_{i'}^{(t)}$ with probability $(1\pm \epsilon) \cdot (s/m^{(t)})$.  Applying Linearity of expectation and equations~\ref{eq:overview:dense:st},~\ref{eq:new:old:new:3} we get:
 \begin{eqnarray}
 \label{eq:new:old:new:8}
 E\left[\dd_x(Z_i^{(t)}, S_{i'}^{(t)})\right]  & \geq &   (1- \epsilon) \cdot (c \log n/d_{k}^{(t)}) \cdot \dd_v(Z_i^{(t)}, E^{(t)}) \\
 & \geq & \frac{(1-\epsilon)^3}{(1+\epsilon)^2} \cdot (\alpha c \log n) \label{eq:new:old:new:9}
 \end{eqnarray}
Equation~\ref{eq:new:old:new:9} follows from equations~\ref{eq:new:old:new:3} and~\ref{eq:new:old:new:8}. Next, for each edge $(u, x) \in E^{(t)}$ incident upon the node $x$, consider the random event that $(u, x) \in S_{i'}^{(t)}$. By Theorem~\ref{main:th:dynamic:sample:dense}, these random events are negatively associated (see Section~\ref{sec:prelim:concentration}).  Thus, applying Chernoff bound on equation~\ref{eq:new:old:new:9}, we get:
\begin{eqnarray}
 \label{eq:new:old:new:10}
 \Pr\left[\dd_x(Z_i^{(t)}, S_{i'}^{(t)}) <  \frac{(1-\epsilon)^4}{(1+\epsilon)^2} \cdot (\alpha c \log n)\right]  \leq 1/(\text{poly } n)
 \end{eqnarray}
 Now, taking an union bound on equation~\ref{eq:new:old:new:10} over all the nodes in $Z_i^{(t)} \setminus W_i^{(t)}$, we get:
 \begin{eqnarray}
 \label{eq:new:old:new:11}
 \Pr\left[\dd_x(Z_i^{(t)}, S_{i'}^{(t)}) <  \frac{(1-\epsilon)^4}{(1+\epsilon)^2} \cdot (\alpha c \log n) \text{ for some } x \in Z_{i}^{(t)} \setminus W_i^{(t)}\right]  \leq 1/(\text{poly } n)
 \end{eqnarray}
 In other words, with high probability every node  $x \in Z_i^{(t)} \setminus W_i^{(t)}$ has a high degree $\dd_x(Z_i^{(t)}, S_{i'}^{(t)})$. Now, taking an union bound over equations~\ref{eq:new:old:new:7} and~\ref{eq:new:old:new:11}, we conclude that:
 \begin{eqnarray}
 \label{eq:new:old:new:12}
 \text{With high probability,  } \dd_x(Z_i^{(t)}, S_{i'}^{(t)}) \geq  \frac{(1-\epsilon)^4}{(1+\epsilon)^2} \cdot (\alpha c \log n) \text{ for every node } x \in A_i^{(t)}.
 \end{eqnarray}
 Using a similar argument for the node set $B_i^{(t)}$, we can infer that:
 \begin{eqnarray}
 \label{eq:new:old:new:13}
 \text{With high probability,  } \dd_x(Z_i^{(t)}, S_{i'}^{(t)}) \leq  \frac{(1+\epsilon)^4}{(1-\epsilon)^2} \cdot (c \log n) \text{ for every node } x \in B_i^{(t)}.
 \end{eqnarray}
 Taking an union bound over equations~\ref{eq:new:old:new:12} and~\ref{eq:new:old:new:13}, we infer that:
 \begin{equation}
 \label{eq:new:old:new:14}
 \text{Given any  instantiation of } Z_i^{(t)}, \text{ the event } \mathcal{F}_{i, i'}^{(t)} \text{ occurs with high probability.}
 \end{equation} 
\end{itemize}
\item From equation~\ref{eq:new:old:new:14}, we infer that:
\begin{equation}
 \label{eq:new:old:new:15}
\text{The event } \mathcal{F}_{i, i'}^{(t)} \text{ occurs with high probability.}
 \end{equation} 
\end{itemize}
The lemma follows by  applying an union bound on equation~\ref{eq:new:old:new:15} over all indices $i, i' \in \{1, \ldots, L-1\}$ with $i < i'$ and time-steps  $t \in [t_0, t_1]$.

\subsubsection{Proof of Lemma~\ref{main:lm:dynamic:stream:runtime}}
\label{sec:extreme:new:1}
Part (a) and part (b) of the lemma hold independently of the event $\mathcal{F}$. It is only during the proof of part (c) that we have to condition of the event $\mathcal{F}$. 

\medskip

\noindent {\bf Proof of part (a).} Fix any time-step $t \in [t_0, t_1]$. The proof follows from three facts.
\begin{enumerate} 
\item We have $0 \leq \Phi(v) \leq (L/\epsilon) \cdot L \cdot (1-\epsilon)^2 \alpha c \log n = \O(1)$ for all $v \in V$. 
\item  We have $0 \leq \Psi_i(u,v) \leq 3L = \O(1)$ for all $i \in [L-1], (u,v) \in S_i^{(t)}$.  
\item We have  $|S_i^{(t)}| \leq s = \O(n)$ for all $i \in [L-1]$. This follows from equation~\ref{eq:overview:dense:s}. 
\end{enumerate}

\medskip
\noindent {\bf Proof of part (b).} By Theorem~\ref{main:th:dynamic:sample:dense}, insertion/deletion of an edge in $G$ leads to at most two  insertions/deletions in each of the random sets $S_1^{(t)}, \ldots, S_{L-1}^{(t)} \subseteq E^{(t)}$. Since $L = \O(1)$ (see equation~\ref{eq:new:L}), it suffices to show that for a single edge insertion/deletion in any given $S_i^{(t)}$, the potential $\B$ changes by at most $\O(1)$ (ignoring the call to the procedure in  Figure~\ref{main:fig:sample:recover}).

Towards this end, fix any $i \in [L-1]$, and suppose that a single edge $(u,v)$ is inserted  into (resp. deleted from)  $S_i^{(t)}$. For each endpoint $x \in \{u, v\}$, this changes the potential $\Phi(x)$ by at most $O(L/\epsilon)$. The potentials $\Phi(y)$ for all other nodes $y \in V \setminus \{u,v\}$ remain unchanged. Additionally, the potential $\Psi_i(u,v) \in [0,3L]$ is created (resp. destroyed). Thus, we infer that the absolute value of the change in the overall potential $\B$ is at most $O(3L+ 2 L/\epsilon) = \O(1)$.

\medskip
\noindent {\bf Proof of part (c).} Fix any time-step $t \in [t_0, t_1]$, and any iteration of the  {\sc For} loop in Figure~\ref{main:fig:sample:recover} while processing the update in time-step $t$. Consider two possible events.

\medskip
\noindent {\bf Case 1: A node $v \in Z_i^{(t)}$ is promoted from level $i$ to level $(i+1)$ in Step~07 of Figure~\ref{main:fig:sample:recover}.} 

\smallskip
\noindent This  happens only if $v \in A^{(t)}_i$. Let $C$ be the amount of computation performed during this step. By Claim~\ref{main:cl:sample:runtime:up}, we have:
 \begin{eqnarray} 
\label{main:eq:up:1}
 C &=& \sum_{i' = (i+1)}^{(L-1)}O\left(1+\dd_v(Z_i^{(t)},S_{i'}^{(t)})\right)
\end{eqnarray} 

\noindent Let $\Delta$ be the net decrease in the overall potential $\B$ due to this step. We observe that:
\begin{enumerate}
\item Consider any $i' > i$. For each edge  $(u,v) \in S_{i'}^{(t)}$ with $u \in Z_i^{(t)}$, the potential $\Psi_{i'}(u,v)$ decreases by at least one.  For every other edge $e \in S_{i'}^{(t)}$, the potential $\Psi_{i'}(e)$ remains unchanged.
\item For each $i' \in [i]$ and each edge $e \in S_{i'}^{(t)}$, the potential $\Psi_{i'}(e)$ remains unchanged. 
\item  Since the node $v$ is being promoted to level $(i+1)$, we have $\dd_v(Z_i^{(t)}, S_i^{(t)}) \geq (1-\epsilon)^2 \alpha c \log n$. Thus, the potential $\Phi(v)$ remains unchanged.  For each node $u \neq v$, the potential $\Phi(u)$ can only decrease (this holds since the degree $D_u(Z_j, S_j)$, for any level $j$, can only increase as node $v$ increases its level from $i$ to $i+1$). 
\end{enumerate}
\noindent Taking into account all these observations, we infer the following inequality.
 \begin{eqnarray}
\label{main:eq:up:2}
\Delta &\geq& \sum_{i' = (i+1)}^{(L-1)} \dd_v(Z_i^{(t)}, S_{i'}^{(t)})
 \end{eqnarray} 
 Since  $v \in A_i^{(t)}$, and since we have conditioned on the event $\mathcal{F}^{(t)}$ (see Definition~\ref{main:def:sample:main:2}), we get:
 \begin{eqnarray}
\label{main:eq:up:3}
\dd_v(Z_i^{(t)}, S_{i'}^{(t)}) &>& 0  \ \ \text{ for all } i' \in [i+1, L-1].
 \end{eqnarray} 
 
\noindent Eq. (\ref{main:eq:up:1}),~(\ref{main:eq:up:2}),~(\ref{main:eq:up:3}) imply that the  decrease in  $\B$ is sufficient to pay for  the computation performed.

\medskip
\noindent {\bf Case 2: A node $v \in Z_i^{(t)}$ is demoted from level $j$ to  level $i$ in Steps~(08-09) of Figure~\ref{main:fig:sample:recover}.}

\smallskip
\noindent This can happen only if $j > i$ and $v \in B^{(t)}_i$. Let $C$ denote the amount of computation performed during this step. By Claim \ref{main:cl:sample:runtime:up}, we have
 \begin{eqnarray}
\label{main:eq:down:1}
 C &=& \sum_{i' = (i+1)}^{(L-1)}O(1+   \dd_v(Z_i^{(t)},S_{i'}^{(t)}))
\end{eqnarray}  
Let $\gamma = (1+\epsilon)^4/(1-\epsilon)^2$.  Equation~(\ref{main:eq:down:2}) holds since $v \in B_i^{(t)}$ and since we conditioned on the event $\mathcal{F}$ (see Definition~\ref{main:def:sample:main:2}). Equation~(\ref{main:eq:down:3}) follows from equations (\ref{main:eq:down:1}), (\ref{main:eq:down:2}) and since  $\gamma, c$ are constants. 
 \begin{eqnarray}
 \label{main:eq:down:2}
 \dd_v(Z_i^{(t)}, S_{i'}^{(t)}) & \leq & \gamma c \log n  \text{ for all } i' \in [i, L-1] \\
\label{main:eq:down:3}
C & = & O(L \log n)
\end{eqnarray}  
Let $\Delta$ be the net decrease in the overall potential $\B$ due to this step. We observe that:
\begin{enumerate} 
\item By eq.~(\ref{main:eq:down:2}), the potential $\Phi(v)$ decreases by at least $(j-i) \cdot (L/\epsilon) \cdot ((1-\epsilon)^2\alpha -\gamma) \cdot (c \log n)$.
\item For $u \in V \setminus \{v\}$ and $i' \in  [1, i] \cup [j+1, L-1]$,  the potential $\Gamma_{i'}(u)$ remains unchanged.  This  observation, along with   equation~(\ref{main:eq:down:2}), implies that the sum $\sum_{u \neq v} \Phi(u)$ increases by at most 
$(L/\epsilon) \cdot \sum_{i' = (i+1)}^j \dd_v(Z_i^{(t)}, S_{i'}^{(t)}) \leq (j-i) \cdot (L/\epsilon) \cdot (\gamma c \log n)$.
\item For every $i' \in [1,i]$, and $e \in S_{i'}^{(t)}$ the potential $\Psi_{i'}(e)$ remains unchanged. Next, consider any $i' \in [i+1, L-1]$. For each edge $(u,v) \in S_{i'}^{(t)}$ with $u \in Z_i^{(t)}$, the potential $\Psi_{i'}(u,v)$ increases by at most $3(j-i)$. For every other edge $e \in S_{i'}^{(t)}$, the potential $\Psi_{i'}(e)$ remains unchanged. These observations, along with equation~(\ref{main:eq:down:2}), imply that the sum $\sum_{i'} \sum_{e \in S_{i'}} \Psi_{i'}(e)$ increases by at most $\sum_{i' = (i+1)}^{(L-1)} 3(j-i) \cdot \dd_v(Z_i^{(t)}, S_{i'}^{(t)}) \leq (j-i) \cdot (3L) \cdot (\gamma c \log n)$. 
\end{enumerate}
\noindent Taking into account all these observations, we get:

 \begin{eqnarray}
\Delta &\geq&  (j-i)  (L/\epsilon)  ((1-\epsilon)^2\alpha -\gamma)  (c \log n)  \nonumber  \\
           &&-  (j-i)  (L/\epsilon)  (\gamma c \log n) - (j-i)  (3L)  (\gamma c \log n) \nonumber  \\ 
          &=&  (j-i) \cdot (L/\epsilon) \cdot ((1-\epsilon)^2\alpha -2\gamma 
          - 3 \epsilon \gamma) \cdot (c \log n)  \nonumber  \\  
          & \geq&  L c \log n    \nonumber  \label{main:eq:down:4}  \\
\end{eqnarray}

The last inequality holds since $(j-i) \geq 1$ and $\alpha \geq  (\epsilon + (2+3\epsilon) \gamma)/(1-\epsilon)^2 = 2 +\Theta(\epsilon)$, for some sufficiently small constant $\epsilon \in (0,1)$. From eq.~(\ref{main:eq:down:3}) and (\ref{main:eq:down:4}), we conclude that  the net decrease in the  overall potential $\B$ is sufficient to pay for the cost of the computation performed.

\section{Extension to Directed Graphs}
\label{sec:directed}

In this section, we extend our  dynamic algorithm from Section~\ref{sec:dynamic} to directed graphs. The notion of  ``density'' of a directed graph was introduced  by Kannan et al.~\cite{KannanV}.  We first recall their definition. 

\begin{definition}
\label{def:directed:main}\cite{KannanV}
Consider two subsets of nodes $X, Y \subseteq V$ in a directed graph $G = (V, E)$. The ``{\em density}'' of this pair equals $\rho(X, Y) = |E(X, Y)|/\sqrt{|X| |Y|}$. Here, $E(X, Y) = \{ (u, v) : u \in X, v \in Y \}$ is the set of edges going from $X$ to $Y$. The value of the densest subgraph is given by $\rho(G) = \max_{X, Y \subseteq V} \rho(X, Y)$. Note that we do not require the sets $X$ and $Y$ to be mutually disjoint. 
\end{definition}

Throughout this section, we denote the input graph by $G = (V, E)$. In the beginning, the input graph is empty, i.e., we have $E = \emptyset$. Subsequently, at each time-step, either a directed edge is inserted into the graph, or an already existing edge is deleted from the graph. The set of nodes, on the other hand, remains unchanged. Our goal is to maintain a good approximation of $\rho(G)$ in this dynamic setting. The main result is stated below.

\begin{theorem}
\label{th:directed:main}
We can deterministically maintain a $(8+\epsilon)$-approximation to the value of the densest subgraph of a directed graph $G = (V, E)$. The algorithm requires $\O(m+n)$-space,  where $m$ (resp. $n$) denotes the number of nodes (resp. edges) in the graph. Furthermore, the algorithm   has an amortized update time of $\O(1)$ and a query time of $O(1)$.
\end{theorem}

We devote the rest of Section~\ref{sec:directed} to the proof of Theorem~\ref{th:directed:main}. We first define the preliminary concepts and notations in Section~\ref{sec:derived:notation}. Next, we extend the notion of an $(\alpha, d, L)$-decomposition to directed graphs in Section~\ref{sec:derived:decomposition}. In Section~\ref{sec:derived:maintain}, we present our main algorithm. Finally, in Section~\ref{sec:directed:conclude}, we combine all these ingredients together and conclude the proof of Theorem~\ref{th:directed:main}.

\subsection{Notations and Preliminaries}
\label{sec:derived:notation}

We first define the notion of a ``{\em derived graph}'',  which will be crucially used in our algorithm.

\begin{definition}
\label{def:derived:graph}
The ``{\em derived graph}'' $G' = (V', E')$ of the input graph $G = (V, E)$ is constructed as follows. For each node $v \in V$, we create two nodes $s_v$ and $t_v$. We define the  node set $V' = S' \cup T'$, where $S' = \{ s_v : v \in V \}$ and $T' = \{ t_v : v \in V \}$. Next, for each directed edge $(u, v) \in E$, we create the directed edge $(s_u, t_v)$, and define the directed edge set $E' = \{ (s_u,t_v)  : (u, v) \in E \}$.  Thus, in the derived graph $G' = (V', E')$, each node in $S'$ (resp. $T'$) has zero in-degree (resp. out-degree). 
\end{definition}

It is easy to check that the derived graph $G' = (V', E')$ can be maintained in the dynamic setting using $O(m+n)$ space and $O(1)$ update time: Fix the set of nodes $V' = S' \cup T'$, and whenever an edge $(u,v)$ is inserted into (resp. deleted from) $G = (V, E)$, insert (resp. delete) the corresponding derived edge $(s_u, t_v)$ in $G'$. From now on, unless explicitly mentioned otherwise, our main algorithm will work on the derived graph $G' = (V', E')$. Before proceeding any further, we introduce some notations that will be used throughout the rest of this section.

\begin{itemize}
\item Consider the derived graph $G' = (V', E')$. Given any node $s \in S'$ and any subset of nodes $T \subseteq T'$, we let $\N_{s}(T) = \{ t \in T : (s, t) \in E'\}$ denote the set of neighbors of $s$ among the nodes in $T$. Furthermore, we let $D_s(T) = | \N_s(T) |$ denote the degree of  $s$ among the nodes in $T$. For a node $t \in T'$ and a subset $S \subseteq S'$, the notations $\N_t(S) = \{ (s,t) \in E' : s \in S\}$ and $D_t(S) = |\N_t(S)|$ are defined analogously. Next, for any two subsets of nodes $S \subseteq S'$ and $T \subseteq T'$, we let $E'(S, T) = \{ (s, t) \in E' : s \in S, t \in T \}$ denote the set of edges in the derived graph that are going from $S$ to $T$. We also define $\rho'(S, T)  = |E'(S,T)|/\sqrt{ |S| |T| }$ for all $S \subseteq S', T \subseteq T'$.  Hence, from Definition~\ref{def:directed:main}, it follows that $\rho'(S, T) = \rho(S, T)$. Accordingly, the value of the densest subgraph of the input graph $G = (V, E)$ is given by $\rho(G) = \max_{S \subseteq S', T \subseteq T'} \rho'(S, T)$. We will denote the densest subgraph of $G = (V, E)$ by the pair $(S^*, T^*)$, where $S^* \subseteq S'$ and $T^* \subseteq T'$. Thus, we have $\rho'(S^*, T^*) = \rho(G)$. Finally, we define the parameters $\lambda_{S'}, \lambda_{T'}$ as follows.
\begin{equation}
\label{eq:derived:lambda}
\lambda_{S'} =  |E(S^*, T^*)| \cdot (1- \sqrt{1 - 1/|S^*|}) \text{ and } \lambda_{T'} =  |E(S^*, T^*)| \cdot (1- \sqrt{1 - 1/|T^*|}).
\end{equation} 
\end{itemize} 

\noindent We now state one crucial lemma that will be used in the analysis of our algorithm. This lemma was proved by Khuller et al.~\cite{KhullerS09}. For the sake of completeness, we state their proof below.  

\begin{lemma}
\label{lm:derived:main}\cite{KhullerS09}
Consider the densest subgraph $(S^*, T^*)$ of the input graph $G = (V, E)$. We have:
\begin{enumerate}
\item $D_s(T^*) \geq \lambda_{S'}$ for all nodes $s \in S^*$, and $D_t(S^*) \geq \lambda_{T'}$ for all nodes $t \in T^*$. 
\item $\rho(G)/2 \leq   \sqrt{\lambda_{S'} \lambda_{T'}} \leq \rho(G)$. 
\end{enumerate}
\end{lemma}

\begin{proof} \ 
\begin{enumerate}
\item For the sake of contradiction, suppose that there exists a node $s \in S^*$ with $D_s(T^*) < \lambda_{S'}$. Then we can show that the density of the pair $(S^* \setminus \{s\}, T^*)$ is strictly larger than the density of the pair $(S^*, T^*)$. Specifically, we can infer the following guarantee:
$$\rho'(S^* \setminus \{s\}, T^*) > \frac{|E'(S^*, T^*)| - \lambda_{S'}}{\sqrt{(|S^*| - 1) \cdot |T^*|}} = \frac{|E'(S^*, T^*)|}{\sqrt{|S^*| \cdot |T^*|}} = \rho'(S^*, T^*) = \rho(G).$$
Since $\rho(G)$ denotes the value of the densest subgraph of $G$, the above inequality leads to a contradiction. Thus, we conclude that $D_s(T^*) \geq \lambda_{S'}$ for all nodes $s \in S^*$. Using the same line of reasoning, we can conclude that $D_t(S^*) \geq \lambda_{T'}$ for all nodes $t \in T^*$. 
\item From part I of the lemma, we have the following guarantee: Every node in $S^*$ (resp. $T^*$) has a degree of at least $\lambda_{S'}$ (resp. $\lambda_{T'}$) among the nodes in $T^*$ (resp. $S^*$). This implies that $|E'(S^*, T^*)| \geq |S^*| \cdot \lambda_{S'}$ and $|E'(S^*, T^*) \geq |T^*| \cdot \lambda_{T'}$.  Thus, we get:
\begin{equation}
\label{eq:derived:verynew:1}
|E'(S^*, T^*)| \geq \sqrt{|S^*| \cdot |T^*|} \cdot \sqrt{\lambda_{S'} \cdot \lambda_{T'}}
\end{equation}
Since $\rho(G) = \rho'(S^*, T^*) = |E'(S^*, T^*)|/\sqrt{|S^*| \cdot |T^*|}$, from equation~\ref{eq:derived:verynew:1} we conclude that:
\begin{equation}
\label{eq:derived:verynew:2}
\rho(G) =  \frac{|E'(S^*, T^*)|}{\sqrt{|S^*| \cdot |T^*|}} \geq \sqrt{\lambda_{S'} \cdot \lambda_{T'}}
\end{equation}
Next, putting $|S^*| = 1/\sin^2 \theta_1$ and $|T^*| = 1/\sin^2 \theta_2$, and recalling equation~\ref{eq:derived:lambda}, we get:
\begin{eqnarray}
\lambda_{S'} \cdot \lambda_{T'} &  = & \frac{|E'(S^*, T^*)|^2}{|S^*| \cdot |T^*|} \cdot \left(|S^*| \cdot |T^*|\right) \cdot \left(1 - \sqrt{1- 1/|S^*|}\right) \cdot \left( 1 - \sqrt{1 - 1/|T^*|}\right) \nonumber \\
& = & \left(\rho'(S^*, T^*)\right)^2 \cdot \frac{(1-\cos \theta_1) \cdot (1-\cos \theta_2)}{\sin^2 \theta_1 \sin^2 \theta_2} \nonumber \\
& = & \frac{\left(\rho'(S^*, T^*)\right)^2}{4 \cos^2 (\theta_1/2) \cos^2 (\theta_2/2)} \nonumber \\
& \geq & \frac{\left(\rho'(S^*, T^*)\right)^2}{4} = \frac{(\rho(G))^2}{4} \label{eq:derived:verynew:10}
\end{eqnarray}
From equation~\ref{eq:derived:verynew:10} we conclude that:
\begin{equation}
\label{eq:derived:verynew:11}
\sqrt{\lambda_{S'} \cdot \lambda_{T'}} \geq \rho(G)/2
\end{equation}
The part II of the lemma follows from equations~\ref{eq:derived:verynew:2} and~\ref{eq:derived:verynew:11}.
\end{enumerate}
\end{proof}

\subsection{$(\alpha, d_{S'}, d_{T'}, L)$-Decomposition and Its Properties}
\label{sec:derived:decomposition}

 We  extend the concept of an $(\alpha, d, L)$-decomposition (see Definition~\ref{def:partition}) to directed graphs.  This requires us to introduce one additional parameter in the definition, and we call the corresponding entity an $(\alpha, d_{S'}, d_{T'}, L)$-decomposition. Specifically, an $(\alpha, d_{S'}, d_{T'}, L)$-decomposition of the derived graph $G' = (V', E')$, where $V' = S' \cup T'$, is given by two laminar families of subsets of nodes $S' = S_1 \supseteq S_2 \supseteq \cdots \supseteq S_L$ and $T' = T_1 \supseteq T_2 \supseteq \cdots \supseteq T_L$. These subsets are iteratively constructed as follows. First, we set $S_1 = S$ and $T_1 = T$. Next, for each $i \in \{1, \ldots, L-1\}$, while constructing the subsets $S_{i+1}, T_{i+1}$ from the subsets $S_i, T_i$, we  ensure  the following conditions:
 \begin{itemize}
 \item Every node $s \in S_i$ with $D_s(T_i) > \alpha \cdot d_{S'}$ is included in the subset $S_{i+1}$. On the other hand, every node $s \in S_i$ with $D_s(T_i) < d_{S'}$ is excluded from the subset $S_{i+1}$.
 \item Every node $t \in T_i$ with $D_t(S_i) > \alpha \cdot d_{T'}$ is included in the subset $T_{i+1}$. On the other hand, every node $t \in T_i$ with $D_t(S_i) < d_{T'}$ is excluded from the subset $T_{i+1}$.
 \end{itemize}  
 \noindent We are now ready to formally define the concept of an $(\alpha, d_{S'}, d_{T'}, L)$-decomposition of $G'$.

\begin{definition}
\label{main:def:partition:derived}
Fix any $\alpha \geq 1$,  $d_{S'},d_{T'} \geq 0$, and any positive integer $L$. Consider a family of subsets   $S' = S_1  \supseteq \cdots \supseteq S_L$
and $T' = T_1  \supseteq \cdots \supseteq T_L$.
The tuples $(S_1, \ldots, S_L)$ and $(T_1, \ldots, T_L)$ form  an 
$(\alpha,  d_{S'}, d_{T'}, L)$-decomposition of the derived  graph $G'$ iff  for every $i \in \{1, \ldots, L-1\}$, we have:
\begin{enumerate} 
\item $S_{i+1} \supseteq \left\{s \in S_i : \dd_s(T_i) > \alpha \cdot d_{S'} \right\}$ and  $S_{i+1} \cap \left\{ s \in S_i : \dd_s(T_i) < d_{S'} \right\} = \emptyset$.
\item $T_{i+1} \supseteq \left\{t \in T_i : \dd_t(S_i) > \alpha \cdot d_{T'} \right\}$ and   $T_{i+1} \cap \left\{ t \in T_i : \dd_t(S_i) < d_{T'} \right\} = \emptyset$.
\end{enumerate}
Let $V'_i = (S_i \cup T_i) \setminus (S_{i+1}
\cup T_{i+1})$ for all $i \in \{1, \ldots, L-1\}$, and $V'_i = S_i \cup T_i$ for $i = L$. We say that the nodes in $V'_i$ constitute the $i^{th}$ level of this decomposition. We also denote the level of a node $v \in V'$ by $\ell(v)$. Thus, we have $\ell(v) = i$ whenever $v \in V'_i$, and the set of nodes $V'$ is partitioned into $L$ subsets $V'_1, \ldots, V'_L$. 
\end{definition}

Theorem~\ref{main:thm:derived} and Corollary~\ref{cor:derived} will play the main role in the rest of this section. Roughly speaking, they are generalizations of Theorem~\ref{thm:test} and Corollary~\ref{cor:test:1} from Section~\ref{sec:decomposition}, and they state that we can use the $(\alpha, d_{S'}, d_{T'}, L)$-decomposition to $4\alpha (1+\epsilon)^{3/2}$-approximate of the value of the densest subgraph of $G = (V, E)$. All we need to do is to  set $L = O(\log n/\epsilon)$ and try different values of $d_{S'}, d_{T'}$ in powers of $(1+\epsilon)$.

\begin{theorem}
\label{main:thm:derived}
Fix any  $\alpha \geq 1$, $d_{S'}, d_{T'} \geq 0$, $\epsilon \in (0,1)$, and $L  \leftarrow 2 \cdot (2 + \lceil \log_{(1+\epsilon)} n \rceil)$. 
Let $(S_1, \ldots, S_L), (T_1, \ldots, T_L)$ be an $(\alpha,  d_{S'}, d_{T'}, L)$-decomposition of the derived graph $G' = (V', E')$ as per Definition~\ref{main:def:partition:derived}. Then the following conditions are satisfied. 
\begin{enumerate}
\item If $d_{S'} d_{T'} >  4  (1+\eps)\lambda_{S'} \lambda_{T'}$, then 
$V'_L = \emptyset$ (i.e., the topmost level of the decomposition is empty). 

\item If $d_{S'} < \lambda_{S'}/\alpha$ and
$d_{T'} < \lambda_{T'}/\alpha$, then  $V'_L \ne \emptyset$ (i.e.,  the topmost level  is nonempty). 
\end{enumerate}
\end{theorem}

\begin{proof} \
\begin{enumerate}
\item For any level $k \in \{2, \ldots , L\}$. Definition~\ref{main:def:partition:derived} states that for every node $s \in S_k$, we have $D_s(T_{k-1}) > d_{S'}$. We thus get a lower bound on the number edges going from $S_k$ to $T_{k-1}$.
\begin{equation}
\label{eq:derived:new:1}
 |E'(S_k, T_{k-1})| \ge |S_k| \cdot d_{S'}
 \end{equation}
 Recall that by Lemma~\ref{lm:derived:main}, we have $\rho(G) \leq 2 \cdot \sqrt{\lambda_{S'} \lambda_{T'}}$. This implies that $\rho'(S_k, T_{k-1}) = |E'(S_k, T_{k-1})|/\sqrt{|S_k| |T_{k-1}|} \leq \rho(G) \leq 2 \cdot \sqrt{\lambda_{S'} \lambda_{T'}} \leq \sqrt{ d_{S'} d_{T'}/ (1+\epsilon)}$. Hence, we get:
 \begin{equation}
 \label{eq:derived:new:2}
 |E'(S_k, T_{k-1})| \leq \sqrt{\frac{|S_k| \cdot |T_{k-1}| \cdot |d_{S'}| \cdot |d_{T'}|}{(1+\epsilon)}}
 \end{equation}
 From equations~\ref{eq:derived:new:1} and~\ref{eq:derived:new:2}, we get:
 \begin{equation}
 \label{eq:derived:new:3}
 |S_k| \cdot d_{S'} \leq |T_{k-1}| \cdot d_{T'}/(1+\epsilon)
 \end{equation}
 Replacing $S$ by $T$ in the above argument, we can analogously show that:
 \begin{equation}
 \label{eq:derived:new:4}
 |T_k| \cdot d_{T'} \leq |S_{k-1}| \cdot d_{S'}/(1+\epsilon)
 \end{equation}
 Now, we consider two possible cases.
 \begin{itemize}
 \item If $|S_{k-1}| \cdot d_{S'} \geq |T_{k-1}| \cdot d_{T'}$, then equation~\ref{eq:derived:new:3} implies that $|S_k| \leq |S_{k-1}|/(1+\epsilon)$.
 \item Else if $|S_{k-1}| \cdot d_{S'} < |T_{k-1}| \cdot d_{T'}$, then equation~\ref{eq:derived:new:4} implies that $|T_k| < |T_{k-1}|/(1+\epsilon)$.
 \end{itemize}
In other words, when going from level $k-1$ to level $k$, either the size of $S$ or the size of $T$ reduces by a factor of $(1+\epsilon)$. Since $|S_1| = |T_1| = n$ and $L = 2 \cdot (2 + \lceil \log_{(1+\epsilon)} n \rceil)$, after $L$ levels both the sets $S$ and $T$ are empty. We thus have $S_L = T_L = \emptyset$.
 
\item Consider the densest subgraph of the input graph $G = (V, E)$, given by the pair  $(S^*, T^*)$ where $S^* \subseteq S'$ and $T^* \subseteq T'$. Lemma~\ref{lm:derived:main} states that each node $s \in S^*$ (resp. $t \in T^*$) has out-degree (resp. in-degree) at least $\lambda_{S'}$ (resp. $\lambda_{T'}$) with respect to the nodes in $T^*$ (resp. $S^*$).  Since $\lambda_{S'} > \alpha \cdot d_{S'}$ and $\lambda_{T'} > \alpha \cdot d_{T'}$, Definition~\ref{main:def:partition:derived} implies the following guarantee:
\begin{itemize}
\item If $S^* \subseteq S_i$, $T^* \subseteq T_i$ for some  $i \in \{1, \ldots, L-1\}$, then we also have $S^* \subseteq S_{i+1}$, $T^* \subseteq T_{i+1}$.
\end{itemize}
Since $S_1 = S'$, $T_1 = T'$, we  have $S^* \subseteq S_1$ and $T^* \subseteq T_1$. Hence, applying the above guarantee in an inductive fashion, we conclude that $S^* \subseteq S_L$ and $T^* \subseteq T_L$. Thus,  $V'_L = S_L \cup T_L \neq \emptyset$.
\end{enumerate}
\end{proof}

\begin{corollary}
\label{cor:derived}
As in Theorem~\ref{main:thm:derived}, fix any $\alpha \geq 1$, $\epsilon \in (0,1)$ and $L = 2 \cdot (2 + \lceil \log_{(1+\epsilon)} n \rceil)$. Define $\lambda^* = 1 - \sqrt{1-1/n}$. Discretize the range $[\lambda^*/(\alpha (1+\epsilon)), n^2]$ in powers of $(1+\epsilon)$, by  letting $q_k = (1+\epsilon)^{k-1} \lambda^*/\alpha$ for every  integer $k \geq 0$. Let $K$ be the smallest integer $k$ for which $q_k \geq n^2$. For all $d_{S'}, d_{T'} \in \{d_0, \ldots, d_K\}$, construct an $(\alpha, d_{S'}, d_{T'}, L)$-decomposition of  $G'$ as per Definition~\ref{main:def:partition:derived}, and let $V'_L(d_{S'}, d_{T'})$ denote the set of nodes at level $L$ of such a decomposition. Let $P = \{ (d_{S'}, d_{T'}) : d_{S'}, d_{T'} \in \{0, \ldots, K\}, V'_L(d_{S'}, d_{T'}) \neq \emptyset \}$ denote the set of those $(d_{S'}, d_{T'})$  pairs for which the topmost level of the decomposition is non-empty. If $P = \emptyset$, then define $\gamma = 0$. Else define $\gamma = \max_{(d_{S'}, d_{T'}) \in P} \left\{ d_{S'} \cdot d_{T'} \right\}$. Then we have:
$$2 \cdot \sqrt{(1+\epsilon)} \cdot \rho(G) \geq \sqrt{\gamma} \geq \frac{\rho(G)}{2 \cdot \alpha \cdot (1+\epsilon)}.$$
\end{corollary}

\begin{proof}
If the derived graph is empty, i.e., if $E' = \emptyset$, then we have $E = \emptyset$, $P = \emptyset$, $\gamma = 0$, and $\rho(G) = 0$. Thus, in this case the corollary is trivially true. 

For the remainder of the proof,  we assume that $E' \neq \emptyset$. Here, it is easy to see that $P \neq \emptyset$: simply consider the $(\alpha, d_{S'}, d_{T'}, L)$-decomposition with $d_{S'} = d_{T'} = q_0 < 1$. In this decomposition, every node in $G'$ with nonzero degree will be promoted to the topmost level $V'_L(d_{S'}, d_{T'})$. 

Accordingly, for the rest of the proof, we fix a $(d_{S'}, d_{T'})$-pair in $P$ for which $d_{S'} \cdot d_{T'} = \gamma$. Let this pair be identified as $(d^*_{S'}, d^*_{T'})$. Since $(d^*_{S'}, d^*_{T'}) \in P$, we infer that $V'_L(d^*_{S'}, d^*_{T'}) \neq \emptyset$. Hence, by the first part of Theorem~\ref{main:thm:derived}, we must have:
\begin{equation}
\label{eq:derived:old:1}
\gamma = d^*_{S'} \cdot d^*_{T'} \leq 4(1+\epsilon) \cdot \lambda_{S'} \cdot \lambda_{T'}.
\end{equation} 
Next, define $\hat{d}_{S'}$ (resp. $\hat{d}_{T'}$) to be the maximum  value of $d_{S'}$ (resp. $d_{T'}$) which is less than the threshold $\lambda_{S'}/\alpha$ (resp. $\lambda_{T'}/\alpha$). Note that since $E' \neq \emptyset$, we have $\lambda^* \leq \lambda_{S'}, \lambda_{T'} < n^2$. Thus, we are guaranteed the existence of such a pair $(\hat{d}_{S'}, \hat{d}_{T'})$. 
\begin{equation}
\label{eq:derived:old:2}
\hat{d}_{S'} = \max_{k \in \{0, \ldots, K\} \,  : \, q_k < \lambda_{S'}/\alpha} \{q_k\} \text{ and } \hat{d}_{T'} = \max_{k \in \{0, \ldots, K\} \,  : \, q_k < \lambda_{T'}/\alpha} \{q_k\}  
\end{equation} 
Next, note that $q_0 < \lambda^*/\alpha$ and $q_{K} \geq n^2$.  Since the consecutive $q_k$ values are within a factor of $(1+\epsilon)$ from each other, equation~\ref{eq:derived:old:2} gives us:
\begin{equation}
\label{eq:derived:old:3}
\hat{d}_{S'} \cdot \hat{d}_{T'} \geq \frac{\lambda_{S'} \cdot \lambda_{T'}}{\alpha^2 \cdot (1+\epsilon)^2}
\end{equation}
Since $\hat{d}_{S'} < \lambda_{S'}/\alpha$ and $\hat{d}_{T'} < \lambda_{T'}/\alpha$, Theorem~\ref{main:thm:derived} (part II) implies that $V'_L(\hat{d}_{S'}, \hat{d}_{T'}) \neq \emptyset$. Thus, we have $(\hat{d}_{S'}, \hat{d}_{T'}) \in P$. Since $(d^*_{S'}, d^*_{T'}) \in P$  maximizes the product of its two components, we get:
\begin{equation}
\label{eq:derived:old:4}
\gamma = d^*_{S'} \cdot d^*_{T'} \geq \hat{d}_{S'} \cdot \hat{d}_{T'}
\end{equation}
From equations~\ref{eq:derived:old:1},~\ref{eq:derived:old:3} and~\ref{eq:derived:old:4}, we infer that:
\begin{equation}
\label{eq:derived:old:5}
\frac{\lambda_{S'} \cdot \lambda_{T'}}{\alpha^2 \cdot (1+\epsilon)^2} \leq \gamma \leq 4 (1+\epsilon) \cdot \lambda_{S'} \cdot \lambda_{T'}
\end{equation}
By the second part of Lemma~\ref{lm:derived:main}, we have:  $\rho(G)/2 \leq   \sqrt{\lambda_{S'} \lambda_{T'}} \leq \rho(G)$. Combining this observation with equation~\ref{eq:derived:old:5}, we get:
\begin{equation}
\label{eq:derived:old:6}
\frac{\rho(G)}{2\alpha (1+\epsilon)} \leq \sqrt{\gamma} \leq 2 \cdot \sqrt{(1+\epsilon)} \cdot \rho(G)
\end{equation}
This concludes the proof of the corollary. 
\end{proof}

\subsection{The Algorithm for Maintaining an $(\alpha, d_{S'}, d_{T'}, L)$-Decomposition}
\label{sec:derived:maintain}

Throughout this section, we fix the values of $d_{S'}, d_{T'}, L$. Furthermore, we fix an $\alpha \geq 2+\epsilon$. We describe an algorithm for  maintaining an $(\alpha, d_{S'}, d_{T'}, L)$-decomposition of the derived graph $G' = (V', E')$  in a dynamic setting (see Definition~\ref{def:derived:graph}). We assume that the input graph $G = (V, E)$ is empty in the beginning, and hence, at that instant we also have $E' = \emptyset$. Subsequently, at each time-step, a directed edge $(u,v)$ is inserted into (resp. deleted from) the graph $G = (V, E)$, and accordingly, the edge $(s_u, t_v)$ is inserted into (resp. deleted from) the derived graph $G' = (V', E')$. Our main result is stated below.

\begin{theorem}
\label{th:derived:maintain}
For every polynomially bounded $\alpha \ge 2 + 3\eps$, we can deterministically  maintain an $(\alpha, d_{S'}, d_{T'}, L)$-decomposition  of the derived graph $G' = (V',E')$. Starting from an empty graph, we can  handle a sequence of $t$ update operations (edge insertions/deletions) in total time $O(t L/\eps)$. Thus, we get an amortized update time of $O(L/\eps)$. The space complexity of the data structure at a given time-step is $O(n +m)$, where $m = |E'|$ denotes the number of edges in the derived graph at that time-step, and $n = |V'|$ denotes the number of nodes in the derived graph (which does not change over time). Note that $|V| = 2n$ and $|E| = m$, where $G = (V, E)$ is the input graph. 
\end{theorem}

The proof of Theorem~\ref{th:derived:maintain} is very similar to the proof of Theorem~\ref{th:sec:dynpart:main} from Section~\ref{sec:dynpart}. Nevertheless, for the sake of completeness, we highlight the main parts of the algorithm and its analysis.

\paragraph{Data Structures.}  Recall the concept of the level of a node from Definition~\ref{main:def:partition:derived}. We now separately describe the data structures associated with the nodes in $S'$ and $T'$.
\begin{itemize}
\item Every node $s \in S'$ maintains $L$  lists $\text{{\sc Friends}}_i[s]$, for $i \in \{1, \ldots, L\}$. For $i < \ell(s)$, the list $\text{{\sc Friends}}_i[s]$ consists of the neighbors of $s$  that are at level $i$: these are nodes belonging to the set $\NN_s(V'_i \cap T_i)$. For $i = \ell(s)$, the set $\text{{\sc Friends}}_i[s]$ consists of the neighbors of $s$ that are at level $i$ or above: these are the nodes belonging to  the set $\NN_s(T_i)$. For $i > \ell(s)$, the list $\text{{\sc Friends}}_i[s]$ is empty.
Each list is stored in a doubly linked list together with its size, $\text{{\sc Count}}_i[s]$. Using appropriate pointers, we can insert or delete a given node to or from a concerned list in constant time. The counter $\text{{\sc Level}}[s]$ keeps track of the level of the node $s$.
\item Analogously, every node $t \in T'$ maintains $L$  lists $\text{{\sc Friends}}_i[t]$, for $i \in \{1, \ldots, L\}$. For $i < \ell(t)$, the list $\text{{\sc Friends}}_i[t]$ consists of the neighbors of $t$  that are at level $i$: these are nodes belonging to the set $\NN_t(V'_i \cap S_i)$. For $i = \ell(t)$, the set $\text{{\sc Friends}}_i[t]$ consists of the neighbors of $t$ that are at level $i$ or above: these are the nodes belonging to  the set $\NN_t(S_i)$. For $i > \ell(t)$, the list $\text{{\sc Friends}}_i[t]$ is empty. Each list is stored in a doubly linked list together with its size, $\text{{\sc Count}}_i[t]$. Using appropriate pointers, we can insert or delete a given node to or from a concerned list in constant time.  The counter $\text{{\sc Level}}[t]$ keeps track of the level of the node $t$.
\end{itemize}

\paragraph{The Algorithm.}
If a node   violates one of the conditions
of an $(\alpha, d_{S'}, d_{T'}, L)$-decomposition (see Definition~\ref{main:def:partition:derived}), then we call the node ``dirty'', else the node is called ``clean''. Specifically,  consider two possible cases depending on the type of the node.
\begin{itemize}
\item A node $s \in S'$ at level $\ell(s) = i$ is dirty iff either (a) $i < L$ and $\dd_s(T_i) > \alpha \cdot d_{S'}$, or (b) $i > 1$ and $\dd_s(T_{i-1}) < d_{S'}$.
\item A node $t \in T'$ at level $\ell(t) = i$ is dirty iff either (a) $i < L$ and $\dd_t(S_i) > \alpha \cdot d_{T'}$, or (b) $i > 1$ and $\dd_t(S_{i-1}) < d_{T'}$.
\end{itemize}
 Initially, the derived graph $G' = (V', E')$ is empty,  every node   is at level $1$, and  every node is clean. 
When an edge $(s, t)$, $s \in S', t \in T'$, is inserted/deleted in the derived graph $G'$,  we first update the $\text{{\sc Friends}}$ lists of $s$ and
$t$ by adding or removing that edge in constant time. Next we check whether $s$ or $t$ becomes dirty due to this edge insertion/deletion.
If yes, then we run the RECOVER-DIRECTED() procedure described in Figure~\ref{fig:dirty:main:derived}. Note that a single iteration of the {\sc While} loop (Steps 01-15) may change the status of some more nodes from clean to dirty (or vice versa). If and when the procedure terminates, however, every node is clean by definition.

\begin{figure}[htbp]
\centerline{\framebox{
\begin{minipage}{5.5in}
\begin{tabbing}
01. \ \ \ \    \=  {\sc While} there exists a dirty node  $y \in V' = S' \cup T'$: \\
02. \> \ \ \ \ \ \ \ \ \ \ \= {\sc If} $y \in S'$, {\sc Then} \\
03. \>  \>  \ \ \ \ \ \ \ \= {\sc If} $\dd_y(T_{\ell(y)}) > \alpha \cdot d_{S'}$ and $\ell(y) < L$, {\sc Then} \\
04.  \> \> \> \ \ \ \ \ \ \ \ \ \= Increment the level of $s$ by setting $\ell(y) \leftarrow \ell(y)+1$. \\
05. \> \> \> \> Update the relevant data structures to reflect this change. \\
06.  \> \> \> {\sc Else if} $\dd_y(T_{\ell(y)-1}) <  d_{S'}$ and $\ell(y) > 1$, {\sc Then} \\
07.  \> \> \> \> Decrement the level of $y$ by setting $\ell(y) \leftarrow \ell(y)-1$. \\
08. \> \> \> \> Update the relevant data structures to reflect this change. \\ 
09. \> \ \ \ \ \ \ \ \ \ \ \= {\sc Else if} $y \in T'$, {\sc Then} \\
10. \>  \>  \ \ \ \ \ \ \ \= {\sc If} $\dd_y(S_{\ell(y)}) > \alpha \cdot d_{T'}$ and $\ell(y) < L$, {\sc Then} \\
11.  \> \> \> \ \ \ \ \ \ \ \ \ \= Increment the level of $s$ by setting $\ell(y) \leftarrow \ell(y)+1$. \\
12. \> \> \> \> Update the relevant data structures to reflect this change. \\
13.  \> \> \> {\sc Else if} $\dd_y(S_{\ell(y)-1}) <  d_{T'}$ and $\ell(y) > 1$, {\sc Then} \\
14.  \> \> \> \> Decrement the level of $y$ by setting $\ell(y) \leftarrow \ell(y)-1$. \\
15. \> \> \> \> Update the relevant data structures to reflect this change.
\end{tabbing}
\end{minipage}
}}
\caption{\label{fig:dirty:main:derived} RECOVER-DIRECTED(). }
\end{figure}

\paragraph{Space complexity.} Since each edge in $G' = (V', E')$  appears in two linked lists (corresponding to each of its endpoints), the space complexity of the data structure is $O(n+ m)$.

\paragraph{Analysis of the Update Time.}
Handling each edge insertion/deletion takes constant time plus the time for the RECOVER-DIRECTED() procedure. We show below that the total time spent in procedure 
RECOVER-DIRECTED() during $t$ update operations is $O(t L/\eps)$.

\paragraph{Potential Function.} To determine the amortized update time we use a potential function $\B$ that depends on the state of the $(\alpha, d_{S'}, d_{T'}, L)$-decomposition. For any two nodes $s \in S', t \in T'$, let $f(s,t) = 1$ if $l(s) = l(t)$ and  0 otherwise. We define $\B$, the node  potentials $\Phi(x)$ (for each $x \in V' = S' \cup T'$), and the edge potentials $\Psi(s,t)$ (for each $(s, t) \in E'$) as follows. 

\begin{eqnarray}
\label{eq:potential:main}
\B & = & \sum_{x \in S' \cup T'} \Phi(x) + \sum_{(s, t) \in E'} \Psi(s,t) \\
\Phi(s) & = & \left({1 \over \eps}\right) \cdot \sum_{i = 1}^{\ell(s)-1} \max(0, \alpha \cdot d_{S'} - \dd_s(T_i)) \ \ \text{ for all nodes } s \in S' \label{eq:potential:node:derived:1} \\
\Phi(t) & = & \left({1 \over \eps}\right) \cdot \sum_{i = 1}^{\ell(t)-1} \max(0, \alpha \cdot d_{T'} - \dd_t(S_i)) \ \ \text{ for all nodes } t \in T' \label{eq:potential:node:derived:2} \\
\Psi(s, t) & = & 2 \cdot (L -  \min(\ell(s), \ell(t))) + f(s,t)\ \ \text{ for all edges } (s,t) \in E' \label{eq:potential:edge:derived}
\end{eqnarray}
\noindent It is easy to check that all these potentials are nonnegative, and that they are uniquely defined by the $(\alpha, d_{S'}, d_{L'}, L)$-decomposition under consideration. Now, mimicking  the potential function based analysis from Section~\ref{sec:dynpart}, we can infer the following facts.
\begin{itemize}
\item (F1) In the beginning, when the derived graph $G' = (V', E')$ is empty, we have $\B = 0$. Subsequently, the potential $\B$ remains always nonnegative. 
\item (F2) Insertion/deletion of an edge in the derived graph $G' = (V', E')$ increases the potential $\B$ by at most $3L/\epsilon$. 
\item  To analyze the amortized running time of the RECOVER-DIRECTED() procedure, we have the following claims.
\begin{itemize}
\item (F3) Consider a single iteration of the {\sc While} loop in Figure~\ref{fig:dirty:main:derived} where a node $s \in S'$ with $\ell(s) = i$ changes (increments or decrements) its level by one. This  takes $O(1+ \dd_s(T_i))$ time. On the other hand, the net drop in the overall potential $\B$ due to the same iteration of the {\sc While} loop is  $\Omega(1+ \dd_s(T_i))$, provided $\alpha \geq 2+3\epsilon$. 
\item (F4) Consider a single iteration of the {\sc While} loop in Figure~\ref{fig:dirty:main:derived} where a node $t \in T'$ with $\ell(t) = i$ changes (increments or decrements) its level by one. This  takes $O(1+ \dd_t(S_i))$ time. On the other hand, the net drop in the overall potential $\B$ due to the same iteration of the {\sc While} loop is $\Omega(1+ \dd_t(S_i))$, provided $\alpha \geq 2+3\epsilon$. 
\end{itemize}
\end{itemize}
\noindent Facts (F1) -- (F4) imply that the RECOVER-DIRECTED() procedure takes a total time of $O(t L/\epsilon)$ to handling the first $t$ edge insertions/deletions in the derived graph $G' = (V', E')$. This gives an amortized update time of $O(L/\epsilon)$ for our algorithm and concludes the proof of Theorem~\ref{th:derived:maintain}.

\subsection{Wrapping Up: Proof of Theorem~\ref{th:directed:main}}
\label{sec:directed:conclude}

We fix a sufficiently small constant $\epsilon \in (0, 1)$ and set $\alpha = 2+3\epsilon$, $L = 2 \cdot (2 + \lceil \log_{(1+\epsilon)} n \rceil)$, and $\lambda^* = 1 -1/\sqrt{n}$. Next, as in Corollary~\ref{cor:derived}, we discretize the range $[\lambda^*/(\alpha (1+\epsilon)), n^2]$ by setting $q_k = (1+\epsilon)^{k-1} \cdot \lambda^*/\alpha$ for every integer $k \geq 0$. We then define $K$ to be the smallest integer $k$ for which $q_k \geq n^2$. Next, we maintain an $(\alpha, d_{S'}, d_{T'}, L)$-decomposition of the derived graph $G' = (V', E')$ for every $d_{S'}, d_{T'} \in \{d_0, \ldots, d_K\}$. By Theorem~\ref{th:derived:maintain}, maintaining each of these decompositions requires $O(m+n)$ space and $O(L/\epsilon)$ amortized update time. Hence, the total space requirement of our scheme is $O(K^2 (m+n)) = \tilde{O}(m+n)$ and the total amortized update time is $O(K^2 L/\epsilon) = \tilde{O}(1)$. We also maintain the value of $\gamma$ as defined in Corollary~\ref{cor:derived}. Since there are $O(K^2)$ decompositions, maintaining the value of $\gamma$ also requires $O(K^2) = \tilde{O}(1)$ update time. By Corollary~\ref{cor:derived}, the quantity $\sqrt{\gamma}/(2 \sqrt{1+\epsilon})$ gives a $4 \alpha \cdot (1+\epsilon)^{3/2} = 8 \cdot (1 + O(\epsilon))$-approximation to the value of the densest subgraph $\rho(G)$. This concludes the proof of Theorem~\ref{th:directed:main}.

\section{Sublinear-Time Algorithm}\label{sec:sublinear}

In this section, we focus on sublinear time algorithms for the approximate densest subgraph problem. Our main results are summarized in Theorems~\ref{th:sublinear:upper} and~\ref{th:sublinear:lower}.

If we assume that an algorithm has to read {\em all} of its input, then no sublinear  (in the input size) time algorithm is possible.
However, if we assume that the input is given by an {\em oracle} that gives efficient access to the input, then sublinear time algorithms
might exist. We present in the following such an oracle that allows us to turn our algorithm from Section~\ref{sec:sketch} into 
a sublinear time algorithm. Specifically, we will give an 
$\tilde O(n)$ time algorithm that requires $\tilde O(n)$ oracle queries and space. Afterwards we will also show that with this oracle no further
assymptotic improvement is possible.

\paragraph{Oracle model.} We first present the oracle model for the input graph.
It is a standard representation that is, e.g., assumed in the sublinear time algorithms of \cite{ChazelleRT05,GoelKK13} and is
called {\em incident-list model}.  In this representation, we allow two types of accesses to the input graph (called {\em oracle queries}): (1) the {\em degree query} which asks for the degree of some node $v$, and (2) the {\em neighbor query} which asks for the $i^{th}$ neighbor of $v$ (i.e. the $i^{th}$ element in the incidence list corresponding to the neighbors of $v$). See, e.g., \cite{Onak-survey10,CzumajS-survey10,Goldreich11q-intrographproperties,Goldreich11p-briefintro,RubinfeldS11,Ron-fttcs09-survey} for further surveys. 



\paragraph{Upper Bound.} In Section~\ref{sec:sketch}, we showed how to compute a $(2+\epsilon)$-approximate solution to the densest subgraph problem using only $\tilde O(n)$ edges sampled uniformly at random. In the above oracle model, sampling an edge can be done using one neighbor query. Thus, the algorithm needs only $\tilde O(n)$ queries. After the sampling is completed we can process the collection of sampled edges using $\tilde O(n)$ time and space, as in the proof of Theorem~\ref{main:th:stream:main}, simply by computing the $(1+\epsilon, d, \tilde O(1))$-decomposition for different $\tilde O(1)$ values of $d$, to get the desired $(2+\epsilon)$-approximate solution.  This leads to the following theorem.

\begin{theorem}
\label{th:sublinear:upper}
There is a sub-linear time algorithm for computing a $(2+\epsilon)$-approximate solution to the densest subgraph problem in the incidence-list model. The algorithm makes $\O(n)$ oracle queries, and requires $\O(n)$ time and $\O(n)$ space. 
\end{theorem}

\paragraph{Lower Bound.} We adapt the proof of \cite[Lemma 7]{BahmaniKV12} to show that for any $\lambda \geq 3/2$, a $\lambda$-approximation algorithm needs to make $\Omega(n/(\lambda^2\poly\log(n)))$ oracle queries. 
%
%
%
Consider the following communication complexity problem ``P1'':
\begin{itemize}
\item There are $k\geq 2$ players, denoted by $p_1, \ldots, p_k$  and an $n$-node input graph $G$ consisting of $\ell$ disjoint subgraphs, denoted by $G_1, \ldots, G_\ell$. Each $G_i$ has $k$ nodes, denoted by $\{u_{i, 1}, \ldots, u_{i, k}\}$ (thus $n=k\ell$). Further each subgraph is either a star or a clique. For any node $u_{i, j}$ in $G_i$, if its degree is more than one then player $p_j$ knows about all edges incident to $u_{i, j}$. In other words, $p_j$ knows about edges incident to nodes with degree more than one among $u_{1, j}, u_{2, j}, \ldots, u_{\ell, j}$. The players want to distinguish between the case where there is a clique (thus the densest subgraph has density at least $(k-1)/2$) and when there is no clique (thus the densest subgraph has density at most $1$). Their communication protocol is in the blackboard model, where in each round a player can write a message on the backboard, which will be seen by all other players, and the communcation complexity is the number of bits written to the board.
Using a reduction from the multi-party set disjointness problem, the papers \cite{ChakrabartiKS03,BahmaniKV12} showed that this problem require $\tilde \Omega(\ell/k) = \tilde \Omega(n/k^2)$ communication bits. 
\end{itemize}

\begin{lemma}
If there is a sublinear-time algorithm with $q$ oracle queries for the problem P1 defined above, then the problem P1 can also be solved using $\tilde O(q)$ communication bits. 
\end{lemma}
\begin{proof}
Let $\cal A$ be such algorithm. Player $p_1$ simulates $\cal A$ by answering each query of $\cal A$ using $\tilde O(1)$ communication bits, as follows. If $\cal A$ makes a degree query on node $u_{ij}$, player $p_1$ will ask for an answer from player $p_j$: either $p_j$ knows all edges incident to $u_{ij}$ (in which case the degree of $u_{ij}$ is $k$) or the degree of $u_{ij}$ is one. If $\cal A$ asks for the $t^{th}$ neighbor of node $u_{ij}$, player $p_1$ asks for this from player $p_j$. If player $p_j$ does not know the answer, then we know that the degree of $u_{ij}$ is one and $G_i$ is a star. In this case, player $p_1$ writes on a blackboard asking for the unique node $u_{ij'}$ in $G_i$ whose degree is more than one. Then, the only edge incident to $u_{ij}$ is $u_{ij}u_{ij'}$. This edge can be used to answer the query. 
\end{proof}
Note that any $((k-1)/2-\epsilon)$-approximation algorithm for the densest subgraph problem solves problem P1. Thus, 
the above lemma implies that any $((k-1)/2-\epsilon)$-approximation algorithm requires $\tilde \Omega(n/k)$ queries. By considering any $k\geq 4$, we get the following theorem. 

\begin{theorem}
\label{th:sublinear:lower}
In the incidence-list model, for any $\lambda\geq 3/2$ and any $\epsilon>0$, any $\lambda-\epsilon$-approximation algorithm for the densest subgraph problem needs to make $\tilde \Omega(n/\lambda^2)$ queries.
\end{theorem}

\section{Distributed Streams}\label{sec:distributed}

In the distributed streaming model (see, e.g., \cite{CormodeMYZ10}), there are $k$ {\em sites} receiving different sequences of edge insertions (without any deletion), and these sites must coordinate with the {\em coordinator}. The objective is to minimize the communication between the sites and the coordinator in order to maintain the densest subgraph. We sample $\tilde O(n)$ edges (without replacement)\danupon{@Sayan: Or do we need a replacement?} as a sketch by using the sampling algorithm of Cormode~et~al.~\cite{CormodeMYZ10}: their algorithm can sample $\tilde O(n)$ edges using $\tilde O(k+n)$ bits of communication, whereas the coordinator needs $\tilde O(n)$ space and each site needs $\O(1)$ space. The coordinator can then use this sketch to compute a $(2+\epsilon)$-approximate solution. 

\begin{theorem}
\label{th:distributed:stream}
In the distributed streaming setting with $k$ sites~\cite{CormodeMYZ10}, we can compute a $(2+\epsilon)$-approximate solution to the densest subgraph problem using $\O(k+n)$ bits of communication. The coordinator needs $\O(n)$ space and each site needs $\O(1)$ space.
\end{theorem}

\section{Open problems}\label{sec:open}

An obvious question is whether the $(4+\epsilon)$ approximation ratio provided by our algorithm is tight. In particular, it will be interesting if one can improve the approximation ratio to $(2+\epsilon)$ to match the case where an update time is not a concern. 
Getting this approximation ratio even with larger space complexity is still interesting. (Epasto~et~al.~\cite{EpastoLS15} almost achieved this except that they have to assume that the deletions happen uniformly at random.)
It is equally interesting to show a hardness result. Currently, there is only a hardness result for maintaining the optimal solution \cite{HenzingerKNS15}. It will be interesting to show a hardness result for approximation algorithms. 
Another interesting question is whether a similar result to ours can be achieved with polylogarithmic {\em worst-case} update time. 
Finally, a more general question is whether one can obtain space- and time-efficient fully-dynamic algorithm like ours for other fundamental graph problems, e.g. maximum matching and single-source shortest paths.

\shortOnly{\bibliographystyle{plain}}
\longOnly{\bibliographystyle{alpha}}
\bibliography{reference}



\newpage
\appendix

\end{document}